\documentclass[aps,pra,twocolumn,superscriptaddress,showpacs,10pt]{revtex4-2}
\usepackage[T1]{fontenc}
\usepackage{amssymb}
\usepackage{amsthm}
\usepackage{dsfont}
\usepackage[shortlabels]{enumitem}
\usepackage{physics}
\usepackage{xcolor}
\definecolor{urlcol}{RGB}{0, 102, 204}
\usepackage{graphicx}
\usepackage[all]{nowidow}
\usepackage{thmtools}

\usepackage[pdfusetitle, hypertexnames=false]{hyperref}
\hypersetup{
  colorlinks = true,
  urlcolor = urlcol,
  linkcolor = blue,
  citecolor = blue,
  breaklinks,
}
\usepackage[capitalize]{cleveref}

\usepackage{mathtools}

\usepackage{tikz}

\newcommand{\bs}[1]{\boldsymbol{#1}}
\newcommand{\scomm}[2]{[\![#1,#2]\!]}
\DeclareMathOperator{\ad}{ad}
\DeclareMathOperator{\Lie}{Lie}
\renewcommand{\epsilon}{\varepsilon}
\newcommand{\bsj}{\boldsymbol{j}}
\newcommand{\bsjs}{\boldsymbol{j}^*}
\newcommand{\bsk}{\boldsymbol{k}}
\newcommand{\claws}{\text{claws}}
\newcommand{\subv}[1]{\mbf{b}_{#1}}
\newcommand{\claw}{-\hspace{-2pt}\langle}
\newcommand{\dgrU}[2]{\delta_{#1}(#2)}
\newcommand{\mc}[1]{\mathcal{#1}}
\newcommand{\mbb}[1]{\mathbb{#1}}
\newcommand{\md}{\text{-}}
\newcommand{\kryl}[1]{\eta^{(#1)}}
\newcommand{\bil}[2]{\kryl{#1}\kryl{#2}}
\newcommand{\bih}[1]{\kappa^{(#1)}}
\newcommand{\cur}[1]{J^{(#1)}}
\newcommand{\liek}{\mathfrak{K}}
\newcommand{\lieso}[1]{\mathfrak{so}(#1)}
\newcommand{\mbf}[1]{\mathbf{#1}}
\NewDocumentCommand\mm{g}{\gamma\IfValueTF{#1}{_{#1}}{}}
\NewDocumentCommand\emm{g}{\varphi\IfValueTF{#1}{_{#1}}{}}
\newcommand{\mrm}[1]{\mathrm{#1}}
\newcommand{\mh}[1]{\mbox{\hspace{#1}}}

\renewcommand{\sd}{\oplus}
\newtheorem{theorem}{Theorem}
\newtheorem{lemma}{Lemma}
\newtheorem{corollary}{Corollary}
\newtheorem{claim}{Claim}
\newtheorem{conjecture}{Conjecture}
\newtheorem{proposition}{Proposition}
\newtheorem{definition}{Definition}
\newtheorem{observation}{Observation}
\newtheorem{remark}{Remark}

\crefname{conjecture}{Conjecture}{Conjectures}
\Crefname{conjecture}{Conjecture}{Conjectures}
\crefname{claim}{Claim}{Claims}
\Crefname{claim}{Claim}{Claims}
\crefname{observation}{Observation}{Observations}
\Crefname{observation}{Observation}{Observations}

\makeatletter
\AfterEndEnvironment{proof}{\@doendpe}
\AfterEndEnvironment{definition}{\@doendpe}
\AfterEndEnvironment{lemma}{\@doendpe}
\AfterEndEnvironment{proposition}{\@doendpe}
\AfterEndEnvironment{corollary}{\@doendpe}
\AfterEndEnvironment{theorem}{\@doendpe}
\AfterEndEnvironment{remark}{\@doendpe}
\AfterEndEnvironment{conjecture}{\@doendpe}
\AfterEndEnvironment{observation}{\@doendpe}
\makeatother
\newcommand{\id}{\mathds{1}}

\DeclarePairedDelimiter\Parentheses{\lparen}{\rparen}
\DeclarePairedDelimiter\Brackets{\lbrack}{\rbrack}
\DeclarePairedDelimiter\Braces{\lbrace}{\rbrace}
\DeclarePairedDelimiter\BraKet{\langle}{\rangle}

\DeclarePairedDelimiterX\SetVerticalLine[2]{\lbrace}{\rbrace}{#1\nonscript\:\delimsize\vert\allowbreak\nonscript\:\mathopen{}#2}
\DeclarePairedDelimiterX\GroupSetVerticalLine[2]{\langle}{\rangle}{#1\nonscript\:\delimsize\vert\allowbreak\nonscript\:\mathopen{}#2}

\NewDocumentCommand\bk{sm}{\IfBooleanTF{#1}{\Brackets{#2}}{\Brackets*{#2}}}

\let\Set\SetVerticalLine \let\GroupSet\GroupSetVerticalLine
\newcommand\set[2]{\Set*{#1}{#2}}
\newcommand\groupset[2]{\GroupSet*{#1}{#2}}

\NewDocumentCommand\e{sd()d[]gd<>}{\IfValueTF{#2}{\IfBooleanTF{#1}{\Parentheses{#2}}{\Parentheses*{#2}}}{}\IfValueTF{#3}{\IfBooleanTF{#1}{\Brackets{#3}}{\Brackets*{#3}}}{}\IfValueTF{#4}{\IfBooleanTF{#1}{\Braces{#4}}{\Braces*{#4}}}{}\IfValueTF{#5}{\IfBooleanTF{#1}{\BraKet{#5}}{\BraKet*{#5}}}{}}

\NewDocumentCommand\ee{g}{\mathrm{e}\IfValueTF{#1}{^{#1}}{}}

\newcommand{\bvs}[1]{b^2[#1]} \newcommand{\cv}[1]{C[#1]} \newcommand{\hv}[1]{h[#1]} \newcommand{\hvs}[1]{h^2[#1]}  
\newcommand{\kv}[1]{K[#1]} \newcommand{\Pv}[1]{P(#1)} \newcommand{\Qv}[1]{Q(#1)} \newcommand{\vv}[1]{\boldsymbol{v}[#1]}

\begin{document}

\title{A graph-theoretic analysis of non-generic free-fermion solvability by Krylov decompositions}

\author{Jannis Ruh}
\affiliation{Centre for Quantum Software and Information, School of Computer Science, Faculty of Engineering \& Information Technology, University of Technology Sydney, NSW 2007, Australia}

\author{Samuel J. Elman}
\affiliation{Centre for Quantum Software and Information, School of Computer Science, Faculty of Engineering \& Information Technology, University of Technology Sydney, NSW 2007, Australia}

\author{Adrian Chapman}
\affiliation{Phasecraft Ltd.}

\date{October 2, 2026}

\begin{abstract}
    Exactly solvable models are often identified from structural features of their interactions.
    For quantum spin Hamiltonians, these are captured by the \emph{frustration graph} of anticommuting Pauli terms. 
    The frustration graph characterizes broad families of spin models that map to free fermions for arbitrary, or \emph{generic}, values of their coupling strengths.
    Nevertheless, there exist \emph{non-generically solvable} models that admit free-fermion solutions only at finely tuned couplings.
    In this work, we characterize a family of such models in terms of graph structures.
    Given an arbitrary instance of Fendley’s free-fermions-in-disguise (FFD) chain, we construct a set of operators, called \emph{Krylov currents}, that are exactly quadratic in its effective fermionic modes and thus extend it to a linear family of free-fermion solvable Hamiltonians.
    As the Krylov currents are nonlinear in the FFD Hamiltonian terms, the elements of this generated family are non-generically solvable, and they accordingly fail the known conditions for generic solvability.  
    We prove that every Krylov current admits an expansion in terms of induced paths of the frustration graph and that they generate the full special-orthogonal algebra of Gaussian transformations.
    Our results identify path support as a signature of non-generic solvability, and they provide a systematic route to finding exactly solvable many-body Hamiltonians.
\end{abstract}
\maketitle

\section{Introduction}
\label{sec:intro}

Exact analytic solutions to quantum many-body systems are invaluable for the mathematical and physical insights they provide.
For one, they present starting points for many well-known approximation schemes, which are used to model realistic systems.
Beyond this immediate utility, these methods frequently reveal surprising connections between previously unrelated mathematical structures, suggesting the existence of deeply unifying principles.

Perhaps the most enlightening method, in both of these regards, consists of mappings between spin and free-fermion systems. 
These mappings allow one to formulate certain calculations in terms of the reduced dynamics of a single particle, achieving an exponential reduction in complexity.
One of the earliest examples of this technique was demonstrated by the solution of the 1d XY model via the Jordan-Wigner transformation~\cite{jordan1928uber, Lieb1961}.
The Jordan-Wigner mapping has since been generalised to higher-dimensional lattices \cite{wang1991ground, huerta1993bose-fermi, kitaev2006anyons}, more general geometries \cite{backens2019jordan-wigner}, and systems of qudits and parafermions \cite{fradkin1989jordan-wigner, batista2001generalized, wang2025particle}. 
These solutions have also found applications in quantum computing through fermionic quantum information processing \cite{knill2001fermionic, bravyi2002fermionic, Terhal2002, DiVincenzo2005, oszmaniec2022fermion}, topological quantum error correction \cite{kells2009description, nussinov2012arbitrary, tantivasadakarn2020jordan-wigner, chapman2022free}, and fermion-to-qubit mappings \cite{verstraete2005mapping, Chien2020Optimising, Derby2021Compact, harrison2024Serpinski}.

A \emph{free-fermion solution} to a quantum many-body spin model identifies collective spin degrees of freedom with fermionic modes such that the spin-model Hamiltonian is quadratic in the fermion picture.
A natural such mapping injectively identifies each term of the Hamiltonian with a quadratic fermionic operator, extended linearly to a free-fermion Hamiltonian.
Because the mapping is \emph{linear}, its existence is independent of the specific values of the nonzero interaction strengths of the terms. 
For any generic such values of these coupling coefficients, the solution exists. 
We call this family of solutions \emph{generic}.

While linearity is sufficient for a solution to be generic, it is, remarkably, not necessary.
In a seminal work \cite{fendley2019free}, Fendley introduced a spin model that admits a generic free-fermion solution but provably does not admit a linear solution. 
The solution maps the entire spin Hamiltonian holistically to a free-fermion one for any values of the coupling strengths.
Models admitting such a nonlinear solution are called ``free fermions in disguise'' models, and the spin model in Ref.~\cite{fendley2019free} has been correspondingly termed the FFD model.
Subsequent works have investigated the dynamics generated by this model \cite{fukai2025quantum, exact2025vona, szszschagrin2026construction} as well as its state-space structure \cite{Vernier2026Hilbert, fukai2026solvinga}.

The existence of a generic free-fermion solution depends only on the algebraic structure of the Hamiltonian.
When the Hamiltonian is written as a sum of Pauli operators, this structure is captured by the \emph{frustration graph}.
This is the graph whose vertices correspond to Pauli terms, with edges between vertices whose associated Paulis anticommute.
In Ref.~\cite{Chapman2020}, it was shown that an injective linear solution exists if and only if its frustration graph is a \emph{line graph}.
In subsequent works~\cite{elman2021free, chapman2023unified}, this was generalised to disguised free-fermion models: a model admits a nonlinear free-fermion solution if its frustration graph is \emph{(even-hole, claw)-free}~\cite{elman2021free}, and more generally if it is \emph{simplicial, claw-free}~\cite{chapman2023unified}.
The latter family includes the former two, providing a unified graph-theoretic characterisation of both linear and disguised free-fermion solutions.
This characterisation has been extended to parafermions~\cite{mann2025graph} and dissipative models~\cite{kohei2026dissipative}.

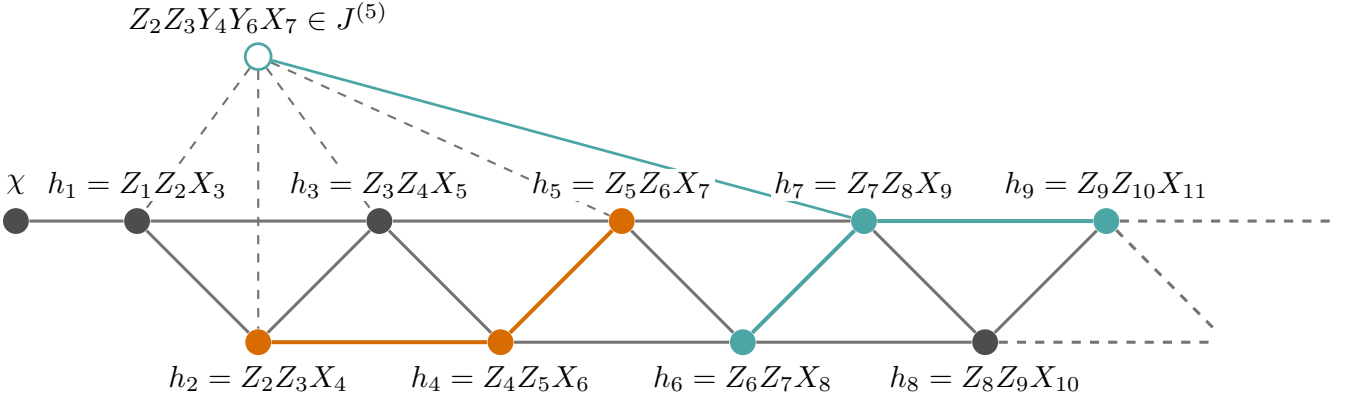
\begin{figure*}[t]
    \centering
    \resizebox{\textwidth}{!}{\begin{tikzpicture}[
        every node/.style={font=\normalsize},
        ]
        \pgfmathsetmacro{\sr}{-1.4}
        \tikzset{dot/.style={circle,fill,inner sep=3pt}}
        \tikzset{hollowdot/.style={circle,draw=teal!70,fill=white,line width=0.9pt,inner sep=3pt}}
        \tikzset{lbl/.style={fill=white,inner sep=1pt}}
        
        \node[dot,fill=black!70] (chi) at (-1.4,0) {};
        \node[dot,fill=black!70] (a) at (0,0) {};
        \node[dot,fill=orange!85!black] (b) at (1.4,\sr) {};
        \node[dot,fill=black!70] (c) at (2.8,0) {};
        \node[dot,fill=orange!85!black] (d) at (4.2,\sr) {};
        \node[dot,fill=orange!85!black] (e) at (5.6,0) {};
        \node[dot,fill=teal!70] (f) at (7.0,\sr) {};
        \node[dot,fill=teal!70] (g) at (8.4,0) {};
        \node[dot,fill=black!70] (h) at (9.8,\sr) {};
        \node[dot,fill=teal!70] (k) at (11.2,0) {};
        \node[hollowdot] (T) at (1.4,1.9) {};
        \node (p) at (12.6,\sr) {};
        \node (q) at (14.0,0) {};
        
        \draw[line width=1pt,black!55] (chi) -- (a);
        \draw[line width=1pt,black!55] (a) -- (c) -- (e) -- (g) -- (k);
        \draw[line width=1pt,black!55] (b) -- (d) -- (f) -- (h);
        \draw[line width=1pt,black!55] (a) -- (b);
        \draw[line width=1pt,black!55] (b) -- (c);
        \draw[line width=1pt,black!55] (c) -- (d);
        \draw[line width=1pt,black!55] (d) -- (e);
        \draw[line width=1pt,black!55] (e) -- (f);
        \draw[line width=1pt,black!55] (f) -- (g);
        \draw[line width=1pt,black!55] (g) -- (h);
        \draw[line width=1pt,black!55] (h) -- (k);
        \draw[line width=1pt,black!55,dashed] (h) -- (p);
        \draw[line width=1pt,black!55,dashed] (k) -- (p);
        \draw[line width=1pt,black!55,dashed] (k) -- (q);
        
        \draw[line width=0.7pt,black!55,dashed] (T) -- (a);
        \draw[line width=0.7pt,black!55,dashed] (T) -- (b);
        \draw[line width=0.7pt,black!55,dashed] (T) -- (c);
        \draw[line width=0.7pt,black!55,dashed] (T) -- (e);
        
        \draw[line width=1.3pt,orange!85!black] (b) -- (d);
        \draw[line width=1.3pt,orange!85!black] (d) -- (e);
        
        \draw[line width=1.3pt,teal!70] (f) -- (g);
        \draw[line width=1.3pt,teal!70] (g) -- (k);
        \draw[line width=0.9pt,teal!70] (g) -- (T);
        
        \node[lbl] at (-1.4,0.42) { \(\chi\)};
        \node[lbl] at (0,0.42) { \(h_1 = Z_1Z_2X_3\)};
        \node[lbl] at (1.4,{\sr-0.42}) { \(h_2 = Z_2Z_3X_4\)};
        \node[lbl] at (2.8,0.42) { \(h_3 = Z_3Z_4X_5\)};
        \node[lbl] at (4.2,{\sr-0.42}) { \(h_4 = Z_4Z_5X_6\)};
        \node[lbl] at (5.6,0.42) { \(h_5 = Z_5Z_6X_7\)};
        \node[lbl] at (7.0,{\sr-0.42}) { \(h_6 = Z_6Z_7X_8\)};
        \node[lbl] at (8.4,0.42) { \(h_7 = Z_7Z_8X_9\)};
        \node[lbl] at (9.8,{\sr-0.42}) { \(h_8 = Z_8Z_9X_{10}\)};
        \node[lbl] at (11.2,0.42) { \(h_9 = Z_9Z_{10}X_{11}\)};
        \node[lbl] at (1.4,2.35) { \(Z_2Z_3Y_4Y_6X_7 \in J^{(5)}\)};
    \end{tikzpicture}}
    \caption{The zigzag-ladder frustration graph of the FFD Hamiltonian in \cref{eq:ffd}, together with the simplicial mode  \(\chi\), which neighbors only  \(h_1\). 
    The term  \(Z_2Z_3Y_4Y_6X_7 \in J^{(5)}\), obtained as the product  \(h_2h_4h_5\) along the induced path highlighted in orange, is also shown. 
    This path is the symmetric difference of the anchored induced paths  \(\chi\)--\(h_1\) and  \(\chi\)--\(h_1\)--\(h_2\)--\(h_4\)--\(h_5\), of lengths  \(1\) and  \(4\) (summing to  \(5\)), so the new term occurs in the Krylov current  \(J^{(5)}\). 
    Including this term introduces a claw with center at  \(h_7\), shown in teal together with its pairwise non-adjacent neighbors.}
    \label{fig:zigzag_ladder}
\end{figure*}

It would therefore be surprising if graph theory could access any aspect of \emph{non-generic} free-fermion solvability, which is exhibited by models whose emergent free-fermions exist only at criticality.
Indeed, the frustration graph alone cannot determine whether a model has a non-generic solution without knowledge of the coupling strengths.
Several families of such models have been discovered~\cite{Fendley2024, fukai2025free} that explicitly violate the graph-theoretic conditions of Ref.~\cite{chapman2023unified}.
Nonetheless, the algebraic solutions of these models circumvent the obstructions that would prohibit a generic solution, generally through a high degree of cancellation.

In this work, we illuminate the relationship between graph theory and non-generic solvability.
While the frustration graph alone cannot tell us whether a model admits a non-generic solution, it can nevertheless signify the existence of critical, free-fermion solvable phases at particular values of the coupling strengths \footnote{Here, we emphasize that we are considering non-generic \emph{solutions}: solutions which apply only at specific coupling-strength configurations. 
We show that these exist for certain frustration graphs, but we do not rule out the existence of others for the same graph.}.
Our main theorem is a result in this direction. 
The strategy we take is to view the FFD model as a kind of ``seed model'' which generates all quadratic Hamiltonians in its fermionic eigenmodes, and these new models are generally non-generically solvable.
This is analogous to the construction by which the transverse-field Ising model can correspondingly be said to generate all linear, generic solutions (possibly over a symmetry-restricted subspace).
The relationships between these possible solution families are summarized schematically in the table.

As a main result, we characterize, in graph-theoretic terms, a certain subspace of quadratics in the fermionic eigenmodes of the FFD model.
Consequently, we identify a linear subspace of extensions---operators that we add to the Hamiltonian---to this model under which it remains free-fermion solvable, yet only for specific values of the coupling strengths in the Pauli basis.
We furthermore show that these additional terms generate the maximal  \(\lieso{d}\) Lie algebra of free-fermionic Hamiltonians as their dynamical Lie algebra \cite{wiersema2024classification, aguilar2024full, goh2025lie, kocku2026classification}.
In this regard, our results can be said to constitute a complete characterization of the free-fermion dynamical Lie algebra of this model in terms of graph theory.

\begin{table}[]
    \begin{tabular}{l|ll}
        Solution & Linear & Nonlinear \\ \hline
        Generic & \begin{tabular}[c]{@{}l@{}}Line graphs sufficient\\ and necessary\end{tabular} & \begin{tabular}[c]{@{}l@{}}SCF graphs sufficient \\ (unknown if necessary)\end{tabular} \\
        Non-Generic & Impossible & \textbf{{[}This work{]}}                                                               
    \end{tabular}
\end{table}

Our paper is structured as follows.
\Cref{sec:mainresults} gives our main results together with some relevant background so as to be self-contained.
\Cref{sec:numerical_analysis} provides numerical results regarding the new free-fermion-solvable models we identify, such as gap structure and the number of obstructions to a generic solution (i.e. claws) in these models.
In \cref{sec:conjectures}, we give a discussion on the implications of our result as well as some conjectured extensions beyond the setting we consider. 
In \cref{sec:background}, we provide more detail regarding the frustration graph formalism, the solution to the FFD model, and the Krylov-mode construction. 
In \cref{sec:discuss} we have a concluding discussion, followed by the appendices that provide the mathematical details.

\section{Results}
\label{sec:mainresults}
Consider the FFD model, a 1D spin-chain model on  \(n\) qubits described by the following 3-local Hamiltonian
\begin{align}
    H_{\mrm{FFD}}=\sum_{j=1}^{n - 2} b_j Z_j Z_{j + 1} X_{j + 2} \eqqcolon \sum_{j = 1}^{m} h_j
    \label{eq:ffd}
\end{align}
for  \(m = n - 2\).
The frustration graph of this model is called the \emph{zigzag-ladder graph}, shown in \cref{fig:zigzag_ladder}. 

Our results involve operators defined by \emph{Krylov decompositions} in terms of this Hamiltonian (\cref{def:krylovmode,eq:current_operators}).
Let  \(\mc{L}\) be the linear map on operators, acting as
\begin{align}
    \mc{L}[X] \coloneqq \frac{i}{2}[H, X]\;.
\end{align}
Fix a particular operator called the \emph{simplicial mode}  \(\chi\), which is defined via the free-fermion solution of  \(H_{\mathrm{FFD}}\).
One generates an operator Krylov subspace by repeated application of  \(\mc{L}\) to  \(\chi\).
This defines the \emph{Krylov modes}
\begin{align}
    \eta^{(k)} \coloneqq \mc{L}^k[\chi]\;,
\end{align}
that span a subspace of dimension  \(d\), where  \(d \simeq 2 \frac{n}{3}\) for the FFD model. 

It is proven in Ref.~\cite{chapman2023unified} that the Krylov modes for general simplicial, claw-free models satisfy a non-orthogonal version of the canonical anticommutation relations for Majorana fermions
\begin{align}
    \{\kryl{j},\kryl{k}\}=2\mathbf{M}^G_{jk} I
    \label{eq:nocar_main}
\end{align}
for \(\mathbf{M}^G\) a symmetric, positive semidefinite matrix of rank \(d\). 
The Krylov modes thus constitute a basis set for the fermionic eigenmodes of these models.
It is additionally shown in Ref.~\cite{chapman2023unified} that these modes are linear combinations of \emph{anchored path-product operators}
\begin{align}
    h_P \coloneqq \chi \left(\overrightarrow{\textstyle\prod}_{i=1}^{\ell} h_{\Pv{i}}\right).
    \label{eq:hpdef}
\end{align}
For the FFD model,  \(P\) is an induced path in the zigzag ladder graph with  \(\Pv{0} = 0\) corresponding to the simplicial mode, i.e.,  \(h_{\Pv{0}} \coloneqq \chi\),  \(0 \leq \Pv{i} \leq m\) for all  \(i \in \{0, \dots, \ell\}\),  \(\Pv{i} < \Pv{j}\) for  \(i < j\), and the only factors in  \(h_P\) with which  \(h_{\Pv{i}}\) anticommutes are  \(h_{\Pv{i - 1}}\) and  \(h_{\Pv{i + 1}}\) (with  \(h_{\Pv{0}} = \chi\)), if they are present. Because  \(h_{\Pv{0}} = \chi\), these path-products are said to be \emph{anchored}, and we distinguish these from general path-product operators for which this condition does not necessarily hold. 
Operators in the span of path-product operators admit a \emph{path-product expansion}.

In Ref.~\cite{fukai2026solvingb}, where this terminology was first introduced, the authors give an explicit anchored path-product expansion for the eigenmodes of any Hamiltonian with an (even-hole, claw)-free frustration graph (including  \(H_{\mathrm{FFD}}\)).
Our first main result is an extension of this to certain quadratics in the fermionic eigenmodes, the \emph{Krylov currents}
\begin{theorem}[
    restate=KrylovCurrents,
    name={Path support of FFD Krylov currents},
    label={thm:ffd_currents_are_paths}]
    For the simplicial mode  \(\chi\) which anticommutes with only  \(h_1\) or  \(h_{m} \coloneqq h_{n - 2}\) in the FFD Hamiltonian, the Krylov current operators
    \begin{align}
        J^{(\ell)} \coloneqq i \sum_{j = 0}^{\ell} (-1)^j \e[\eta^{(j)}, \eta^{(\ell - j)}] \mh{3mm} \text{(\(\ell\) odd)}
        \label{eq:current_operators}
    \end{align}
    admit a path-product expansion.
\end{theorem}
In \cref{thm:ffd_currents_are_paths}, we restrict to odd \(\ell\) because \(J^{(\ell)} = 0 \) trivially for \(\ell\) even.
The Krylov currents thus span a subspace of extensions to  \(H_{\mathrm{FFD}}\) 
\begin{align}
    H_{\mathrm{NG}}[\bs{a}, \bs{b}] = a_0 H_{\mathrm{FFD}}[\bs{b}] + H_J[\bs{a}, \bs{b}]\;,\label{eq:ngmodels}
\end{align}
where
\begin{equation}
    H_{J}[\bs{a}, \bs{b}] = -\sum_{\ell=1,\text{ odd}}^d a_{\ell} J^{(\ell)}[\bs{b}]\label{eq:hjdef}
\end{equation}
under which the Hamiltonian admits a non-generic solution and whose terms correspond to induced paths in the zigzag ladder graph.
The  \(J^{(\ell)}\) themselves can be viewed as a Lagrange bilinear concomitant in Sturm--Liouville theory.
Namely, letting  \(\chi(t) \coloneqq e^{iHt/2} \chi e^{-iHt/2}\), we have
\begin{align}
    \int_0^T dt \,\e[\chi(t), \frac{d^{\ell + 1}}{dt^{\ell + 1}} \chi(t)] = -\frac{i}{2}\left(J^{(\ell)}(T) - J^{(\ell)}(0)\right)\;
\end{align}
where, similarly,  \(J^{(\ell)}(t) = e^{iHt/2}J^{(\ell)}e^{-iHt/2}\).

The solutions to the models in \cref{eq:ngmodels} thus consist of transforming to the eigenmodes of  \(H_{\mathrm{FFD}}\), where they are quadratic, but contain couplings between different eigenmodes.
In the eigenmode basis, the model can then be completely solved by a Gaussian transformation.
These solutions are non-generic in that they apply only when the coefficients in the Pauli expansion of \(H_J[\bs{a}, \bs{b}]\) have the correct relation to those of  \(H_{\mathrm{FFD}}\).
In particular, the frustration graph of \(H_{\text{NG}}[\bs{a}, \bs{b}]\) generally contains \emph{claws}, which are known obstructions to generic free-fermion solvability.
An example of a term appearing in the Krylov current \(J^{(5)}\) is shown in \cref{fig:zigzag_ladder}, together with the induced claw to which it belongs.

We prove \cref{thm:ffd_currents_are_paths} using identities involving the generating function of the Krylov currents.
In particular, we show that non-path contributions to this generating function either cancel pairwise or via a recursion relation involving the frustration graph.
The result is shown by equating coefficients of the generating function.

Our second main result concerns the algebraic completeness of the Krylov currents
\begin{theorem}[
    restate=currentAlgebra,
    name={Current algebra},
    label={thm:ffd_current_algebra}]
    Let  \(H_{\mrm{FFD}}\) be the FFD Hamiltonian with  \(b_1 = \ldots = b_6 = 1\) and  \(\chi = X_1\) (many, if not all, other parameter and mode choices are also fine; cf. \cref{rem:lie_condition}). 
    Let
    \begin{equation}
        \liek \coloneqq \Lie_{\mathbb{R}} \groupset{J^{(\ell)}}{\ell>0 \text{ odd}}
    \end{equation}
    be the real Lie algebra generated by the odd Krylov currents with the Lie bracket  \(\e[\cdot, \cdot]_{\liek} = i\e[\cdot, \cdot]\). 
    Then  \(\liek\) is the full Lie algebra of quadratic Majorana fermions, that is
    \begin{align}
        \liek &= \Lie_{\mathbb{R}} \groupset{i\mm_a\mm_b}{a, b \in \e{0, \ldots, d-1}} \label{eq:lie_mm} \\
        &\cong \lieso{d}\;, \label{eq:lie_so}
    \end{align}
    where  \(\e{\mm_a}_a\) are Majorana operators and  \(\lieso{d}\) is the Lie algebra of skew symmetric matrices.
\end{theorem}
The proof of \cref{thm:ffd_current_algebra} (in \cref{app:current_algebra}) follows inductively: including successive currents in the free-fermion dynamical Lie algebra results in the inclusion of a fixed set of Krylov bilinears. 
It relies on the recursive definition of the Krylov modes as well as the generalised canonical anticommutation relations \cref{eq:nocar_main}.
In fact, this proof uses only this relation and a nondegeneracy condition \cref{rem:lie_condition}, so it applies to any simplicial, claw-free model satisfying that condition.
However, we explicitly state it in a restricted setting to connect it to \cref{thm:ffd_currents_are_paths}, which is more restricted in scope.
Note that, in this more restricted setting, there is no basis for the Lie algebra for which every operator in the basis admits a path-product expansion. 
This is because there exist bilinears in the Krylov modes that do not have a path-product expansion (e.g. \([\eta^{(2)}, \eta^{(5)}]\) for \(m \geq 7\)), and the property of having a path-product expansion is closed under linear combination. 
Thus, at least one element of any such basis must not have this property.

We propose that our results can be extended in future work by the resolution of several conjectures.
Foremost, we conjecture that \cref{thm:ffd_currents_are_paths} holds for Hamiltonians with general chordal, claw-free frustration graphs and arbitrary choice of simplicial mode.
In this work, we prove the result for the FFD model with the simplicial clique defining  \(\chi\) chosen as a simplicial \emph{vertex} of the frustration graph.
If the latter condition is necessary---namely, if  \(\chi\) must be defined relative to a simplicial vertex as opposed to an arbitrary simplicial clique---then the former statement may still be true, as all chordal, claw-free graphs are known to contain a simplicial vertex \cite{Dirac1961}.

Our second conjecture concerns the free-fermion solvability of general models. 
We conjecture that a many-body quantum system with a given frustration graph is free-fermion solvable if and only if one can find an analog of  \(\chi\) such that the accompanying Krylov modes all admit anchored path product expansions (after cancellations and symmetries are accounted for).
Sufficiency follows straightforwardly from the proof for simplicial, claw-free models in Ref.~\cite{chapman2023unified}.
We reproduce that proof in \cref{app:kpsconjecture}.
To the best of our knowledge, we have no counterexamples to this conjecture, even in the non-generic case.
We return to the formal statements of both of these conjectures in \cref{sec:conjectures}.

\subsection{Numerical Analysis}\label{sec:numerical_analysis}

We restrict the model defined in \cref{eq:ngmodels} to the case where  \(m = 3N\) for some  \(N \in \mathbb{N}\). 
The parameters  \((b_j)_j\) and  \((a_{\ell})_{\ell}\) are all real and we set  \(a_0 = 1\). 
 \(\alpha(G)\) is the number of consecutive disjunct triangles in the frustration graph  \(G\) of the FFD model as well as its independence number. 

We first analyse the growth of the Pauli support of the Hamiltonian, as well as the number of claws in the frustration graph with  \(\alpha(G)\).
The exact numbers of vertices and claws depend on the connectivity of the simplicial mode. 
Specifically we can choose a simplicial mode that is either connected to only the first vertex, only the first two vertices or only the first three vertices in the frustration graph of the FFD model. 
The specific representation of the simplicial mode does not change the frustration graph, however it changes the specific representation of those vertices as operators. 
The choice of the parameters  \((b_m)_m\) and  \((a_l)_l\) seems not to change which operators appear in the model (e.g., a random choice results in the same operators as when setting all parameters to  \(1\).), except, of course, if they are specifically chosen such that some terms cancel or if they are set to zero.

\Cref{fig:merged_number_plot}~(a) shows the number of vertices,  \(N_V\), depending on the number of triangles  \(\alpha(G)\). 
We fit the  \(\chi = Z_1Z_2Z_3\) case with the exponential function,
\begin{equation}
    \text{exp}(N, F_1, F_2, F_3) = F_1 \ee{F_2 N} + F_3\;, \label{e"exp_fit}
\end{equation}
and the monomial
\begin{equation}
    \text{mono}[5](N, F_1) = F_1 N^5 \label{e"mono_fit}\;,
\end{equation}
where  \(F_1, F_2, F_3 \in \mathbb{R}\) are the fitting parameters.
The curve of the monomial fit leads us to conjecture that in the limit  \(\alpha(G) \to \infty\), the number of terms grows faster than  \(N^5\), potentially rather exponential, implying that the model is not sparse.

\Cref{fig:merged_number_plot}~(b) shows the number of claws,  \(N_{\claw}\). 
Again, we see that  \(N_{\claw}\) behaves approximately exponentially in the plotted regime; for higher  \(\alpha(G)\) there are so many claws that it soon becomes numerically infeasible to count them. 
In particular, the model is far from having a generic free-fermionic solution based on a claw-free frustration graph.

\begin{figure}[t]
    \centering
    \includegraphics{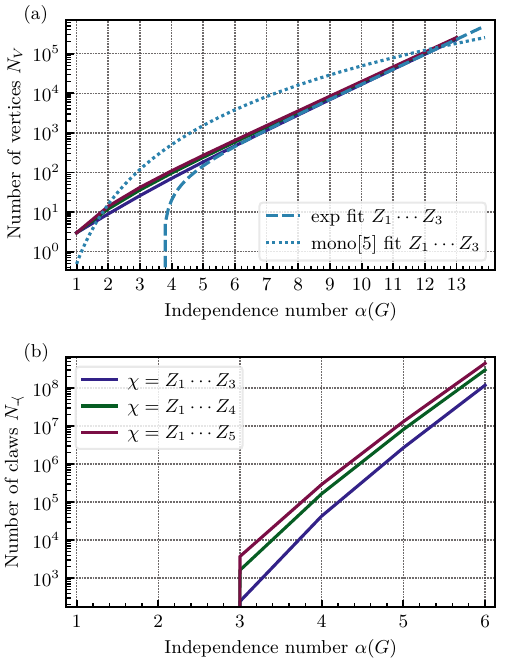}
    \caption{(a) The number of vertices  \(N_V\) in the frustration graph of the model in \cref{eq:ngmodels} depending on the independence number  \(\alpha(G)\), for (the) three different choices of simplicial modes. 
    The case  \(\chi = Z_1Z_2Z_3\) is fitted with an exponential function,  \(\text{exp}\), and a monomial,  \(\text{mono}[5]\), defined in \cref{e"exp_fit,e"mono_fit} respectively.
    (b) The number of claws  \(N_{\claw}\) in the frustration graph of the model in \cref{eq:ngmodels} depending on the independence number  \(\alpha(G)\), for (the) three different choices of simplicial modes.}
    \label{fig:merged_number_plot}
\end{figure}

In the limit  \(\alpha(G) \to \infty\), the FFD model is gapless at the critical point where  \(b_j = \mathrm{const}\) for all  \(j \in \e{1, \ldots, m}\). 
That is, there is no finite energy gap  \(\Delta \coloneqq 2 \mathrm{min}_j \lvert \varepsilon_j \rvert\) between the ground state energy and the first excited energy.
Reference \cite{fendley2019free} analyses this gap in more detail for the parametrization  \(H_{\mrm{FFD}} = H_{\mrm{FFD}}[\alpha, \beta, \gamma]\) defined by  \(b_{j = 1 \!\!\mod 3} = \sqrt{\alpha}\),  \(b_{j = 2 \!\!\mod 3} = \sqrt{\beta}\) and  \(b_{j = 0 \!\!\mod 3} = \sqrt{\gamma}\) for  \(\alpha, \beta, \gamma \in \mathbb{R}_+\). 
In the following we do a numerical analysis of the model
\begin{equation}
    H[\alpha, \beta, \gamma, \bs{a}] = H_{\mrm{FFD}}[\alpha, \beta, \gamma] + H'_J[\alpha, \beta, \gamma, \bs{a}]\label{e"prime_currents_model}
\end{equation}
where  \(H'_J\) is defined as  \(H_J\) in \cref{eq:hjdef} however with the current operators normalised (cf. \cref{e"renormalised_currents_model}). 
We fix the simplicial mode  \(\chi\) to be connected to the full simplicial clique (all three vertices of the first triangle) in the FFD model for technical reasons, e.g.,  \(\chi = Y_3\). 
This choice lies outside the hypothesis of \cref{thm:ffd_currents_are_paths}, and the path support of the corresponding currents is covered by \cref{conj:graphical_current_algebra}(i). 
\Cref{sec:phase_diagram_details} provides more details on the underlying calculations.

\Cref{fig:currents_phase} shows the energy gap  \(\Delta\) at the isotropic point  \(\alpha = \beta = \gamma = 1\) for different  \(\alpha(G)\), where we choose the weight of the currents to be constant w.r.t.  \(\ell\). 
Namely, we set  \(a_{\ell} = a \in \mathbb{R}\) and vary  \(a\). 
The phase diagram of the FFD model can be recovered at  \(a = 0\).
For nonzero  \(a\), the gap is very nearly constant except for regularly spaced minima on the logarithmic scale, where it approaches zero.
We see that, as the system size increases, the constant value of the gap is also decreasing.
\Cref{fig:triangle_phase_main} shows the corresponding phase diagram of the model, i.e., the gap as a function of the parameters  \(\alpha\),  \(\beta\) and  \(\gamma\) for a fixed number of triangles  \(\alpha(G) = 20\).
The current strength,  \(a\), is chosen around the first minimum (\(a \approx 1.8\)) that we see in \cref{fig:currents_phase}.

We see that the behaviour of the model in this regime is strikingly distinct from that of the FFD model.
In particular, \Cref{fig:currents_phase,,fig:triangle_phase_main}, suggest that there exists a value of  \(a\)---at least at finite  \(\alpha(G)\)---for which the isotropic point  \(\alpha = \beta = \gamma\) is an isolated critical point. 
Here, the perturbation prevents the gap from closing along the lines where two of the coefficients are equal and greater than the third coefficient, in contrast to the  \(a = 0\) case.
As we make a small displacement in  \(a\) away from this value, we see in \cref{fig:triangle_phase_main} that the isotropic point belongs to a distinct phase from the extreme points where one of the coefficients is much larger than the other two, as it is enclosed by a phase boundary where the gap closes.

The code for the numerical simulations can be found at \href{https://github.com/taeruh/nongeneric_free_fermions}{\path{github.com/taeruh/nongeneric_free_fermions}} with the numerical data at commit \href{https://github.com/taeruh/nongeneric_free_fermions/tree/26546a2414953fd0214c5ff70e9480d963d853f0}{26546a2} (in the \path{./results} directory).

\begin{figure}[t]
    \centering
    \includegraphics[width=\linewidth]{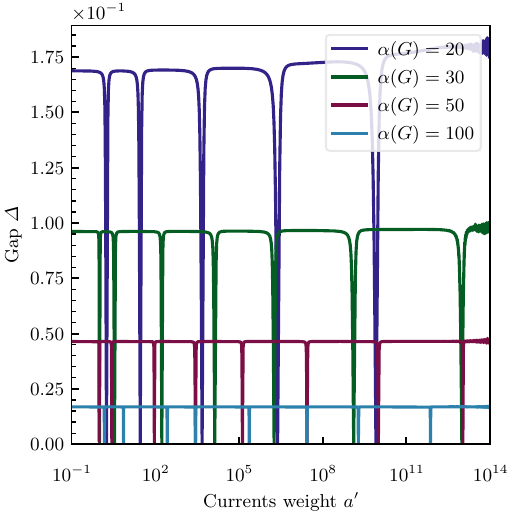}
    \caption{The energy gap,  \(\Delta\), of the model in \cref{e"prime_currents_model,e"renormalised_currents_model} at the isotropic point  \(\alpha = \beta = \gamma = 1\) with  \(a_\ell = a\) for all  \(\ell\) and different model sizes  \(\alpha(G)\). 
    The minima are continuous and genuine zeros, resolved as the
    sampling resolution is increased. 
    For high  \(a > 10^{13}\) the skew-diagonalisation numerics become noisy.}
    \label{fig:currents_phase}
\end{figure}
\begin{figure}[hbt!]
    \centering
    \includegraphics[width=\linewidth]{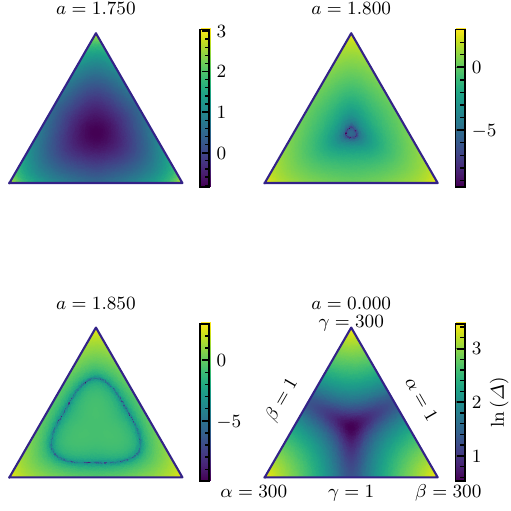}
    \caption{
    The phase diagram of the model in \cref{e"prime_currents_model,e"renormalised_currents_model} for  \(\alpha(G) = 20\) and  \(a_\ell = a\) chosen around the first minimum that we see in \cref{fig:currents_phase}, for all  \(\ell\). 
    The bottom-right diagram shows the axes labels. 
    It holds  \(\alpha + \beta + \gamma = 302\) (note that this absolute value can be renormalised; more interesting are the relative values of  \(\alpha\),  \(\beta\) and  \(\gamma\) to each other; also note that we do not allow these parameters to be zero, e.g., we never have  \(\alpha = 0\), in the numerical calculations as this causes instabilities, however, as we choose the maximal value fairly high, the relative values approximate this case, i.e.,  \(1/300 \approx 0\)). \Cref{fig:triangle_phase_dip1} shows the same phase diagrams, however, with more  \(a_\ell=a\) sample values, and \cref{fig:triangle_phase_dip2} shows the phase diagrams around the second minimum (\(a \approx 30\)), for  \(\alpha(G) = 20\), that we see in \cref{fig:currents_phase}.}
    \label{fig:triangle_phase_main}
\end{figure}

\subsection{Generalising Conjectures}\label{sec:conjectures}

\Cref{thm:ffd_currents_are_paths,thm:ffd_current_algebra} perhaps hint at something deeper:
\begin{conjecture}[
    restate=currentconjecture,
    name={Graphical current algebra},
    label={conj:graphical_current_algebra}]
    Let  \(H\) be a Hamiltonian with chordal, claw-free frustration graph  \(G\), and let  \(\chi\) be any simplicial mode of  \(H\).
    Then for every odd  \(\ell\) the Krylov current operator  \(J^{(\ell)}\) admits a path-product expansion.
\end{conjecture}
A positive resolution to \cref{conj:graphical_current_algebra} would directly generalize our result in two ways.
\begin{enumerate}[(i)]
    \item It would show that  \(\chi\) may be taken to be any simplicial mode of  \(H_{\text{FFD}}\) as opposed to one that only neighbors a simplicial vertex in the frustration graph.
    \item It would show that  \(G\) may itself be generalized from the zigzag-ladder graph to any chordal, claw-free frustration graph.
\end{enumerate}
Notably, it is possible for (i) or (ii) to independently fail without the conjecture being entirely false. 
If (i) is false, it may be that  \(G\) can still be generalized, since all chordal, claw-free graphs contain a simplicial vertex.
If (ii) is false, it may be that other simplicial modes of  \(H_{\mathrm{FFD}}\) satisfy the associated counterpart to our result.

The evidence for \cref{conj:graphical_current_algebra} is both computational and structural. 
First, we have verified the conjectured path-support property in all numerical examples we have tested. 
Second, our proof of \cref{thm:ffd_currents_are_paths} relies on features of the zigzag-ladder graph that make cancellations tractable (as shown in~\cref{app:current-paths}).  
In particular, we organize pairs of induced paths using the interval structure of the ladder and then show that some non-path contributions cancel by exchanging segments of them. 
We then show that the remaining non-path structures cancel via the recursion relations of the independence polynomial.
This gives a transparent proof in the FFD case, but it is still essentially a case analysis. 
We do not know of a fundamental obstruction to removing the dependence on this particular graph structure. 
In fact, it is known that all chordal, claw-free graphs admit a representation in terms of non-crossing paths on a clique tree \cite{chaplick2023intersection}.
This representation strictly generalizes the interval structure of the zigzag-ladder, and it forms the basis for the focus of our conjecture on chordal, claw-free graphs.
The difficulty is that a proof for arbitrary such graphs would have to replace the convenient interval ordering by the more general clique-tree representation.
We expect the required cancellations to be analogous, but the bookkeeping becomes substantially more involved and would obscure the main construction. 

If \cref{conj:graphical_current_algebra} is true, it suggests a yet broader organizing principle. 
The chordal, claw-free condition would be one way of guaranteeing that the Krylov tower (the sequence of Krylov modes  \(\{\eta^{(k)}\}_{k \geq 0}\) generated recursively by repeated application of  \(\mc{L}\)) has an induced-path description, but it may not be the only one. 
In other models, such as those involving even-hole or bubble-wand structures~\cite{chapman2023unified,fukai2026solvingb}, the same effective path support can arise only after imposing graphical cancellation relations specific to the frustration graph, or restricting to a symmetric subspace related to a particular graphical structure. 
This suggests that the induced-path support of the Krylov tower, possibly after such cancellations, is the more fundamental signature of free-fermion solvability. 

\begin{conjecture}[
    restate=Krylovconjecture,
    name={Krylov path-support for free fermions},
    label={conj:krylov_path_support_free_fermions}]
    Let  \(H\) be a Hamiltonian with frustration graph  \(G\), let  \(\chi\) be a fiducial operator belonging to the extended graph Clifford algebra of  \(G\), and denote the associated Krylov operators by  \(\{\kryl{k}\coloneqq\mc{L}^{k}[\chi]\}_{k}\). 
    Then  \(H\) admits a free-fermion solution on the Krylov space generated by  \(\chi\) iff  \(\kryl{k}\) admits an anchored path-product expansion for all  \(k\ge0\), possibly after applying the exact graphical cancellation relations of the model.
\end{conjecture}
As mentioned in the beginning of this section, the sufficiency of the anchored path-product expansion is shown in \cite{chapman2023unified}.
We do not know of an example that would contradict necessity, including the ones described in this work.

\cref{conj:krylov_path_support_free_fermions} states that graphical cancellations are necessary to enforce the algebraic closure needed for free-fermion solvability. 
Namely, once the Krylov tower closes, the anticommutator-form  \(\{\kryl{j},\kryl{k}\}\) should reduce to a scalar positive form on the Krylov space, allowing the  \(\kryl{k}\) to be orthonormalised into Majorana operators \(\{\gamma_j\}\). 
The Hamiltonian adjoint action then effectively becomes a linear map on the space of Majorana vectors, giving a free-fermion solution.
A restatement of the proof of sufficiency, together with further discussion on this conjecture as it relates to models in the literature, can be found in \cref{app:kpsconjecture}.

\section{Methods} \label{sec:background}
\subsection{Graph-theoretic approach to Jordan-Wigner}
\label{sec:tfim_seed_JW}

\begin{figure*}[t]
    \centering
    \includegraphics[width=0.8\linewidth]{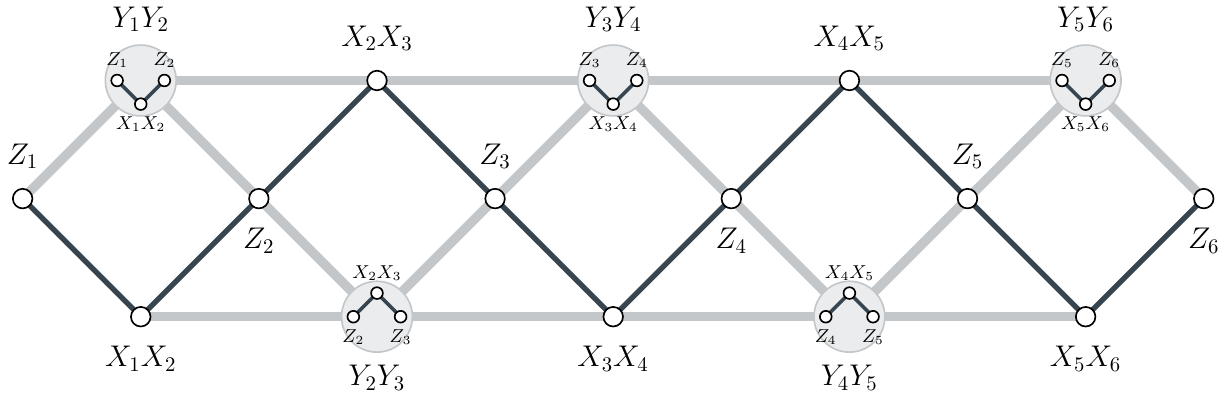}
    \caption{
    The transverse-field Ising model as a seed: from paths in the TFIM, we can construct all other models which admit an exact solution via a linear mapping to free fermions. 
    Here we depict the construction of the XY model from paths in the frustration graph of the TFIM. 
    The TFIM runs zigzag through the XY model as a black line, with the other terms formed from products of paths in this model.}
    \label{fig:TFIM_JW}
\end{figure*}

We introduce the frustration-graph formalism via a motivating argument, by which we demonstrate that all linear models are derived from the open chain transverse-field Ising model (TFIM) \cite{jordan1928uber,schultz1964two-dimensional,Lieb1961}.
Our main results can be viewed as a generalization of this construction from the TFIM to the FFD model.
The key idea is shown in \cref{fig:TFIM_JW}.
Consider the Hamiltonian:
\begin{equation}
    H_{\mrm{TFIM}}=\sum_{j=1}^{n}Z_j +\delta\sum_{j=1}^{n-1} X_j X_{j + 1}
    \label{eq:TFIM}
\end{equation}
on an open chain of  \(n\) qubits.
The solution to this model proceeds via an exact mapping to non-interacting fermions. 
Introduce  \(2n\) Majorana operators via the Jordan-Wigner transform
\begin{equation}
    \mm{2j - 1} = \bigotimes_{1 \leq m<j} Z_m \otimes X_j, \quad \mm{2j} = \bigotimes_{1 \leq m<j} Z_m \otimes Y_j \;.
    \label{eq:JW_majoranas}
\end{equation}
The Pauli strings in \cref{eq:JW_majoranas} satisfy the canonical anti-commutation relations (CAR):
\begin{equation}
    \{\mm{j},\mm{k}\}=2\delta_{jk} I \qquad \forall \; j,\, k \in \{1, \dots, 2n\}\;.
\end{equation}
In particular, this implies  \(\gamma_j^2 = I\) for all  \(j\).
Substituting \cref{eq:JW_majoranas} gives 
\begin{equation}
    \label{eq:tfimsolution}
    H_{\mrm{TFIM}}=-i \sum_{j=1}^{n}\mm{2j-1}\mm{2j} - i \delta \sum_{j=1}^{n-1} 
    \mm{2j}\mm{2j+1}\;
\end{equation}
Thus the interacting spin Hamiltonian has become a quadratic Majorana Hamiltonian.
The CAR imply that Majorana bilinears generate the Lie subgroup SO(\(2n\))  \(<\) U(\(2^n\)), whose dimension is exponentially smaller than that of the full unitary group on  \(n\) qubits.
Accordingly, quadratic Majorana Hamiltonians are well-known to be exactly diagonalizable by efficient classical methods, and many objects derived from them admit efficient exact classical descriptions.
As a result, a mapping to a Hamiltonian of the form in \cref{eq:tfimsolution} is said to constitute an exact solution of the model.

Since both Pauli and Majoranas operators either commute or anti-commute, it is convenient to describe these relations by a graph, called the frustration graph

\begin{definition}[Frustration graph]
    \label{definition:frustrationgraph}
    Let  \(\{\mu^{\bs{j}}\}\) constitute a graph Clifford algebra: a set of Hermitian operators such that  \((\mu^{\bsj})^2 = I\) which pairwise either commute or anticommute (e.g. a subset of Pauli operators  \(\{\sigma^{\bsj}\}\) or Majorana monomials  \(\{\mm{}^{\bsj}\}\)).
    Given
    \begin{equation}
        \label{eq:hamiltonian_sum}
        H=\sum_{\bsj \in V}h_{\bsj}=\sum_{\bsj \in V}b_{\bsj}\mu^{\bsj}
    \end{equation}
    define its frustration graph as the graph  \(G = (V, E)\) with
    \begin{equation}
        E = \{\{\bsj,\bsk\} \mid h_{\bsj}h_{\bsk} = -h_{\bsk}h_{\bsj}\}\;.
    \end{equation}
    Throughout this work, we assume the frustration graph is connected without loss of generality.
\end{definition}

For a given fermionic Hamiltonian  \(H\), the frustration graph is determined by the CAR. 
Furthermore, the existence of a free-fermion solution of the form in \cref{eq:tfimsolution} can be inferred from the frustration graph.
This is the main result of Ref.~\cite[Theorem 2]{Chapman2020}: there is \textit{linear map} from the Pauli strings of  \(H\) to bilinear fermions if and only if frustration graph is the \emph{line graph} of a graph  \(R\).
\begin{definition}[Line graph]
    \label{definition:linegraph}
    Given a graph  \(R=(V,E)\), the \emph{line graph}  \(L(R)=(E,F)\) of a root graph  \(R\) is the graph whose vertices correspond to the edges of  \(R\) and
    \begin{equation}
        F = \{\{e,f\} \mid e,f\in E,\ \abs{e \cap f}=1\}\;.
    \end{equation}
    \(R\) is called the \emph{root graph} of  \(L(R)\).
\end{definition}

It is straightforward to see that the frustration graph  \(G\) of the  \(\mrm{TFIM}\) is a path (see~\cref{fig:TFIM_JW}) and thus a line graph (the root graph is a path with one additional vertex). 
By taking products of terms in  \(H_{\mathrm{TFIM}}\) along sub-paths of  \(G\), we generate new quadratic Majorana operators, as interior Majorana factors cancel, and we can generate any Majorana bilinear  \(\mm{j} \mm{k}\) this way. 
Thus, all models which are solvable by the Jordan-Wigner transformation can be constructed as a path-product expansion deriving from the TFIM.  
For example, we may construct the one-dimensional transverse-field  \(\mrm{XY}\)-model
\begin{equation}
    \begin{split}
        H_{\mrm{XY}}= \sum_{j=1}^{n-1} \Big( &X_j X_{j + 1}+Y_j Y_{j + 1}\Big)+\lambda\sum_{j=1}^{n}Z_j\;,
    \end{split}
\end{equation}
by adding paths of length three to the Hamiltonian for the  \(\mrm{TFIM}\),  \(Y_j Y_{j + 1} = \left(Z_j\right)\left(X_j X_{j + 1}\right)\left(Z_{j + 1}\right)\), as depicted in \cref{fig:TFIM_JW}. 
The TFIM on periodic boundaries, and even the Kitaev honeycomb model~\cite{kitaev2006anyons,kells2009description}, are further examples of line-graph models, and can likewise be constructed as a path-product expansion from the TFIM on open boundaries, together with the inclusion of stabilizer symmetries.

In this regard, we consider the TFIM a `seed' model from which we may generate all models which admit a linear mapping to free fermions. 
We refer to these models as \textit{line-graph models}.

\subsection{Free fermions in disguise}
\label{sec:ffd_solution}

Not all free-fermion-solvable models admit a solution via a linear mapping. 
The central example for our work is the aforementioned four-fermion model, introduced by Fendley and defined in \cref{eq:ffd} \cite{fendley2019free}.
The model admits an exact nonlinear solution via a mapping to free-fermions, and it provably admits no linear solution.
As the solution is generic, it nevertheless also follows from graph theory. 
Furthermore, the solution can be generalized to apply to all Hamiltonian spin models with an (even-hole, claw)-free (ECF) frustration graph \cite{elman2021free}, and we broaden our discussion to include all such \emph{ECF models} in our summary of these results here.

\begin{definition}[Claw and chordal graphs]
    \label{def:clawchordal}
    A \emph{claw} is the graph  \(K_{1,3}\): a vertex adjacent to three pairwise nonadjacent vertices. 
    A graph is \emph{claw-free} if it contains no claw as an induced subgraph. 
    A graph is \emph{chordal} if every induced cycle has length three, or equivalently, if every cycle of length four or greater has a chord.
\end{definition}
In general, we denote the subgraph induced by a vertex set \(U\) by \(G[U]\), so a graph is claw-free if and only if there is no such \(U\) for which \(G[U] \simeq K_{1, 3}\).

\subsubsection{Independent Sets and Transfer Operator}

An important concept for the solutions of ECF models is the independence structure of the frustration graph.
We motivate this by introducing a set of conserved charges.

\begin{definition}[Independent-Set Charges]
    \label{def:indset}
    An independent set in a graph  \(G \coloneqq (V, E)\) is a subset  \(S \subseteq V\) of mutually nonadjacent vertices:  \(E \cap S^{\times 2} = \varnothing\).
    Let  \(\mc{S}^{(k)}\) denote the collection of all independent sets of order  \(k \in \{0, \dots, \alpha(G)\}\) in  \(G\), where  \(\mc{S}^{(0)} \coloneqq \e{\varnothing}\), and  \(\alpha(G)\) is the independence number of  \(G\), the order of its largest independent set.
    Additionally, let  \(\mc{S}_G = \cup_{k = 0}^{\alpha(G)} \mc{S}^{(k)}\).
    
    The \textit{independent-set charges} of the Hamiltonian with frustration graph  \(G\) are defined as
    \begin{equation}
        Q^{(k)} \coloneqq \sum_{S\in \mc{S}^{(k)}}\prod_{\bsj\in S}h_{\bsj} \coloneqq \sum_{S\in \mc{S}^{(k)}}h_{S}\;,
    \end{equation}
    where  \(h_S\) is the product of the Hamiltonian terms in the set  \(S\) and  \(Q^{(0)} = I\) and  \(Q^{(1)}=H\).
\end{definition}
Note that the operator ordering in the product over each independent set  \(S\) is irrelevant, as all factors in  \(S\) commute by construction.
The  \(Q^{(k)}\) satisfy
\begin{proposition}[Theorem 1~\cite{elman2021free}] \label{lem:qconservation}
    Let  \(G\) be a claw-free frustration graph, then the independent-set charges commute,
    \begin{equation}
        [Q^{(j)},Q^{(k)}]=0\;.
    \end{equation}
    This implies the integrability of Hamiltonians whose frustration graphs are claw-free.
\end{proposition}
The transfer operator is constructed as the generating function of the independent set charges.
\begin{definition}[Transfer operator] \label{def:transfer}
    Given a Hamiltonian  \(H\) with frustration graph  \(G\), the transfer operator is a function of the spectral parameter  \(u \in \mathds{C}\) as 
    \begin{align}
        T_G(u) &\coloneqq \sum_{k = 0}^{\alpha(G)} (-u)^k Q^{(k)} = \sum_{S \in \mc{S}_G} \prod_{\bs{j} \in S} (-u h_{\bs{j}})\;.
    \end{align}
\end{definition}
In particular, \cref{lem:qconservation} implies that, for  \(G\) claw-free, the transfer operator commutes with itself at distinct values of its spectral parameter  \([T_G(u), T_G(u')] = 0\).

The transfer operator satisfies a useful product relation, which we specify with the help of the following definition
\begin{definition}[Independence polynomial] \label{def:indeppoly}
    Given a Hamiltonian  \(H = \sum_{\bs{j} \in V} b_{\bsj} \mu^{\bsj} \) with frustration graph  \(G\). 
    The vertex-weighted independence polynomial is given by
    \begin{align}
        P_{G}(x) = \sum_{k = 0}^{\alpha(G)} x^k\left(\sum_{S \in \mc{S}^{(k)}_G} \prod_{\bsj \in S} b_{\bs{j}}^2 \right) = \sum_{S \in \mc{S}_G} \prod_{\bsj \in S} (x b_{\bsj}^2) \label{eq:indeppoly}
    \end{align}
    For brevity, we refer to this simply as the independence polynomial.
\end{definition}

In what follows, it is convenient to utilize the following definitions.
We denote the (open) neighborhood of a vertex \(\bs{j}\) in a graph \(G\) by \(N_G(\bs{j}) = \{\bsk \in V| \{\bsj, \bsk\} \in E\}\); we denote the closed neighborhood by \(N_G[\bs{j}] = N_G(\bs{j}) \cup \{\bs{j}\}\).
For a vertex subset \(U\), let \(N_G[U] = \cup_{\bs{j} \in U} N_{G}[\bs{j}]\).

Both the transfer operator and the independence polynomial satisfy recursion relations with respect to a clique, or complete subgraph, in  \(G\).
In particular, we have
\begin{align}
    T_G(u) &= T_{G\setminus K}(u) - u\sum_{\bsj \in K} h_{\bsj} T_{G\setminus N[\bsj]}(u) \;,\\
    P_G(x) &= P_{G\setminus K}(x) + x \sum_{\bsj \in K} b_{\bsj}^2 P_{G\setminus N[\bsj]}(x) \;. \label{eq:independence_polynomial_recursion}
\end{align}
We also have

\begin{proposition}[Product of transfer operators]
    \label{prop:productoftransfer}
    For the transfer operator of a Hamiltonian with ECF frustration graph  \(G\), as given in \cref{def:transfer}, we have
    \begin{align}
        T_G(u)T_G(-u) = P_G(-u^2)I
    \end{align}
    for  \(P_G\) the vertex-weighted independence polynomial given in \cref{def:indeppoly}.
\end{proposition}
Especially, this implies that the transfer operator is invertible when  \(P_G(-u^2) \neq 0\).
It is known that the zeros of the independence polynomial as  \(P_G(x)\) lie on the negative real axis for claw-free graphs \cite{leake2016generalizations} (if the zeros are real, then they must be negative since \(P_G\) has a constant term and positive coefficients).
Furthermore, by construction, we have  \(T_{G}^{\dagger}(iu) = T_G(-iu)\) when  \(u \in \mathds{R}\).
This allows us to define the \emph{unitary transfer operator}:
\begin{definition}[Unitary Transfer Operator]
    \label{def:uto}
    The unitary transfer operator, defined for  \(u \in \mathds{R}\), is given by
    \begin{align}
        V(u) = P_G(u^2)^{-1/2} T_G(iu),
    \end{align}
    and satisfies  \(V^{\dagger}(u)V(u) = I.\)
\end{definition}

\subsubsection{Simplicial Clique and Path Products}

As stated in the previous section, a \emph{clique} is a complete subgraph.
A \textit{simplicial clique}  \(K\) is a clique such that, for all  \({\bs{v}}\in K\), the neighbourhood of  \(\bs{v}\) in  \(G\setminus K\) induces a clique, i.e.  \(G[N_{G\setminus K}(\bs{v})]\) is a clique.
A graph is \emph{simplicial} if it contains a simplicial clique.

\begin{definition}[Simplicial mode]
    Let  \(H\) be a Hamiltonian where its frustration graph  \(G\) contains a simplicial clique  \(K\). 
    A simplicial mode with respect to  \(K\) is an operator  \(\chi=\mu^{\bsjs}\) such that  \(\bsjs\notin V\) and  \(\chi\) satisfies
    \begin{align}
        \chi \mu^{\bsk} = (-1)^{\mathbf{1}(\bsk \in K)} \mu^{\bsk} \chi
    \end{align}
    where \(\mathbf{1}(E)\) is the indicator function for the event \(E\).
\end{definition}

The simplicial extension to the frustration graph is defined as
\begin{definition}[Simplicial extension] 
    Let  \(\chi\) be a \emph{simplicial mode} for  \(H\), a Hamiltonian with simplicial frustration graph  \(G\). 
    Define
    \begin{align}
        G^* \coloneqq (V \cup \{\bs{j}^*\}, E \cup \{(\bs{j}^*, \bs{j})\}_{\bs{j} \in K})
    \end{align}
    as the graph  \(G\) with an additional fiducial vertex  \(\bs{j}^*\), which neighbours only the vertices in  \(K\), and as such corresponds to  \(\chi\). 
\end{definition}

Using the definition of the simplicial extension, we define the path-product operator, which generalizes \cref{eq:hpdef} to arbitrary models

\begin{definition}[Path-Product operator]
    \label{def:indpath}
    Let  \(\mc{P}\) be the set of all induced paths in  \(G\), and let  \(\mc{P}^*\) be the set of all anchored induced paths in  \(G^*\) (i.e. paths with an endpoint at  \(\bs{j}^*\)). 
    Define the anchored path-product operator as
    \begin{align}
        h_P = \chi h_{\Pv{1}} \dots h_{\Pv{\ell - 1}} h_{\Pv{\ell}}
    \end{align}
    for
    \begin{align}
        P \coloneqq \bs{j}^* \md \Pv{1} \md \dots \md \Pv{\ell - 1} \md \Pv{\ell} \in \mc{P}^*
    \end{align}
    We denote the endpoint  \(\Pv{\ell} \neq \bs{j}^*\), for \(\ell > 0\), as the free endpoint of  \(P\), and also write  \(\Pv{0} = \bs{j}^*\).
    General path-product operators are defined analogously.
\end{definition}
Finally, we define the \emph{single-vertex deformation} of a path  \(P\) relative to a vertex  \(\bsk \notin P\) for which 
\begin{align}
    N_{G^*}[\bsk] \cap P = \Pv{i - 1} \md \Pv{i} \md \Pv{i + 1}
\end{align}
as the path  \(P' = (P\setminus \{\Pv{i}\}) \cup \{\bsk\}\). 
The vertex  \(\Pv{i}\) is called the \emph{clone} to  \(\bsk\) in  \(P\), and this relationship is denoted  \(\bsk \prec_{P} \Pv{i}\).
We utilize this construction in the proof of \cref{thm:krylov-restated}, stated below.

\subsubsection{Full solution}

The transfer operator and simplicial mode allow us to express the disguised fermionic eigenmode basis in which the Hamiltonian is diagonalized
\begin{definition}[Incognito modes \cite{elman2021free}]\label{def:incognito}
    Given a Hamiltonian  \(H\) with (even-hole, claw)-free frustration graph  \(G\), the incognito modes are defined as
    \begin{align}
        \psi_j &\coloneqq N_j^{-1} T(-u_j) \chi T(u_j)
    \end{align}
    with  \(N_j\) a normalization factor, and  \(u_j\) the positive roots of the vertex weighted independence polynomial  \(P(-u_j^2) = 0\). 
\end{definition}

The results of refs.~\cite{elman2021free,chapman2023unified,mann2025graph} imply that an ECF Hamiltonian admits a free-fermion solution of the form
\begin{align}
    \label{eq:diaghamiltonian}
    H=\sum_{j=1}^{\alpha(G)}\epsilon_j[\psi_j,\psi_j^\dagger] = -\sum_{j=1}^{\alpha(G)}2\epsilon_j \left(\psi_j^\dagger \psi_j - \frac{1}{2}\right),
\end{align}
where the incognito modes  \(\{\psi_j\}_j \cup \{\psi_j^\dagger\}_j\) of  \(H\) are fermionic ladder operators, obeying the canonical anti-commutation relations. 
They are eigenmodes of the Hamiltonian commutator, i.e.,
\begin{equation}
    [H, \psi_j^{\dagger}] = -2\varepsilon_j \psi_j^{\dagger} \mh{5mm} \text{and} \mh{5mm} [H, \psi_j] = 2\varepsilon_j \psi_j \mh{5mm}\label{eq:eigenmodes}
\end{equation}
for all  \(j \in \{1, \dots, \alpha(G)\}\)\;.

\subsection{Krylov Modes for Chordal, Claw-Free Models}
\label{subsec:krylov}
The known solution to an ECF model constructs fermionic ladder operators from conserved quantities defined by graphical structures, and this solution is generic.
Furthermore, the solution is nonlocal and nonlinear in that it does not generally map monomials in the graph Clifford algebra to fermionic monomials. 
It is thus difficult to identify localized fermions by an approach similar to that outlined in \cref{sec:tfim_seed_JW}.
However, in Ref.~\cite{chapman2023unified} a new solution method for these models was presented (and subsequently developed in Ref.~\cite{fukai2026solvingb}) that relies on the construction of Krylov operators, generated from nested commutators of the Hamiltonian with the simplicial mode. 

\begin{definition}[Krylov modes] 
    \label{def:krylovmode}
    Given a Hamiltonian  \(H\) with simplicial frustration graph, let  \(\chi\) be the corresponding simplicial mode.
    The Krylov modes  \(\eta^{(k)}\) are 
    \begin{align}
        \kryl{k} = \mc{L}^k[\chi] \mh{4mm} \text{where} \mh{4mm} \mc{L}[X] \coloneqq \frac{i}{2}[H, X]\;.
    \end{align}
\end{definition}
These operators are Hermitian. 
In Ref.~\cite{chapman2023unified}, it was shown that when the frustration graph of a model is simplicial, claw-free, the Krylov operators admit a path-product expansion.
This result, restricted to chordal, claw-free graphs, is central to the present work, and we recall the lemma~\cite[Theorem 3]{chapman2023unified} as follows:
\begin{lemma}[{restate=[name=restatement]KrylovModes}]
    \label{thm:krylov-restated}
    Given a Hamiltonian  \(H\) with chordal, claw-free frustration graph  \(G\). 
    The Krylov modes satisfy
    \begin{align}
        \kryl{k} = (-i)^k \sum_{P \in \mc{P}^*} \left(\mbf{A}^{k}\right)_{P, \{\bs{j}^*\}}h_P \label{eq:pathsupport}
    \end{align}
    where  \(\mbf{A}\) is a directed adjacency matrix of the induced path tree \cite{leake2016generalizations} of  \(G^*\) with real weights, and
    \begin{align}
        \left(\mbf{A}^{k}\right)_{P, \{\bs{j}^*\}} = \left(\mbf{A}^{k}\right)_{\widetilde{P}, \{\bs{j}^*\}} \label{eq:deformeq}
    \end{align}
    for all  \(k \geq 0\), for all pairs of induced paths  \(P\),  \(\widetilde{P}\), related by a single-vertex deformation. 
    Furthermore, for  \(j\),  \(k \geq 0\) we have
    \begin{equation}
        \{\kryl{j},\kryl{k}\}=2\mathbf{M}^G_{jk} I
        \label{eq:nocar}
    \end{equation}
    where  \(\mathbf{M}^G\) is real-symmetric and  \(\e*(\mathbf{M}^G_{jk})_{0 \leq j,k \leq d-1}\) is positive definite (\(d\) specified below).
\end{lemma}
\cref{thm:krylov-restated} establishes that the operators generated by repeated commutators  \(\kryl{k}\) live in a finite-dimensional Krylov subspace. 
Indeed, as there is only a finite number of induced paths, the matrix  \(\mbf{A}\) is finite dimensional.
Therefore, there must be a minimal  \(d \in \mathbb{N}\) -- the dimension of the Krylov subspace -- such that  \(\kryl{d}\) is a \emph{real} linear combination of  \(\{\eta^{(k)}\}_{0 \leq k < d}\). 
Furthermore, the Krylov operators satisfy a nonorthogonal version of the CAR (\cref{eq:nocar}), implying that  \(\mc{L}\) acts as a single-particle matrix on these modes, and therefore that the Hamiltonian is quadratic.

Whereas the original proof of this lemma in the simplicial, claw-free setting requires additional considerations due to the presence of holes in the graph, we provide a simpler proof of this result for chordal, claw-free frustration graphs in \cref{app:induced-paths}. 
This follows by expressing the first Krylov step in the incognito-mode basis and using the fact that these modes diagonalize the adjoint action of the Hamiltonian, so higher Krylov iterates simply attach powers of the eigenvalues. 
Transforming back to the Hamiltonian terms then shows that each  \(\kryl{k}\) is a sum over induced-path operators  \(h_P\).

Since  \(\mbf{M}^G\) is real-symmetric and positive definite, this provides us with another construction for quasi-local fermions in terms of sums of paths operators. 
Indeed, these Krylov fermions  \(\{\kryl{k}\}\) may be orthonormalised to form a set of Majorana operators \cite{chapman2023unified}.

\subsection{Cayley Transform}
The proof of \cref{thm:ffd_currents_are_paths} proceeds by replacing the individual Krylov modes with a single generating function that encodes the entire Krylov tower. 
The Cayley transform \begin{equation}
    \mc{C}(u) = \big(I+u\mc{L}\big)\big(I-u\mc{L}\big)^{-1}
\end{equation}
satisfies, as a formal power series,
\begin{equation}
    \mc{C}(u)[\chi] = \chi+2\sum_{k=1}^\infty u^k\kryl{k}\;, \label{eq:cayley_krylov_generating_function}
\end{equation}
for  \(u \in U_0 \subset \mathds{R}\), a small enough region around  \(u = 0\).
The commutator  \(\mathfrak{S}(u) \coloneqq \e[\mc{C}(u)[\chi], \mc{C}(-u)[\chi]]\) packages all odd Krylov currents into a single generating series. 
After separating the terms involving  \(\chi\), the coefficient of each odd power of  \(u\) is, up to a factor and additional path-supported term, the corresponding current operator  \(J^{(\ell)}\) of \cref{eq:current_operators}. 
It is therefore sufficient to show that this generating commutator is path-supported. 
The key step is the following Cayley-transform identity, which turns the adjoint Hamiltonian flow into conjugation by the unitary  \(V(u)\). \begin{lemma}[
    restate=CayleyTransform,
    name={Cayley transform of Hamiltonian adjoint operator},
    label={lem:cayley_transform}]
    Let  \(u \in U_0 \subset \mathds{R}\) and
    \begin{align}
        \mc{C}(u)[X] \coloneqq \big(I+u\mc{L}\big)\big(I-u\mc{L}\big)^{-1}[X]
        \label{eq:cayley_transform}
    \end{align}
    be the Cayley transform of  \(\mc{L} = \frac{i}{2}\ad_H\).
    For  \(H_{\mathrm{FFD}}\) and  \(X\) an operator in the span of the eigenmodes \(\{\psi_j\}_j \cup \{\psi^{\dagger}_j\}_j\) and zero mode \(\phi_0\), for which 
    \begin{align}
        \mc{L}[\phi_0] = 0 \mh{2mm} \text{and} \mh{2mm} \{\phi_0, \psi_j\} = \{\phi_0, \psi_j^{\dagger}\} = 0
    \end{align} 
    we have 
    \begin{align}
        \mc{C}(u)[X] = V(u)^{\dagger} X V(u) 
        \label{eq:adjointunitary}
    \end{align}
    where  \(V(u)\) is the unitary transfer operator as given in \cref{def:uto}.
\end{lemma}

\begin{remark}[Relation to the generating function of Ref.~\cite{fukai2026solvingb}]
    In our notation, the Krylov generating function of Ref.~\cite{fukai2026solvingb} is \(\Phi_G(u)=\sum_j (-iu)^j\kryl{j}\), and one has \(\mc{C}(u)[\chi]=2\,\Phi_G(iu)-\chi\).
    One could therefore work with \(\Phi_G\) directly, which would bypass the Cayley transform in the proof of \cref{thm:ffd_currents_are_paths}.
    We retain the Cayley transform because of its geometric meaning: \(\mc{C}(u)\) is the rational parametrization of the unitary flow generated by \(\mc{L}\), and \cref{lem:cayley_transform} realizes it as conjugation by the unitary transfer operator \(V(u)\).
\end{remark}

Our main argument is the following. 
We aim to characterize the family of models described in \cref{eq:ngmodels} in terms of graph theory, and we view these as generated by induced paths in the FFD model analogously to the relationship between line-graph models and induced paths of the TFIM as described in \cref{sec:tfim_seed_JW}.
Utilizing the free-fermion construction of the FFD model described in \cref{sec:ffd_solution}, we show \cref{lem:cayley_transform} in \cref{app:cayley}. 
Together with \cref{thm:krylov-restated} (\cref{app:induced-paths}) we find a relationship between the Krylov currents and commutators between anchored path-product operators with coefficients given in terms of independence polynomials of zigzag-ladder subgraphs.
This latter expression apparently has precisely the right properties such that all of the non-path-supported terms cancel, and this proves \cref{thm:ffd_currents_are_paths}, as shown in \cref{app:current-paths}.
We prove \cref{thm:ffd_current_algebra} in \cref{app:current_algebra} by an inductive argument in the Krylov order using \cref{eq:nocar}.
In \cref{sec:phase_diagram_details}, we provide more details on the numerical analysis, and in \cref{app:kpsconjecture}, we provide further discussion regarding \cref{conj:krylov_path_support_free_fermions}.
In \cref{app:induced-paths,,app:cayley,,app:structural-graph}, we give some further supporting technical results, as well as other results that may be of independent interest.

\section{Discussion}\label{sec:discuss}
We have introduced a mechanism for generating non-generically free-fermion-solvable Hamiltonians from extensions to the FFD model.
These extensions constitute a complete set of generators for the hidden free-fermion dynamical Lie algebra of the model, and they admit path-product expansions in the Hamiltonian terms.
The FFD model therefore generates a family of path-supported Hamiltonians whose solvability is inherited from its original Krylov modes, much akin to the way that the TFIM generates all line-graph models.

We have also shown that the currents carry a natural algebraic structure. 
Since the Krylov operators satisfy generalized CAR, their bilinears are quadratic Majorana operators after orthonormalisation. 
Furthermore, as the CAR are defined by an alternating Hankel-type matrix, the Krylov currents generate a Lie algebra containing all these quadratic Majorana operators, i.e., the full special orthogonal Lie algebra, and their exponentials generate the corresponding Lie group of Gaussian transformations.
The technical content of \cref{thm:ffd_current_algebra} is therefore the existence of a generating set of Gaussian transformations whose generators admit a path-product expansion in the zigzag-ladder frustration graph of the FFD model.

This gives a new perspective on non-generic free-fermion solvability. 
In the usual graph-theoretic approach, solvability is guaranteed by structural conditions on the frustration graph itself, such as the line-graph condition for linear mappings or the simplicial, claw-free condition for disguised free-fermion solutions. 
The generated family of models described by \cref{eq:ngmodels} need not satisfy those conditions. 
Indeed, the frustration graphs of these models contain many claws, as demonstrated in \cref{fig:zigzag_ladder,,fig:merged_number_plot}. 
Furthermore, these claws do not appear to be resolved by the same mechanism as the models in~\cite{fukai2025free,fendley2019free,sinha2026hidden}.
The solvability instead comes from the fact that the Hamiltonians belong to the quadratic algebra generated by the Krylov fermions of the seed model.

The construction appears to produce a family of dense models, see~\cref{fig:merged_number_plot,sec:numerical_analysis}. 
In the FFD examples that we have checked, both the number of induced-path terms in the generated Hamiltonians and the number of claws in their frustration graphs grow rapidly with the size of the seed. 
In this sense, the models we describe here are perhaps more abstract than previously known free-fermion-solvable systems: they are generally not local, and we have evidence that the number of terms can grow superpolynomially in the number of qubits for these models. 
An interesting open direction for investigation would therefore be to find models for which the locality structure is recovered by suitable choices of  \(\bs{a} \neq 0\) and  \(\bs{b}\) in \cref{eq:ngmodels}.

These results raise several natural open problems. 
The first is to understand the complexity of recognizing Hamiltonians that belong to the family we describe. 
In the line-graph case, the corresponding decomposition into path operators is efficient, since the root graph gives a direct free-fermion representation, and line graph recognition is itself efficient by, e.g., the algorithm of Ref.~\cite{degiorgi1995dynamic}. 
For the FFD currents, \cref{thm:ffd_currents_are_paths} gives existence, but the complexity of constructing the decomposition is not yet clear. 
The proof proceeds through two different descriptions of the same objects: one in terms of the Krylov path-tree adjacency matrix, and another in terms of `transfer operators' and independence polynomials. 
Equating these descriptions should imply nontrivial identities relating the induced-path-tree expansion to the independence-polynomial expansion~ \cite{bencs2018}.

A second open problem is to prove the graphical current-algebra conjecture for all chordal, claw-free seeds (\cref{conj:graphical_current_algebra}). 
Our proof for the FFD model relies on the special interval structure of the zigzag ladder. 
This makes it possible to pair non-path contributions by exchanging sub-ladder chains and to prove their cancellation by an explicit case analysis. 
We do not know of a fundamental obstruction to extending this argument to arbitrary chordal, claw-free graphs. 
The difficulty is rather that the uniform interval structure of the ladder must be replaced by the more general clique-tree representation of chordal, claw-free graphs. 
Cliques of order four already complicate the local picture, although higher-order cliques do not appear to introduce a qualitatively new obstruction. 
Branching in the clique tree also seems manageable: once two paths branch and do not rejoin, the cancellation problem effectively reduces to independent non-branching regions. 

If \cref{conj:graphical_current_algebra} holds, it would provide a systematic way to generate new non-generically free-fermion solvable Hamiltonians from any chordal, claw-free seed. 
The proposed procedure is graphical: identify a suitable chordal, claw-free frustration graph, choose a simplicial mode, construct the Krylov tower, and form the associated currents and their Lie closure.
The representation of chordal, claw-free graphs as intersection graphs of non-crossing paths in a tree~\cite{chaplick2023intersection} gives a structural route to finding and organizing such seeds. 
In this way, the conjecture would turn chordal, claw-free graph structure into an algorithmic source of non-generic free-fermion phases.

We also proposed a broader recognition principle: the essential feature underlying free-fermion solvability may be the confinement of the Krylov tower to induced-path operators, possibly after graphical cancellation relations~(\cref{conj:krylov_path_support_free_fermions}).
This would place line-graph models~\cite{Chapman2020}, chordal, claw-free models~\cite{chapman2023unified}, FFD-generated current Hamiltonians, and other non-generic examples~\cite{fukai2025free,Fendley2024,sinha2026hidden} into a common framework.
The conjectured criterion is that the Hamiltonian admits a free-fermion solution on a Krylov space if and only if that space is finite, path-supported, and preserved by the Hamiltonian adjoint action.

It would also be interesting to connect these results with recent work on FFD-guided Hamiltonian problems. 
Since the current-generated models share a common Krylov fermion basis with their seed model, they may enlarge the class of Hamiltonians whose low-energy structure can be controlled by FFD-like methods.
A natural question is whether the present construction, together with the techniques of Ref.~\cite{fukai2026solvingb}, can help close the BQP-containment question for the FFD-guided local Hamiltonian problem~\cite{Waite2025}.
More broadly, Krylov currents suggest that non-generic free-fermion solvability is not an isolated phenomenon, but a source of structured families of interacting-looking spin models with hidden Gaussian dynamics.

\begin{acknowledgments}
    This work was funded by Phasecraft. 
    We thank Joel Klassen for thought provoking conversations. 
    JR was supported by the Sydney Quantum Academy, Sydney, NSW, Australia.
\end{acknowledgments}

\onecolumngrid 

\appendix
\newpage

\section{Proof of the path support of the Krylov currents}
\label[appendix]{app:current-paths}

In this appendix we prove \cref{thm:ffd_currents_are_paths}.
\KrylovCurrents*

\subsection{Notation and Labeling}
\label{sec:notation}
In this section, we denote vertices as separate from their indices with labeling by square brackets (e.g.  \(\vv{i}\) to mean ``vertex with index \(i\)''). 
We similarly promote the indices for Hamiltonian terms to avoid overloading subscripts (e.g.  \(\hv{i}\)).
We fix an interval-graph representation and cutsets of  \(G\) via the linear order
\begin{align}
    \setcounter{MaxMatrixCols}{20}
    \begin{matrix}
        V &=& \vv{0} \coloneqq \bs{j}^* & > & \vv{1} & > & \dots & > & \vv{m-1} & > & \vv{m}\\
        \Upsilon &=& & & \kv{1} & \md & \dots & \md & \kv{m - 1} & \md & \kv{m}\\
        \Gamma &=& & & \cv{1} & \md & \dots & \md & \cv{m - 1} & \md & \cv{m}
    \end{matrix}
    \label{eq:ppmodel}
\end{align}
where we define 
\begin{align}
    \begin{cases}
        \kv{1} = \{\vv{1}, \vv{2}\} & \\
        \kv{j} = \{\vv{j - 1}, \vv{j}, \vv{j + 1} \} & j \in \{2, \dots, m - 1\} \\
        \kv{m} = \{\vv{m - 1}, \vv{m}\}
    \end{cases} \mh{5mm} \text{and} \mh{5mm}
    \begin{cases}
        \cv{1} = \{\vv{1}\} & \\
        \cv{j} = \{\vv{j - 1}, \vv{j} \} & j \in \{2, \dots, m\}
    \end{cases}
\end{align}
The \emph{simplicial vertex}  \(\boldsymbol{j}^*\) satisfies  \(N(\boldsymbol{j}^*) = \{\vv{1}\}\), and we have
\begin{align}
    \vv{i} \sim \vv{j} \quad\Longleftrightarrow\quad 0<|i-j|\leq 2
    \label{eq:app_zigzag_adjacency}
\end{align}
for all \(i\), \(j \in \{1, \dots, m\}\).
While we do not directly use the relationships between objects in \cref{eq:ppmodel}, they are perhaps helpful for keeping the graph structure in mind: the ordering on vertices in this equation is a \emph{perfect elimination ordering}, and the cliques are ``neighboring'' in their respective paths \(\Upsilon\) or \(\Gamma\) if their intersection is nonempty.

Our proof rests heavily on the following criteria for systematically identifying an induced path in the zigzag ladder graph.
\begin{proposition}[Subset-representation of induced paths] 
    \label{prop:subsetrep}
    A vertex subset \(P = \{\vv{\Pv{i}}\}_{i = 0}^{\ell}\) is an induced path in \(G\) if and only if
    \begin{enumerate}
        \item For all \(j \in \{\Pv{0} + i\}_{i = 1}^{\Pv{\ell} - \Pv{0}}\), we have \(\cv{j} \cap P \neq \varnothing\).
        \item For all \(j \in \{2, \dots, m - 1\}\), we have \(\kv{j} \nsubseteq P\).
    \end{enumerate}
    Furthermore, \(P\) is anchored iff \(\Pv{0} = 0\).
\end{proposition}
\begin{proof}
    Clearly, an induced path must satisfy both of these: it cannot be connected if 1 is not satisfied, since \(\cv{j}\) is a clique cutset, and it contains the triangle \(\kv{j}\) if 2 is not satisfied.
    Suppose a vertex subset \(P\) satisfies both criteria, then any vertex \(\vv{j} \in P \setminus \{\vv{\Pv{\ell}}\}\) has a unique neighbor in \(P\) with larger index.
    By condition 1, at least one of \(\cv{j + 2} = \{\vv{j + 1}, \vv{j + 2}\}\) must be in \(P\), and by condition 2, both cannot be in \(P\), or \(\kv{j + 1} = \{\vv{j}, \vv{j + 1}, \vv{j + 2}\} \subseteq P\).
    Thus, exactly one vertex in \(\cv{j + 2}\) belongs to \(P\), and since \(\cv{j + 2}\) is the entire set of vertices with larger index than \(j\) which are neighboring \(\vv{j}\), every such vertex \(\vv{j}\) has exactly one neighbor with larger index in \(P\).
\end{proof}

For a path \(P \in \mathcal{P}^*\) with \(|P| > 1\), we write \(\vv{P} \coloneqq \vv{P(\ell)}\) for its free endpoint, i.e. the vertex of \(P\) with the largest index.
We denote complements with an ``overline'': for  \(Y \subseteq V\) we write  \(\overline{Y} \coloneqq V {\setminus} Y\).

\subsection{Proof Strategy}

We first demonstrate the following lemma in \cref{sec:omega_reduction}, whose proof uses generating function identities.
\begin{lemma} 
    \label{lem:app_exact_product_sign}
    Let  \(u \in U_0\), where  \(U_0\) is a small region around  \(0\). 
    If
    \begin{align}
        \Omega(u) &\coloneqq \sum_{\substack{P,Q\in\mc{P}^* \setminus \{\{\bs{j}^*\}\}\\
        |P|+|Q|\ \mrm{odd}}} \Omega_{(P, Q)}(u) 
        \label{eq:app_structural_series} 
    \end{align}
    admits a path-product expansion, where
    \begin{align}
        \Omega_{(P, Q)}(u) &\coloneqq i^{|P|-|Q|}u^{|P|+|Q|-2} P_{G\setminus N[P]}(u^2) P_{G\setminus N[Q]}(u^2)[h_P,h_Q],\label{eq:app_structural_series_terms}
    \end{align}
    then \cref{thm:ffd_currents_are_paths} holds.
\end{lemma}
Note that each term \(\Omega_{(P, Q)}(u)\) is symmetric in the anchored induced paths \(P\) and \(Q\), though \(\Omega_{(P, Q)}(u) = 0\) unless \(h_P h_Q = -h_Q h_P\).
We say that a pair of such paths appearing in the sum in \cref{eq:app_structural_series} satisfies the \emph{contribution conditions}.  
\begin{definition}[Contribution Conditions]
    A pair of anchored induced paths \(P\), \(Q \in \mc{P}^* \setminus \{\{\bs{j}^*\}\}\) satisfies the contribution conditions if
    \begin{align}
        \begin{cases}
            h_P h_Q = -h_Q h_P & \\
            |P| + |Q| \text{ is odd} &
        \end{cases}.
        \label{eq:cc}
    \end{align} 
\end{definition}
In particular, we note the following
\begin{observation}
    A term \(\Omega_{(P, Q)}(u)\) is nonzero and included in the sum if and only if \((P, Q)\) satisfies the contribution conditions.
\end{observation}
Necessity holds by definition, and sufficiency follows from the fact that if \((P, Q)\) satisfies the contribution conditions, then it evaluates to a nonzero polynomial multiplied by a nonzero operator, and so is necessarily nonzero.

If \(L \coloneqq P \oplus Q \notin \mc{P}\) for a given pair \((P, Q)\) satisfying the contribution conditions, then \(L\) fails to satisfy one of the conditions in \cref{prop:subsetrep}.
We first show in \cref{sec:pairwise_cancellations} that if \(L\) does not satisfy condition 1, then there is a unique additional such pair \((P', Q')\) for which \(\Omega_{(P', Q')}(u) = -\Omega_{(P, Q)}(u)\), and these terms cancel pairwise.
We then show in \cref{sec:recursive_cancellations} that if \(L\) satisfies condition 1 but does not satisfy condition 2, then there is a unique collection of terms \(\mc{C} \coloneqq \{(P_i, Q_i)\}_i\), that all do not satisfy condition 2, and include  \((P, Q)\), for which
\begin{align}
    \sum_i \Omega_{(P_i, Q_i)}(u) = 0.
\end{align}
This leaves only the terms for which \(L\) satisfies the conditions of \cref{prop:subsetrep}, and we finally apply \cref{lem:app_exact_product_sign} to prove the theorem.

In what follows, we may drop the explicit \(u\)-dependence for notational convenience.
Additionally, we define the following useful terminology.
\begin{definition}[Coloring of the symmetric difference]
    \label{def:coloring}
    For a pair \((P, Q)\) with \(L \coloneqq P \oplus Q\), each vertex \(\vv{j} \in L\) is either \emph{\(P\)-colored} (if \(\vv{j} \in P {\setminus} Q\)) or \emph{\(Q\)-colored} (if \(\vv{j} \in Q {\setminus} P\)); we call this partition of \(L\) the \emph{coloring} of \(L\). 
    Note this is not a proper coloring in the graph-theoretic sense, since two adjacent vertices of \(L\) may share a color.
\end{definition}

\subsection{Reduction to \texorpdfstring{ \(\Omega(u)\)}{omegau}}\label{sec:omega_reduction}

\subsubsection{The current generating function}

Recall that  \(\kryl{r}=\mc{L}^r[\chi]\), where  \(\mc{L}[X]=\frac{i}{2}[H,X]\).
The Cayley transform therefore has the formal expansion
\begin{equation}
    \mc{C}(u)[\chi] = (I+u\mc{L})(I-u\mc{L})^{-1}[\chi] = \chi+2\sum_{r\geq 1}u^r\kryl{r}.
    \label{eq:app_cayley_series_corrected}
\end{equation}
Define
\begin{equation}
    \mathfrak{S}(u) \coloneqq [\mc{C}(u)[\chi],\mc{C}(-u)[\chi]].
    \label{eq:app_omega_corrected}
\end{equation}
Clearly \(\mathfrak{S}(u)\) is odd in \(u\).
For odd  \(\ell\), the coefficient of  \(u^\ell\) in this commutator is
\begin{align}
    [u^\ell]\mathfrak{S}(u) &=-4[\chi,\kryl{\ell}] -4\sum_{j=1}^{\ell-1}(-1)^j [\kryl{j},\kryl{\ell-j}].
    \label{eq:app_omega_coefficient_intermediate}
\end{align}
On the other hand, because  \(\ell\) is odd,
\begin{align}
    -i J^{(\ell)} = \sum_{j=0}^{\ell}(-1)^j [\kryl{j},\kryl{\ell-j}] =2[\chi,\kryl{\ell}] +\sum_{j=1}^{\ell-1}(-1)^j [\kryl{j},\kryl{\ell-j}].
\end{align}
Using \cref{eq:current_operators}, we obtain the exact generating identity
\begin{equation} \label{eq:app_omega_currents_corrected}
    \mathfrak{S}(u) = 4\sum_{\substack{\ell\geq1\\ \ell\ \mrm{odd}}} u^{\ell} \left([\chi, \eta^{(\ell )}] + i J^{(\ell)}\right)
\end{equation}
By \cref{thm:krylov-restated}, every  \(\kryl{\ell}\) is a linear combination of anchored path products.
Since  \(\chi\) anticommutes only with \(\hv{1}\), commuting an anchored path product with  \(\chi\) removes the simplicial factor, up to a scalar.
Hence every  \([\chi,\kryl{\ell}]\) admits a path-product expansion in  \(G\).
The idea is to prove that  \(\mathfrak{S}(u)\) is path-supported, and this proves that the Krylov currents are path-supported by equating coefficients.

\subsubsection{Reduction to a two-path identity}
In the following we assume that the roots of the independence polynomial of the FFD model are distinct. 
As in Ref.~\cite{fendley2019free}, if for some vertex weights  \(\bs{b}\) some of the roots are not distinct, one may change infinitesimally as  \(\bs{b} \mapsto \bs{b} + \bs{\delta b}\) so that the roots become distinct, and then take the limit  \(\bs{\delta b} \to \bs{0}\) after the following analysis.

We first use a result of ref.~\cite{fukai2026solvingb} to see that the simplicial mode \(\chi\) satisfies the conditions of \cref{lem:cayley_transform}.

\begin{proposition}[Theorem 3.14 and Corollary 3.15, \cite{fukai2026solvingb}]
    \label{prop:zero_mode}
    For \(G\) a connected ECF graph, we have for a fixed simplicial mode \(\chi\)
    \begin{equation}
        \chi = \phi_0 + \sum_{j=1}^{\alpha(G)} \frac{C_j}{u_j} \e(\psi_j^\dagger + \psi_j), \label{eq:chi_eigenmodes}
    \end{equation}
    where \(\phi_0\) is the zero mode, for which
    \begin{equation}
        [H, \phi_0]=0 \qquad \text{and} \qquad \{\psi_j, \phi_0\} = \{\psi_j^\dagger, \phi_0\} = 0 \qquad \text{for all }j=1,\ldots,\alpha(G).
    \end{equation}
    Additionally, we have  \(\phi_0 = 0\) if  \(\alpha(G\setminus K) < \alpha(G)\).
\end{proposition}

\begin{corollary}\label{cor:zero_mode_vanishes}
    If every maximum independent set of  \(G\) contains an element of the simplicial clique  \(K\), then  \(\phi_0 = 0\).
\end{corollary}
\begin{proof}
    In that case we have  \(\alpha(G\setminus K) < \alpha(G)\), and the claim is the last statement of \cref{prop:zero_mode}.
\end{proof}

\begin{corollary}\label{prop:l1span}
    Assume that the roots  \(\{u_j\}_j\) of  \(P_G(-u^2)\) are distinct, and set  \(\varepsilon_j = u_j^{-1}\).
    The Krylov operators are linear combinations of the fermion modes  \(\{\psi_j, \psi_j^\dagger\}_j\) (except that the simplicial mode might have a non-zero zero mode). 
    More specifically,
    \begin{equation}
        \kryl{k} = \mc{L}\e[\kryl{k-1}] = \delta_{k,0}\,\phi_0 + \frac{i^k}{4} \sum_{j=1}^{\alpha(G)} N_j L_{0j}\, \varepsilon_j^k \e((-1)^k \psi_j^{\dagger} + \psi_j)\;,
    \end{equation}
    where  \(N_j\) is the normalisation of \cref{def:incognito} and  \(L_{0j}\) is the constant coefficient of the  \(j\)th Lagrange polynomial, as in \cref{sec:phase_diagram_details}.
    In particular,  \(\mc{L}[\chi] = -\frac{i}{4} \sum_{j} \varepsilon_j N_jL_{0j} (\psi_j^{\dagger} - \psi_j)\).
\end{corollary}
\begin{proof}
    By \cref{prop:zero_mode},  \(\chi = \phi_0 + \sum_j c_j (\psi_j^\dagger + \psi_j)\) with  \(c_j = C_j/u_j\), which in the normalisation of \cref{def:incognito} equals  \(N_jL_{0j}/4\).
    Since  \([H, \phi_0] = 0\), \cref{eq:eigenmodes} gives  \(\mc{L}[\phi_0] = 0\),  \(\mc{L}[\psi_j] = i\varepsilon_j\psi_j\) and  \(\mc{L}[\psi_j^\dagger] = -i\varepsilon_j\psi_j^\dagger\).
    Applying  \(\mc{L}^k\) to the expansion of  \(\chi\) gives the claim.
    The claim that the coefficient is proportional to the Lagrange coefficient follows straightforwardly from the polynomial expansions of \(\psi_j\) and \(\psi_j^{\dagger}\), together with the defining property of this coefficient as an element of the inverse Vandermonde matrix.
\end{proof}

Using  \(V(u)=P_G(u^2)^{-1/2}T(iu)\) in \cref{lem:cayley_transform} thus gives
\begin{equation}
    \mc{C}(u)[\chi] = P_G(u^2)^{-1}T(-iu)\chi T(iu).
    \label{eq:app_cayley_transfer_corrected}
\end{equation}
Consequently,
\begin{equation}
    P_G(u^2)^2\mathfrak{S}(u) = [T(-iu)\chi T(iu),T(iu)\chi T(-iu)].
    \label{eq:app_omega_transfer_corrected}
\end{equation}
To obtain a path-product expansion of \(\mc{C}(u)[\chi]\), we apply the following additional results, restated in our notation
\begin{proposition}[Theorems 3.8 and 4.1, \cite{fukai2026solvingb}]\label{prop:TchiT}
    For \(G\) a connected ECF graph, we have for a fixed simplicial mode \(\chi\)
    \begin{align}
        &P_G(-u^2)^{-1}T(-u) \chi T(u) \overset{4.1}{=} 2\left[\sum_{k = 0}^{\infty} (-i u)^k \eta^{(k)} \right] - \chi \\ 
        &\mbox{\hspace{5mm}}\overset{3.8}{=} \left(2 P_G(-u^2)^{-1} P_{G\setminus K}(-u^2) - 1 \right) \chi + 2\sum_{P \in \mathcal{P}^* \setminus \{\{\bs{j}^*\}\}} (-u)^{|P| - 1} P_G(-u^2)^{-1} P_{G\setminus N[P]}(-u^2) \hv{P} 
    \end{align}
\end{proposition}
We apply \cref{prop:TchiT} to \cref{eq:app_omega_transfer_corrected} with its formal variable replaced by  \(iu\).
Since  \(-(iu)^2=u^2\), every independence-polynomial argument becomes  \(u^2\); the substitution is made in formal power series and requires no analytic continuation.
Separating terms according to parity then gives
\begin{equation}
    P_G(u^2)^2\mathfrak{S}(u) =4\e( \sum_{\substack{P\in\mc{P}^*\\ |P|>1, \ |P|\ \mrm{even}}} (iu)^{|P|-1} \bigl(2P_{G\setminus K}(u^2)-P_G(u^2)\bigr) P_{G\setminus N[P]}(u^2) [\chi, \hv{P}]) -4\,\Omega(u),
    \label{eq:app_transfer_reduction_corrected}
\end{equation}
where
\begin{align}
    \Omega(u) \coloneqq \sum_{\substack{P, Q \in \mathcal{P}^* \setminus \{\{\bs{j}^*\}\} \\ |P| + |Q|\ \mrm{odd}}} &i^{|P|-|Q|}u^{|P|+|Q|-2} P_{G\setminus N[P]}(u^2) P_{G\setminus N[Q]}(u^2) [\hv{P},\hv{Q}].
\end{align}
The first term in the parentheses of \cref{eq:app_transfer_reduction_corrected} is path-supported, because  \([\chi,\hv{P}]\) is proportional to the path product obtained by deleting \(\bs{j}^*\) from  \(P\).

\begin{remark}
    As mentioned in the main text, the results Ref.~\cite[Theorems~3.8 and~4.1]{fukai2026solvingb} already give a direct route to the path-product form of  \(P_G(u^2)\,\mc{C}(u)[\chi]\) in terms of  \(\Phi_G\).
    We nevertheless go through the Cayley transform and the transfer operator, so that \cref{eq:app_omega_transfer_corrected} is manifestly a statement about conjugation of  \(\chi\) by a unitary, in keeping with the aforementioned geometric interpretation of the Cayley transform.
\end{remark}

\subsubsection{Coefficient Comparison}

We can now prove \cref{lem:app_exact_product_sign}, completing the first step: If  \(\Omega(u)\) is path supported, then  \(P_G(u^2)^2\mathfrak{S}(u)\) is path supported, that is
\begin{equation}
    P_G(u^2)^2\mathfrak{S}(u) \in \operatorname{span}_{\mbb{C}[[u]]} \{h_R:R\text{ is an induced path of }G\}\;.
    \label{eq:app_scaled_omega_path_span}
\end{equation}
The scalar series  \(P_G(u^2)^2\) has the constant term  \(1\) and is therefore invertible in  \(\mbb{R}[[u]]\), for  \(u \in U_0\), where  \(U_0\) was a small region around  \(0\).
Multiplying \cref{eq:app_scaled_omega_path_span} by this inverse gives
\begin{equation}
    \mathfrak{S}(u) \in \operatorname{span}_{\mbb{C}[[u]]} \{h_R:R\text{ is an induced path of }G\}.
    \label{eq:app_omega_path_span_final}
\end{equation}
Comparing the coefficient of  \(u^\ell\) in \cref{eq:app_omega_currents_corrected}, and recalling that  \([\chi,\kryl{\ell}]\) is already path-supported, yields
\begin{equation}
    J^{(\ell)} \in \operatorname{span}_{\mbb{C}} \{h_R:R\text{ is an induced path of }G\}
\end{equation}
for every odd  \(\ell\).
This proves \cref{thm:ffd_currents_are_paths}, and we continue proving the cancelling relations when the conditions in \cref{prop:subsetrep} are not fulfilled by the symmetric difference \(P \oplus Q\).

\subsection{Cancellations}

\subsubsection{Pairwise Cancellation}\label{sec:pairwise_cancellations}

\begin{proposition}
    \label{prop:symdiff_cond1}
    Let \(P\), \(Q \in \mc{P}\) be induced paths with common endpoints \(\vv{\Pv{0}} = \vv{\Qv{0}}\) and \(\vv{\Pv{|P|-1}} = \vv{\Qv{|Q|-1}}\), symmetric difference \(L \coloneqq P \oplus Q\), and \(N \coloneqq |P \cap Q| - 1\).
    If \(L\) satisfies condition 1 of \cref{prop:subsetrep}, then we must have that \(P \cap Q\) is an independent set in \(G\).
    Defining indices \(\ell_0 < \ell_1 < \dots < \ell_N\) such that \(\{\vv{\ell_i}\}_{i = 0}^N \coloneqq P \cap Q\), we have
    \begin{align}
        L = \bigcup_{i = 0}^{N - 1} L_i \mh{10mm} \text{with} \mh{10mm} \begin{cases}
            L_i = \{\vv{\ell_i + j}\}_{j = 1}^{\ell_{i + 1} - \ell_i - 1} & (a) \\ 
            \ell_{i} = \Pv{0} + i + \sum_{j < i} |L_j| & (b) \\
            \ell_{N} = \Pv{|P|-1} = \Qv{|Q|-1} & (c)
        \end{cases}.
        \label{eq:symdiff_structure}
    \end{align}
    We must additionally have that \(|L_i| > 1\) for all \(i \in \{0, \dots, N - 1\}\).
    Letting \(\{U, W\} = \{P{\setminus} Q, Q {\setminus} P\}\) denote the two color classes of \(L\) (\cref{def:coloring}), then if \(\vv{\ell_0 + 1} \in U\), we have for all  \(i \in \e{0, \ldots, N-1}\) that
    \begin{align}
        \begin{cases}
            \vv{(\ell_i + 2j - 1) + (\bigoplus_{k = 0}^{i - 1} |L_k|)} \in U & \\
            \vv{(\ell_i + 2j - 1) + (1 \oplus \bigoplus_{k = 0}^{i - 1} |L_k|)} \in W &
        \end{cases}
        \label{eq:symdiff_paths}
    \end{align}
    for all \(j \geq 1\) such that the corresponding vertex index is less than \(\ell_{i + 1}\).
\end{proposition}
\begin{proof}
    First, note that if \(P \cap Q\) is not an independent set, then we have a vertex \(\vv{j} \in P \cap Q\) such that \(\{\vv{j + 1}, \vv{j + 2}\} \cap (P \cap Q) \neq \varnothing\).
    We cannot have \(\vv{j + 1} \in P \cap Q\), as then \(L\) would not satisfy condition 1 for \(\cv{j + 1}\).
    Thus, \(\{\vv{j}, \vv{j + 2}\} \subseteq P \cap Q\) with \(\vv{j + 1} \in L\).
    However, if \(\vv{j + 1} \in P {\setminus} Q\), then \(P\) fails condition 2 of \cref{prop:subsetrep} for \(\kv{j + 1}\), and if \(\vv{j + 1} \in Q {\setminus} P\), then \(Q\) fails condition 2 for \(\kv{j + 1}\).
    This gives that \(P \cap Q\) must be an independent set.
    
    If there exists a vertex \(\vv{j} \notin P \cup Q\) with \(\ell_0 < j < \ell_N\), then both \(\vv{j - 1}\), \(\vv{j + 1} \in L\), or \(L\) fails condition 1 for either \(\cv{j}\) or \(\cv{j + 1}\).
    If \(\vv{j - 1} \in P {\setminus} Q\), then \(Q\) fails condition 1 for \(\cv{j}\).
    Similarly, if \(\vv{j - 1} \in Q {\setminus} P\), then \(P\) fails condition 1 for \(\cv{j}\).
    Thus, if any vertex \(\vv{j} \notin L \) with \(\ell_0 \leq j \leq \ell_N\), then \(j \in \{\ell_i\}_{i = 0}^N\). 
    Defining the \(\{L_i\}_{i = 0}^{N - 1}\) by Eq.~(\ref{eq:symdiff_structure}.a), we see that \cref{eq:symdiff_structure} follows immediately, and we have \(|L_i| > 1\) for all \(i \in \{0, \dots, N - 1\}\) as \(P \cap Q\) is an independent set.
    We next show \cref{eq:symdiff_paths}.
    
    Without loss of generality, suppose \(U = P {\setminus} Q\).
    If there exists a pair \(\{\vv{j}, \vv{j + 1}\} \subseteq U\), then \(Q\) does not satisfy condition 1 of \cref{prop:subsetrep}.
    By an identical argument, we cannot have a pair \(\{\vv{j}, \vv{j + 1}\} \subseteq W\). 
    This gives that, for any \(L_i\), the colors alternate along \(L_i\) with increasing index. 
    For a pair \(\{L_i, L_{i + 1}\}\), suppose without loss of generality that \(\vv{\ell_{i + 1} - 1} \in U\).
    If \(\vv{\ell_{i + 1} + 1} \in U\), then \(P\) does not satisfy condition 2 of \cref{prop:subsetrep}. 
    Thus, we must have \(\vv{\ell_{i + 1} + 1} \in W\), and an identical argument follows upon exchange of \(U\) and \(W\).
    From these considerations, \cref{eq:symdiff_paths} immediately follows.
\end{proof}
\begin{corollary}
    \label{cor:PQ_anticommuting}
    Let \(P\), \(Q\), and \(L\) satisfy the conditions of \cref{prop:symdiff_cond1}, then \(\hv{P}\hv{Q} = (-1)^{|L| + 1}\hv{Q}\hv{P}\).
\end{corollary}
\begin{proof}
    Suppose without loss of generality that \(\vv{\ell_0 + 1} \in P\).
    We prove the corollary by induction on \(N\).
    For \(N = 1\), \cref{eq:symdiff_paths} gives that every vertex in \(P \setminus Q\) is of the form \(\vv{\ell_0 + 2j - 1}\) for \(0 < j < (\ell_1 - \ell_0 + 1)/2\).
    We have
    \begin{align}
        N_{Q}(\vv{\ell_0 + 2j - 1}) &= \{\vv{\ell_0 + 2j - 2}, \vv{\ell_0 + 2j}\} \cup (\{\vv{\ell_0}, \vv{\ell_1}\} \cap \{\vv{\ell_0 + 2j - 3},\vv{\ell_0 + 2j + 1}\}) \\
        &= \{\vv{\ell_0 + 2j - 2}, \vv{\ell_0 + 2j}\} \cup (\{\vv{\ell_1}\} \cap \{\vv{\ell_0 + 2j + 1}\}).
    \end{align}
    The first line follows from the fact that the vertices with index differing from \(\ell_0 + 2j - 1\) by one are in \(Q\) (possibly in \(P \cap Q\)) for all such \(j\), and vertices with index differing by two are only in \(Q\) if they are boundaries of \(L_0\) with indices in \(\{\ell_0, \ell_1\}\).
    In the second line, we used the fact that \(\vv{\ell_0}\) cannot be in the intersection with a set of vertices whose indices differ from \(\ell_0\) by an odd integer as well as the fact that \(\ell_0 + 2j - 3 \neq \ell_1\) if \(\ell_0 + 2j - 1 < \ell_1\).
    This gives that the neighborhood in \(Q\) of \(\vv{\ell_0 + 2j - 1} \in P {\setminus} Q\) consists of two vertices unless \(\ell_1 = \ell_0 + 2j + 1\), which is only possible if \(|L_0| = \ell_1 - \ell_0 - 1\) is even.
    Since \(P \cap Q = \{\ell_0, \ell_1\}\) consists of the endpoints of \(Q\), we have \(\hv{\ell_0} \hv{\ell_1} \hv{Q} = \hv{Q} \hv{\ell_0} \hv{\ell_1}\), and
    \begin{align}
        \scomm{\hv{P}}{\hv{Q}} = \scomm{\hv{P {\setminus} Q}}{\hv{Q}} = \prod_{0 < j < (\ell_1 - \ell_0 + 1)/2} \scomm{\hv{\ell_0 + 2j - 1}}{\hv{Q}} = (-1)^{|L_0| + 1}
    \end{align}
    proving the base case.
    
    Next, we assume the corollary holds for all orders \(|P \cap Q| - 1 < N\), and we define \(P' = P \cap \{\vv{j}\}_{j = \ell_1}^{\ell_N}\), \(Q' = Q \cap \{\vv{j}\}_{j = \ell_1}^{\ell_N}\), and \(L' = P' \oplus Q'\).
    We have
    \begin{align}
        &\scomm{\hv{P}}{\hv{Q}} = \scomm{\hv{\ell_0}\hv{P \cap L_0} \hv{P'}}{\hv{\ell_0}\hv{Q \cap L_0} \hv{Q'}} \\
        &\mh{5mm} =\scomm{\hv{\ell_0}\hv{P \cap L_0}}{\hv{\ell_0}\hv{Q \cap L_0}} \underbrace{\scomm{\hv{\ell_0}\hv{P \cap L_0}}{\hv{Q'}} \scomm{\hv{P'}}{\hv{\ell_0}\hv{Q \cap L_0}}}_{=-1} \underbrace{\scomm{\hv{P'}}{\hv{Q'}}}_{=(-1)^{|L'| + 1}} \\
        &\mh{5mm} =(-1)^{|L'|} \scomm{\hv{\ell_0}\hv{P \cap L_0}}{\hv{\ell_0}\hv{Q \cap L_0}} \scomm{\hv{\ell_1}}{\hv{\ell_0}\hv{Q \cap L_0}} \scomm{\hv{\ell_0}\hv{P \cap L_0}}{\hv{\ell_1}} \\
        &\mh{5mm} =(-1)^{|L'|} \scomm{\hv{\ell_0}\hv{P \cap L_0} \hv{\ell_1}}{\hv{\ell_0}\hv{Q \cap L_0} \hv{\ell_1}} \\
        &\mh{5mm} =(-1)^{|L'| + |L_0| + 1} = (-1)^{|L| + 1}.
    \end{align}
    In the second line, we used the inductive hypothesis for \(P'\) and \(Q'\) together with the fact that precisely one of the middle two factors is \(-1\) by \cref{eq:symdiff_paths} and \cref{eq:app_zigzag_adjacency}.
    In the third line, we inserted two factors of \(-1\), and then to obtain the fourth line, we used the product rule together with the fact that \(\scomm{\hv{\ell_1}}{\hv{\ell_1}} = 1\).
    We again applied the base-case statement in the fourth line to obtain the fifth line, and the corollary follows.  
\end{proof}

Let  \(M\) be the set of (sorted) pairs \((P, Q)\) satisfying the contribution conditions but for which \(P \oplus Q\) fails to satisfy condition 1 of \cref{prop:subsetrep}. 
We use the above to show that we can partition  \(M\) into two sets  \(M_1, M_2\) and define a mapping  \(f: M \to M\) such that  \(f\vert_{M_1} \to M_2\) is bijective and  \(\Omega_x = -\Omega_{f(x)}\) for all  \(x \in M_1\). 
This allows us to pair the elements in  \(M\) such that their contributions in \cref{eq:app_structural_series_terms} cancel.

Let  \((P, Q) \in M\) and put  \((P, Q)\) into  \(M_1\).
Let \(P_0\) be the maximal connected-path component of \(P \cap Q\) containing \(\bs{j}^*\), and let \(\vv{\ell_0}\) be the vertex with maximal index belonging to \(P_0\).
Further, suppose \(\ell_N\) is the minimal index such that \(P \oplus Q\) fails to satisfy condition 1 for \(\cv{\ell_N + 1}\), in particular,  \(\vv{\ell_N} \notin P \oplus Q\).
By the minimality of \(\ell_N\), we must have that \(\vv{\ell_N - 1} \in P \oplus Q\), and without loss of generality, suppose it is in \(P {\setminus} Q\).
If \(\vv{\ell_N} \notin P \cap Q\), then  \(\vv{\ell_N} \notin Q\) as  \(\vv{\ell_N} \notin P \oplus Q\), and then \(Q\) fails condition 1 of \cref{prop:subsetrep} for \(\cv{\ell_N}\).
Thus, we have \(\vv{\ell_N} \in P \cap Q\).
Letting \(\widetilde{P} \coloneqq P \cap \{\vv{j}\}_{j = \ell_0}^{\ell_N}\), \(\widetilde{Q} \coloneqq Q \cap \{\vv{j}\}_{j = \ell_0}^{\ell_N}\), and \(\widetilde{L} = \widetilde{P} \oplus \widetilde{Q}\), we have that \(\widetilde{P}\), \(\widetilde{Q}\), and \(\widetilde{L}\) satisfy the conditions of \cref{prop:symdiff_cond1}.
In particular, \(\widetilde{L}\) satisfies condition 1 of \cref{prop:subsetrep} again by the minimality of \(\ell_N\).

We define
\begin{align}
    P' &\coloneqq (P {\setminus} \widetilde{P}) \cup \widetilde{Q} \\
    Q' &\coloneqq (Q {\setminus} \widetilde{Q}) \cup \widetilde{P},
\end{align}
and we show that \((P', Q')\) satisfies the contribution conditions.
Since \(\ell_0 \geq 1\), we have \(\bs{j}^* \in P {\setminus} \widetilde{P}\), and \(\bs{j}^* \in Q {\setminus} \widetilde{Q}\). 
Thus, \(\bs{j}^* \in P'\), and \(\bs{j}^* \in Q'\). 

We next show that \(P'\) and \(Q'\) are induced paths.
If \(P'\) fails condition 1 of \cref{prop:subsetrep}, then it must do so for either \(\cv{\ell_0}\), or \(\cv{\ell_N + 1}\).
Otherwise, this would imply that either \(P\) or \(\widetilde{Q}\) fails this condition.
However, we have that \(\vv{\ell_0} \in \cv{\ell_0} \cap \widetilde{Q}\) and \(\vv{\ell_N} \in \cv{\ell_N + 1} \cap \widetilde{Q}\), so \(P'\) satisfies condition 1.
An identical argument holds for \(Q'\).

If \(P'\) fails condition 2 of \cref{prop:subsetrep}, then it must again do so for a clique among \(\kv{\ell_0 - 1}\), \(\kv{\ell_0}\), \(\kv{\ell_N}\), and \(\kv{\ell_N + 1}\), or either \(P\) or \(\widetilde{Q}\) also fails condition 2.
If \(\vv{\ell_0 - 1} \in P \cap Q\), then either \(P\) or \(Q\) fails condition 2 for \(\kv{\ell_0}\). 
Thus, \(\vv{\ell_0 - 1} \notin P \cup Q\) by the assumption on \(\vv{\ell_0}\), so \(\vv{\ell_0 - 1} \notin P'\).
Thus, \(P'\) satisfies condition 2 for \(\kv{\ell_0 - 1}\) and \(\kv{\ell_0}\).
Similarly, if \(\vv{\ell_N + 1} \in P \cap Q\), then either \(P\) or \(Q\) fails condition 2.
Thus, \(\vv{\ell_N + 1} \notin P \cup Q\) since \(\vv{\ell_N + 1} \notin P \oplus Q\) by the assumption on \(\ell_N\). 
Thus, \(P'\) satisfies condition 2 for \(\kv{\ell_N}\) and \(\kv{\ell_N + 1}\).
An identical argument holds for \(Q'\), and together these give \(P'\), \(Q' \in \mc{P}^*\).

We have
\begin{align}
    \scomm{\hv{P'}}{\hv{Q'}} &= \scomm{\hv{P {\setminus }\widetilde{P}}\hv{\widetilde{Q}}}{\hv{Q {\setminus} \widetilde{Q}} \hv{\widetilde{P}}} \\
    &= \scomm{\hv{P {\setminus} \widetilde{P}}}{\hv{Q {\setminus} \widetilde{Q}}}\scomm{\hv{P {\setminus} \widetilde{P}}}{\hv{\widetilde{P}}} \scomm{\hv{\widetilde{Q}}}{\hv{Q {\setminus} \widetilde{Q}}} \scomm{\hv{\widetilde{Q}}}{\hv{\widetilde{P}}} \\
    &= \scomm{\hv{P {\setminus} \widetilde{P}}}{\hv{Q {\setminus} \widetilde{Q}}} \scomm{\hv{Q {\setminus} \widetilde{Q}}}{\hv{\widetilde{P}}} \scomm{\hv{\widetilde{Q}}}{\hv{P {\setminus} \widetilde{P}}} \scomm{\hv{\widetilde{Q}}}{\hv{\widetilde{P}}} \\
    &= \scomm{\hv{P}}{\hv{Q}} = -1.
\end{align}
where, in the third line, we used the fact that \(P = (P {\setminus} \widetilde{P}) \cup \widetilde{P} \in \mc{P}^*\) and \(P' = (P {\setminus} \widetilde{P}) \cup \widetilde{Q} \in \mc{P}^*\) with \(\widetilde{P}\), \(\widetilde{Q} \in \mc{P}\), so \(\scomm{\hv{P{\setminus} \widetilde{P}}}{\hv{\widetilde{P}}} = \scomm{\hv{P {\setminus} \widetilde{P}}}{\hv{\widetilde{Q}}}\).
Similarly \(\scomm{\hv{Q{\setminus} \widetilde{Q}}}{\hv{\widetilde{Q}}} = \scomm{\hv{Q {\setminus} \widetilde{Q}}}{\hv{\widetilde{P}}}\).
We then apply the product rule and the fact that \((P, Q)\) satisfies the contribution conditions.
Additionally, we have
\begin{align}
    |P'| + |Q'| = (|P| - |\widetilde{P}| + |\widetilde{Q}|) + (|Q| - |\widetilde{Q}| + |\widetilde{P}|) = |P| + |Q| = 1 \text{ (mod 2)}
\end{align}
so \((P', Q')\) satisfies the contribution conditions. 
Furthermore, as  \(P' \oplus Q' = P \oplus Q\),  \(P' \oplus Q'\) fails condition 1 in \cref{prop:subsetrep} with the same  \(\vv{\ell_0}\) and  \(\vv{\ell_N}\), and with analogous constructions, we get  \(\widetilde{P'} = \widetilde{Q}\) and  \(\widetilde{Q'} = \widetilde{P}\), which gives  \(P'' = P\) and  \(Q'' = Q\).

Next, we calculate the  \(\Omega\) contribution of  \((P', Q')\). 
We have
\begin{align}
    i^{|P'| - |Q'|} = i^{(|P| - |\widetilde{P}| + |\widetilde{Q}|) - (|Q| - |\widetilde{Q}| + |\widetilde{P}|)} = (-1)^{|\widetilde{P}| + |\widetilde{Q}|} i^{|P| - |Q|} = (-1)^{|\widetilde{L}|} i^{|P| - |Q|} 
\end{align}
and
\begin{align}
    [\hv{P'}, \hv{Q'}] &= 2\left(\hv{P_0 {\setminus} \{\vv{\ell_0}\}}\hv{\widetilde{Q}}\hv{P {\setminus} (\widetilde{P} \cup P_0)}\right) \left(\hv{P_0 {\setminus} \{\vv{\ell_0}\}}\hv{\widetilde{P}}\hv{Q {\setminus} (\widetilde{Q} \cup P_0)}\right) \\
    &= (-1)^{|\widetilde{L}| + 1} 2\left(\hv{P_0 {\setminus} \{\vv{\ell_0}\}}\hv{\widetilde{P}}\hv{P {\setminus} (\widetilde{P} \cup P_0)}\right) \left(\hv{P_0 {\setminus} \{\vv{\ell_0}\}}\hv{\widetilde{Q}}\hv{Q {\setminus} (\widetilde{Q} \cup P_0)}\right) \\
    &= (-1)^{|\widetilde{L}| + 1} [\hv{P}, \hv{Q}].
\end{align}
Since \(\vv{P'} = \vv{P}\) we have \(G{\setminus} N[P'] = G{\setminus} N[P]\), and similarly for \(Q\).
Therefore
\begin{align}
    \Omega_{(P', Q')} &= i^{|P'| - |Q'|} u^{|P'| + |Q'| - 2} P_{G{\setminus} N[P']} P_{G {\setminus} N[Q']} [\hv{P'}, \hv{Q'}] \\
    &= ((-1)^{|\widetilde{L}|} i^{|P| - |Q|}) u^{|P| + |Q| - 2} P_{G{\setminus} N[P]} P_{G {\setminus} N[Q]} ((-1)^{|\widetilde{L}| + 1}[\hv{P}, \hv{Q}]) \\
    &= -i^{|P| - |Q|} u^{|P| + |Q| - 2} P_{G{\setminus} N[P]} P_{G {\setminus} N[Q]} [\hv{P}, \hv{Q}] = -\Omega_{(P, Q)}\;,
\end{align}
which allows us to put  \((P', Q')\) into  \(M_2\), and we define  \(f((P, Q)) \coloneqq (P', Q')\). 
The calculation above also shows that  \(x \neq f(x)\) for  \(x \in M\): assuming the contrary, we have, because of the symmetry of  \(\Omega\) in the path-tuples, that  \(\Omega_{f(x)} = \Omega_x = -\Omega_{f(x)}\), in contrast to  \(f(x)\) fulfilling the contribution conditions.

It remains to show that this partitioning process --- that is, we initialise  \(M' = M\) and \emph{take} sets  \(\e{x, f(x)}\) \emph{out} of  \(M'\) one-by-one and put  \(x\) into  \(M_1\) and  \(f(x)\) into  \(M_2\) --- is well-defined and that  \(f\) is bijective. 
Specifically, we have to show that if we pick an  \(x\) out of  \(M'\) that  \(f(x) \in M'\) (this is clear for the first  \(x\) we pick, but not for the ones afterwards as  \(M'\) becomes smaller) and that  \(f\) is injective (it is surjective per construction). 
Let  \(x \in M\). 
It holds that  \(f(f(x)) = x\), per construction as shown above, and therefore  \(f\) is bijective as a mapping from  \(M\) to  \(M\). 
Now, during the partitioning process, assume we pick an  \(x \in M'\) such that  \(f(x) \notin M'\). 
Then we must have already put  \(f(x)\) either into  \(M_1\) or into  \(M_2\). 
The latter would be a direct violation to  \(f\) being bijective, so let us assume  \(f(x) \in M_1\), but then we already put  \(x = f(f(x))\) into  \(M_2\), and therefore we could not have picked  \(x \in M'\) in the first place; contradiction.

This shows that partitioning process is well-defined, and that  \(f\) is bijective, and the conclusion follows.

\subsubsection{Recursive Cancellation}\label{sec:recursive_cancellations}

Next let  \((P, Q)\) fulfill the contribution conditions such that \(L = P \oplus Q\) satisfies condition 1 of \cref{prop:subsetrep} but does not satisfy condition 2.
Let \(\ell\) be the maximal index for which 
\begin{align}
    \kv{\ell} \coloneqq \{\vv{\ell - 1}, \vv{\ell}, \vv{\ell + 1}\} \subseteq L.
    \label{eq:kldef}
\end{align}
We fix  \(L\) proceed to characterize the possible path pairs \((P, Q)\) for which \(P \oplus Q = L\).
\begin{claim}
    \label{cl:clique_prefix_coloring}
    For the aforementioned collection of \(P\), \(Q\), \(L\), and \(\kv{\ell}\), we must have that precisely one of \(\{\vv{\ell - 1}, \vv{\ell}\}\) belongs to \(P {\setminus} Q\), and the other belongs to \(Q {\setminus} P\) (i.e. these two vertices receive different colors).
\end{claim}
\begin{proof}
    We first suppose without loss of generality that \(\cv{\ell} = \{\vv{\ell - 1}, \vv{\ell}\} \subset P {\setminus} Q\). 
    If \(\vv{\ell + 1} \in P {\setminus} Q \), we have \(\kv{\ell} \subseteq P\), in contradiction with condition 2 of \cref{prop:subsetrep} for \(P\).
    If \(\vv{\ell + 1} \in Q {\setminus} P \), then we have \(\cv{\ell} \cap Q = \varnothing \) and \(\vv{Q}\) has index larger than \(\ell\), in contradiction with condition 1 for \(Q\).
    An identical argument applies if \(\cv{\ell} \subset Q {\setminus} P\), so we have that precisely one of \(\{\vv{\ell - 1}, \vv{\ell}\}\) belongs to \(P{\setminus} Q\) and the other belongs to \(Q {\setminus} P\).
\end{proof}
\begin{claim}
    \label{cl:max_vert}
    For the aforementioned collection of \(P\), \(Q\), \(L\), and \(\kv{\ell}\), let \(j\) be the maximal index such that \(\vv{j} \in L\), then we cannot have \(\vv{k} \in P \cup Q\) for any \(k > j\) if \(P\) and \(Q\) satisfy the contribution conditions. 
\end{claim}
\begin{proof}
    For all \(\vv{k} \in P \cup Q\) with \(k > j\), we have \(\vv{k} \in P \cap Q\).
    In particular, this implies that \(\vv{P} = \vv{Q}\) if this common free endpoint has index larger than \(j\), so \(P\), \(Q\), and \(L\) satisfy the conditions of \cref{prop:symdiff_cond1}.
    By \cref{cor:PQ_anticommuting}, we have \(\hv{P}\hv{Q} = (-1)^{|L| + 1} \hv{Q}\hv{P}\), and \(|P| + |Q| = |L| + 2 |P \cap Q|\), so it is not possible for \((P, Q)\) to satisfy the contribution conditions.
\end{proof}

Next, we define in our setting the maximal integer \(N\) for which
\begin{align}
    A &\coloneqq \bigcup_{j = 0}^N \{\vv{\ell_j - 1}, \vv{\ell_j + 1}\} \subset L \mh{5mm} \text{with} \mh{5mm} \ell_j = \ell + 2 + 3j, \mh{5mm} \text{and} \label{eq:A_def} \\
    S &\coloneqq \{\vv{\ell_j}\}_{j = 0}^N \label{eq:S_def}.
\end{align}
Note that, in the case of \(S \cap \overline{P \cup Q} = \varnothing \), this is analogous to the decomposition of \cref{prop:symdiff_cond1}, with \(S\) playing the role of \(P \cap Q\): the vertices of \(L\) between two consecutive elements \(\vv{\ell_j}\), \(\vv{\ell_{j+1}} = \vv{\ell_j + 3}\) of \(S\) form the block \(L_j = \{\vv{\ell_j + 1}, \vv{\ell_j + 2}\}\) in the notation of \cref{eq:symdiff_structure}, so that \(|L_j| = 2\) for each block.
Note also that \(S \cap L = \varnothing\), that is, \(S \subseteq (P \cap Q) \cup \overline{P \cup Q}\): if \(\vv{\ell_j} \in L\) for some \(j\), then \(\kv{\ell_j} = \{\vv{\ell_j - 1}, \vv{\ell_j}, \vv{\ell_j + 1}\} \subseteq L\) with \(\ell_j > \ell\), in contradiction with the maximality of \(\ell\).
With these defined, we show the following
\begin{claim}
    \label{cl:A_S_nonempty}
    For the aforementioned collection of \(P\), \(Q\), \(L\), \(\kv{\ell}\), \(A\), \(S\), and \(N\), we have that \(N \geq 0\) (i.e. \(\varnothing \notin \{A, S\}\)) if \(P\) and \(Q\) satisfy the contribution conditions. 
\end{claim}
\begin{proof}
    If \(\vv{\ell_0} = \vv{\ell + 2} \in L\), then we contradict the maximality of \(\ell\) since then we would have \(\kv{\ell + 1} \subset L\).
    Thus, \(\vv{\ell_0} \notin L\). 
    If \(\vv{\ell_0 + 1} = \vv{\ell + 3} \notin L\), then \(L\) fails to satisfy condition 1 of \cref{prop:subsetrep} for \(\cv{\ell + 3}\) unless \(\ell + 1\) is the maximal index of any vertex in \(L\).
    Suppose this is the case.
    Since the pair is unordered, we may label the paths such that \(\vv{\ell - 1} \in P {\setminus} Q\), and then \(\vv{\ell} \in Q {\setminus} P\) by \cref{cl:clique_prefix_coloring}.
    By \cref{cl:max_vert}, this gives that \(\vv{\ell - 1} = \vv{P}\) and \(\vv{\ell + 1} = \vv{Q}\), or \(\vv{\ell} = \vv{Q}\) and \(\vv{\ell + 1} = \vv{P}\).
    Suppose that the first case holds, and let \(\widetilde{P} = P \cup \{\vv{\ell + 1}\}\) and \(\widetilde{L} = \widetilde{P} \oplus Q\).
    We see that \(\widetilde{P} \in \mc{P}^*\), as it satisfies the conditions of \cref{prop:subsetrep}, and we see that \(\widetilde{P}\), \(Q\), and \(\widetilde{L}\) satisfy the conditions of \cref{prop:symdiff_cond1}.
    Thus, by \cref{cor:PQ_anticommuting}
    \begin{align}
        (-1)^{|\widetilde{L}| + 1} = \scomm{\hv{\widetilde{P}}}{\hv{Q}} = \scomm{\hv{P}}{\hv{Q}} \scomm{\hv{\ell + 1}}{\hv{Q}} = -\scomm{\hv{P}}{\hv{Q}},
    \end{align}
    so \(\scomm{\hv{P}}{\hv{Q}} = (-1)^{|\widetilde{L}|}\).
    Similarly, 
    \begin{align}
        |P| + |Q| = (|\widetilde{P}| - 1) + |Q| = |\widetilde{L}| + 2|\widetilde{P} \cap Q| - 1 = |\widetilde{L}| + 1 \text{ (mod 2)}.
    \end{align}
    Thus, it is not possible for \((P, Q)\) to satisfy the contribution conditions in this case.
    The second case, in which \(\vv{\ell} = \vv{Q}\) and \(\vv{\ell + 1} = \vv{P}\), follows in the same way upon setting \(\widetilde{Q} = Q \cup \{\vv{\ell + 1}\}\) and \(\widetilde{L} = P \oplus \widetilde{Q}\), so we must have \(\vv{\ell_0 + 1} \in L\).
    We therefore have \(\{\vv{\ell_0 - 1}, \vv{\ell_0 + 1}\} \subset L\), \(N \geq 0\), \(A \neq \varnothing\), and \(S \neq \varnothing\).
\end{proof}

\begin{claim}
    \label{cl:ordering}
    For the aforementioned collection of \(P\), \(Q\), \(L\), \(\kv{\ell}\), \(A\), \(S\), and \(N\), suppose that both \(S \cap \overline{P \cup Q}\) and \(S \cap (P \cap Q)\) are nonempty.
    The index of every vertex in \(S \cap \overline{P \cup Q}\) is larger than the index of any vertex in \(S \cap (P \cap Q)\).
\end{claim}

\begin{proof}
    Suppose by contradiction that there exists \(j\), \(k\) with \(j < k\) such that \(\vv{\ell_j} \in \overline{P \cup Q}\) and \(\vv{\ell_k} \in P \cap Q\).
    If such a pair exists, then, since \(S \subseteq (P \cap Q) \cup \overline{P \cup Q}\), there is also a pair for which \(k = j + 1\), so we assume this is the case.
    We have \(\{\vv{\ell_j + 1}, \vv{\ell_j + 2}\} \subset L\), as  \(\ell_j + 2 = \ell_k - 1\).
    If both of these vertices are in \(P {\setminus} Q\), then \(P\) fails condition 2 of \cref{prop:subsetrep} for \(\kv{\ell_j + 2}\) since \(\vv{\ell_j + 3} = \vv{\ell_k} \in P \cap Q\).
    If both vertices are in \(Q {\setminus} P\), then \(Q\) fails the same condition.
    If \(\vv{\ell_j + 1} \in P {\setminus} Q\) and \(\vv{\ell_j + 2} \in Q {\setminus} P\), then \(Q\) fails condition 1 of \cref{prop:subsetrep} for \(\cv{\ell_j + 1}\) since \(\vv{\ell_j} \notin P \cup Q\).
    Similarly, if \(\vv{\ell_j + 1} \in Q {\setminus} P\) and \(\vv{\ell_j + 2} \in P {\setminus} Q\), then \(P\) fails the same condition.
    Thus no such pair \(j\), \(k\) exists, which is exactly the claim.
\end{proof}

\begin{claim}
    \label{cl:tail_max_index}
    With \(\ell\), \(A\), \(S\), and \(N\) as defined in \cref{eq:kldef,,eq:A_def,,eq:S_def}, if \(\vv{\ell_N + 3} \notin L\), then the maximal index of any vertex in \(L\) is \(\ell_N + 1\) or \(\ell_N + 2\).
\end{claim}
\begin{proof}
    Suppose \(\vv{\ell_N + 3} \notin L\).
    If \(\vv{\ell_N + 2} \notin L\), then \(\ell_N + 1\) is the maximal index of any vertex in \(L\). 
    Otherwise, some \(\vv{j} \in L\) with \(j > \ell_N + 1\) would require \(j \ge \ell_N + 4\) (since \(\vv{\ell_N+2}, \vv{\ell_N+3} \notin L\)), and \(L\) would fail condition 1 of \cref{prop:subsetrep} for \(\cv{\ell_N + 3}\), contrary to assumption.
    If instead \(\vv{\ell_N + 2} \in L\), then it is the maximal index of any vertex in \(L\): otherwise we require \(\vv{\ell_N + 4} \in L\) so that \(L\) satisfies condition 1 of \cref{prop:subsetrep} for \(\cv{\ell_N + 4}\), but then this produces an additional subset \(\{\vv{\ell_{N+1} - 1}, \vv{\ell_{N+1} + 1}\} = \{\vv{\ell_N + 2}, \vv{\ell_N + 4}\} \subset L \), contradicting the maximality of \(N\).
\end{proof}

\begin{proposition}[Unique path-coloring]
    \label{prop:symdiff_cond2}
    If \(L = P \oplus Q\) satisfies condition 1 of \cref{prop:subsetrep} but does not satisfy condition 2 for \(\kv{\ell}\) of maximal index \(\ell\) as defined in \cref{eq:kldef}, let \(N\) denote the maximal index defining \(A\) and \(S\) (nonempty by \cref{cl:A_S_nonempty}) as given in \cref{eq:A_def,,eq:S_def}.
    Since the pair \((P, Q)\) is unordered, we label the two paths such that \(\vv{\ell - 1} \in P {\setminus} Q\), and hence \(\vv{\ell} \in Q {\setminus} P\) by \cref{cl:clique_prefix_coloring}.
    If \(S \cap (P \cap Q) = \varnothing\), then the color of \(\vv{\ell + 1}\) fully determines the path-membership assignment to all vertices in \((P \cup Q) \cap \{\vv{j}\}_{j \geq \ell - 1}\) (and hence the coloring of \(L \cap \{\vv{j}\}_{j \geq \ell - 1}\)).
    If \(S \cap (P \cap Q) \neq \varnothing\), then letting \(\ell_m\) be the maximal index of any vertex in this set, the value of \(m\) fully determines the path-membership assignment to all of the vertices in \((P \cup Q) \cap \{\vv{j}\}_{j \geq \ell - 1}\) (and hence the coloring of \(L \cap \{\vv{j}\}_{j \geq \ell - 1}\)). 
    In this case \(\vv{\ell_j - 1} \in P {\setminus} Q\) and \(\vv{\ell_j + 1} \in Q {\setminus} P\) for all \(j \leq m\), and \(\{\vv{j}\}_{j > \ell_m + 1} \cap L \subseteq P\) if \(m < N\).
\end{proposition}
\begin{proof}
    With the labeling fixed in the statement, we have \(\vv{\ell - 1} \in P {\setminus} Q\), and thus \(\vv{\ell} \in Q {\setminus} P\).
    If \(S \cap (P \cap Q)\) is empty, then if \(\vv{\ell + 1} \in Q {\setminus} P\), then \(\vv{\ell - 1} = \vv{P}\), or \(P\) fails condition 1 of \cref{prop:subsetrep} for \(\cv{\ell+ 1}\).
    This gives \(\{\vv{j}\}_{j > \ell - 1} \cap L \subseteq Q\).
    If \(\vv{\ell + 1} \in P {\setminus} Q\), then \(\vv{\ell} = \vv{Q}\) or \(Q\) fails condition 1 of \cref{prop:subsetrep} for \(\cv{\ell + 2} = \cv{\ell_0}\).
    This gives \(\{\vv{j}\}_{j > \ell} \cap L \subseteq P\), so we see that the path to which \(\vv{\ell + 1}\) belongs fully constrains the path-membership assignment to all vertices in \((P \cup Q) \cap \{\vv{j}\}_{j \geq \ell - 1}\); the colors of \(\vv{\ell - 1}\) and \(\vv{\ell}\) are fixed by the labeling convention and \cref{cl:clique_prefix_coloring}, and not by any further choice.
    
    Next suppose \(S \cap (P \cap Q)\) is nonempty, and let \(\ell_m\) be the maximal index of any vertex in this set according to the proposition.
    That is \(S \cap \{\vv{\ell_j}\}_{j = 0}^m \subseteq (P \cap Q)\), and \(S \cap \{\vv{\ell_j}\}_{j = m + 1}^N \subseteq \overline{P \cup Q}\) by \cref{cl:ordering}.
    We must have that \(\vv{\ell + 1} \in P {\setminus} Q\) or \(P\) fails condition 1 of \cref{prop:subsetrep} for \(\cv{\ell + 1}\).
    For any \((P, Q)\) satisfying these conditions such that \(P \oplus Q = L\), we let \(P' \coloneqq P \cap \{\vv{j}\}_{j \leq \ell_m}\), \(Q' \coloneqq Q \cap \{\vv{j}\}_{j \leq \ell_m}\), and \(L' = P' \oplus Q'\).
    Since \(P'\), \(Q'\), and \(L'\) satisfy the conditions of \cref{prop:symdiff_cond1}, we have that the path-membership assignment of the vertices in \(P' \cap \{\vv{j}\}_{j \geq \ell - 1}\) and \(Q' \cap \{\vv{j}\}_{j \geq \ell - 1}\) are fixed by \cref{eq:symdiff_paths} and the fact that \(\vv{\ell + 1} \in P {\setminus} Q\).
    Explicitly, colors alternate along each block of \(L'\) and are exchanged across each vertex of \(P' \cap Q'\), so \(\vv{\ell_0 - 1} = \vv{\ell + 1} \in P {\setminus} Q\) implies \(\vv{\ell_j - 1} \in P {\setminus} Q\) and \(\vv{\ell_j + 1} \in Q {\setminus} P\) for all \(j \leq m\).
    
    Finally, we have two cases for the vertices in \(L \cap \{\vv{j}\}_{j \geq \ell_m}\): either \(m < N\) or \(m = N\).
    If \(m < N\), then we cannot have that both \(\{\vv{\ell_m + 1}, \vv{\ell_m + 2}\}\) belong to \(P {\setminus} Q\) or \(P\) fails condition 2 of \cref{prop:subsetrep} for \(\kv{\ell_m + 1}\), and an identical argument holds if both of these vertices belong to \(Q {\setminus} P\).
    Since \(\vv{\ell_m - 1}\) belongs to \(P {\setminus} Q\) as shown above, \(P\) again fails condition 2 for \(\kv{\ell_m}\) if \(\vv{\ell_m + 1}\) belongs to \(P {\setminus} Q\), so we must have \(\vv{\ell_m + 1} \in Q {\setminus} P\) and \(\vv{\ell_m + 2} \in P {\setminus} Q\).
    We must additionally have that \(\vv{\ell_m + 1} = \vv{Q}\) or \(Q\) will fail condition 1 for \(\cv{\ell_m + 3}\) (indeed, \(\vv{\ell_m + 2} \in P {\setminus} Q\), and \(\vv{\ell_m + 3} = \vv{\ell_{m+1}} \notin Q\) since \(S \subseteq (P \cap Q) \cup \overline{P \cup Q}\) and \(m\) is maximal).
    This gives that \(\{\vv{j}\}_{j > \ell_m + 1} \cap L \subset P\), and this fixes the path-membership assignment of all vertices in \((P \cup Q) \cap \{\vv{j}\}_{j \geq \ell - 1}\) in the case that \(m < N\)
    
    If \(m = N\), then we again have two cases: either \(\vv{\ell_N + 3} \notin L\) or \(\vv{\ell_N + 3} \in L\).
    If \(\vv{\ell_N + 3} \notin L\), then by \cref{cl:tail_max_index} the maximal index of any vertex in \(L\) is \(\ell_N + 1\) or \(\ell_N + 2\).
    Let this maximal index be \(k\); recall that \(\vv{\ell_N - 1} \in P {\setminus} Q\) as shown above.
    If \(\vv{\ell_N + 1} \in P {\setminus} Q\), then \(P\) fails condition 2 of \cref{prop:subsetrep} for \(\kv{\ell_N}\), so \(\vv{\ell_N + 1} \in Q {\setminus} P\).
    If \(k = \ell_N + 1\), then by \cref{cl:max_vert}, this gives \(\vv{\ell_N} = \vv{P}\) and \(\vv{\ell_N + 1} = \vv{Q}\). 
    If \(k = \ell_N + 2\), then if \(\vv{\ell_N + 2} \in Q {\setminus} P\), then \(Q\) fails condition 2 for \(\kv{\ell_N + 1}\), so \(\vv{\ell_N + 2} \in P {\setminus} Q\).
    This gives \(\vv{\ell_N + 1} = \vv{Q}\) and \(\vv{\ell_N + 2} = \vv{P}\).
    In both cases, the path-membership assignment to all vertices in \((P \cup Q) \cap \{\vv{j}\}_{j \geq \ell - 1}\) is completely determined.
    
    If \(\vv{\ell_N + 3} \in L\), then we cannot have \(\vv{\ell_N + 2} \in L\), or we contradict the assumed maximality of \(\ell\) because \(\kv{\ell_N + 2} = \{\ell_N + 1, \ell_N + 2, \ell_N + 3\} \subset L\).
    As shown above, \(\vv{\ell_N - 1} \in P {\setminus} Q\) and \(\vv{\ell_N + 1} \in Q {\setminus} P\).
    If \(\vv{\ell_N + 2} \in P \cap Q\), then \(Q\) fails condition 2 of \cref{prop:subsetrep} for \(\kv{\ell_N + 1}\), so we must have \(\vv{\ell_N + 2} \notin P \cup Q\).
    If \(\vv{\ell_N} \neq \vv{P}\), then \(P\) fails condition 1 of \cref{prop:subsetrep} for \(\cv{\ell_N + 2}\).
    Thus, \(\{\vv{j}\}_{j > \ell_N} \subset Q\), and the path-membership assignment to all vertices in \((P \cup Q) \cap \{\vv{j}\}_{j \geq \ell - 1}\) is completely determined in this case.
\end{proof}

\Cref{prop:symdiff_cond2} allows us to partition the set of pairs \((P, Q)\) with \(P \oplus Q = L\) into collections which recursively cancel.
In particular, we let 
\begin{align}
    \mc{C}[P_{\leq \ell}, Q_{\leq \ell}, L] = \{(P, Q) \in \mc{P}^* \times \mc{P}^* \mid P \cap \{\vv{j}\}_{j \leq \ell} = P_{\leq \ell},\ Q \cap \{\vv{j}\}_{j \leq \ell} = Q_{\leq \ell},\ P \oplus Q = L\}
\end{align}
where \(L\) is a fixed vertex set that satisfies condition 1 but not condition 2 of \cref{prop:subsetrep}, \(\ell\) is determined by \(L\) as specified in \cref{eq:kldef}, and \(P_{\leq \ell}\), \(Q_{\leq \ell} \subseteq \{\vv{j}\}_{j \leq \ell}\) are fixed vertex sets with \(\vv{\ell - 1} \in P_{\leq \ell}\) (the labeling convention of \cref{prop:symdiff_cond2}), and the sum ranges over pairs of anchored induced paths.
Every unordered pair of paths with \(P \oplus Q = L\) belongs to exactly one such class, and the requirement \(P \oplus Q = L\) already encodes that condition 1 holds and condition 2 fails.
We prove that 
\begin{align}
    \sum_{(P, Q) \in \mc{C}[P_{\leq \ell}, Q_{\leq \ell}, L]} \Omega_{(P, Q)} = 0.
\end{align}

Using the definitions and the labeling convention of \cref{prop:symdiff_cond2}, \(\vv{\ell - 1}\) is \(P\)-colored (and thus \(\vv{\ell}\) is \(Q\)-colored by \cref{cl:clique_prefix_coloring}) for every pair in \(\mc{C}[P_{\leq \ell}, Q_{\leq \ell}, L]\), since \(P_{\leq \ell}\) and \(Q_{\leq \ell}\) are fixed, and this convention only names the two paths of an unordered pair.
\Cref{prop:symdiff_cond2} gives that there is a unique pair for which \(S \cap (P \cap Q) = \varnothing\) and \(\vv{\ell + 1}\) is \(Q\)-colored; let this be \((P_a, Q_a)\).
Similarly, there is a unique pair for which \(S \cap (P \cap Q) = \varnothing\) and \(\vv{\ell + 1}\) is \(P\)-colored; let this be \((P_b, Q_b)\).
Finally, let the unique pair for which \(S \cap (P \cap Q) = \{\vv{\ell_j}\}_{j = 0}^m\) be \((P_m, Q_m)\).
In particular, we denote \((P_{-1}, Q_{-1}) \coloneqq (P_b, Q_b)\) in the following expressions, and we have
\begin{align}
    P_b &= P_a \cup (Q_a \cap \{\vv{j}\}_{j > \ell}) \label{eq:pbdef} \\
    Q_b &= Q_a \setminus (Q_a \cap \{\vv{j}\}_{j > \ell}), \label{eq:qbdef}
\end{align}
and
\begin{align}
    \label{eq:pmqmdef}
    \begin{cases}
        \begin{cases}
            P_m = (P_{m - 1} \setminus \{\vv{\ell_m + 1}\}) \cup \{\vv{\ell_m}\} & \\
            Q_m = (Q_{m - 1} \cup \{\vv{\ell_m + 1}\}) \cup \{\vv{\ell_m}\} &
        \end{cases} & 0 \leq m < N \text{ or } m = N \text{ and } \vv{\ell_N + 3} \notin L \\
        \begin{cases}
            P_m = [P_{m - 1} \setminus (P_{m - 1} \cap \{\vv{j}\}_{j \geq \ell_m + 1})] \cup \{\vv{\ell_m}\} & \\
            Q_m = [Q_{m - 1} \cup (P_{m - 1} \cap \{\vv{j}\}_{j \geq \ell_m + 1})] \cup \{\vv{\ell_m}\} & 
        \end{cases} & m = N \text{ and } \vv{\ell_N + 3} \in L
    \end{cases}.
\end{align}
\begin{proposition}
    \label{prop:allornothing}
    All of the elements of \(\mc{C}[P_{\leq \ell}, Q_{\leq \ell}, L]\) satisfy the contribution conditions if and only if \((P_a, Q_a)\) does.
\end{proposition}
\begin{proof}
    By \cref{eq:pbdef,,eq:qbdef}, we have \(|P_b| + |Q_b| = |P_a| + |Q_a|\).
    Additionally, \(|P_m| + |Q_m| = |P_{m - 1}| + |Q_{m - 1}| + 2\) for \(m \geq 0\).
    
    Let \(R_a = Q_a \cap \{\vv{j}\}_{j > \ell} \in \mc{P}\), then
    \begin{align}
        \scomm{\hv{P_a}}{\hv{R_a}} = \scomm{\hv{Q_a}}{\hv{R_a}} = -1
    \end{align}
    since \(P_a \cup R_a \in \mc{P}^*\) and \(Q_a \setminus R_a \in \mc{P}^*\).
    Thus
    \begin{align}
        \scomm{\hv{P_b}}{\hv{Q_b}} = \scomm{\hv{P_a}\hv{R_a}}{\hv{Q_a}\hv{R_a}} = \scomm{\hv{P_a}}{\hv{Q_a}}
    \end{align}
    If \(N > 0\), let \(R_0 = \{\vv{\ell_0}, \vv{\ell_0 + 1}\}\).
    We have \(N_{P_b}(\vv{\ell_0}) = \{\vv{\ell_0 - 1}, \vv{\ell_0 + 1}, \vv{\ell_0 + 2}\}\) and \(Q_b \cup R_0 \in \mc{P}^*\).
    Thus,
    \begin{align}
        \scomm{\hv{P_b}}{\hv{R_0}} = \scomm{\hv{Q_b}}{\hv{R_0}} = -1
    \end{align}
    and
    \begin{align}
        \scomm{\hv{P_0}}{\hv{Q_0}} = \scomm{\hv{P_b}\hv{R_0}}{\hv{Q_b}\hv{R_0}} = \scomm{\hv{P_b}}{\hv{Q_b}}.     
    \end{align}
    If \(0 < m < N\), then we similarly let \(R_m = \{\vv{\ell_m}, \vv{\ell_m + 1}\}\).
    Recall from the proof of \cref{prop:symdiff_cond2} that \(\{\vv{j}\}_{j > \ell_{m - 1} + 1} \cap L \subseteq P_{m - 1}\) in this case.
    We have \(N_{P_{m - 1}}(\vv{\ell_m}) = \{\vv{\ell_m - 1}, \vv{\ell_m + 1}, \vv{\ell_m + 2}\}\) and \(Q_{m - 1} \cup R_{m} \in \mc{P}^*\).
    Thus, 
    \begin{align}
        \scomm{\hv{P_m}}{\hv{Q_m}} = \scomm{\hv{P_{m - 1}}\hv{R_m}}{\hv{Q_{m - 1}}\hv{R_m}} = \scomm{\hv{P_{m - 1}}}{\hv{Q_{m - 1}}}\label{eq:temp_62}
    \end{align}
    
    If \(N = 0\) and \(\vv{\ell_N + 3} \notin L\), then by \cref{cl:tail_max_index}, we have that either \(\ell_N + 1\) or \(\ell_N + 2\) is the maximal index in \(L\).
    In the former case, \(N_{P_b}(\vv{\ell_N + 1}) = \{\vv{\ell_N - 1}\}\) and \(N_{P_b}(\vv{\ell_N}) = \{\vv{\ell_N - 1}, \vv{\ell_N + 1}\}\) (recall that \(N = 0\) here, so \(\vv{\ell_N - 2} = \vv{\ell} \in Q_b {\setminus} P_b\)). 
    Thus, again setting \(R_0 = \{\vv{\ell_N}, \vv{\ell_N + 1}\}\), we have
    \begin{align}
        \scomm{\hv{P_0}}{\hv{Q_0}} = \scomm{\hv{P_b}\hv{R_0}}{\hv{Q_b}\hv{R_0}} = \scomm{\hv{P_b}}{\hv{Q_b}}.
    \end{align}
    In the latter case, \(N_{P_b}(\vv{\ell_N + 1}) = \{\vv{\ell_N - 1}, \vv{\ell_N + 2}\}\) and \(N_{P_b}(\vv{\ell_N}) = \{\vv{\ell_N - 1}, \vv{\ell_N + 1}, \vv{\ell_N + 2}\}\).
    This gives
    \begin{align}
        \scomm{\hv{P_0}}{\hv{Q_0}} = \scomm{\hv{P_b}\hv{R_0}}{\hv{Q_b}\hv{R_0}} = \scomm{\hv{P_b}}{\hv{Q_b}}.
    \end{align}
    If \(\vv{\ell_N + 3} \in L\), then again \(\vv{\ell_N + 2} \notin L\), and let \(R_0 = (P_b \cap \{\vv{j}\}_{j \geq \ell_0 + 1}) \cup \vv{\ell_0}\). 
    Since \(N_{P_b}(\vv{\ell_0}) = \{\vv{\ell_0 - 1}, \vv{\ell_0 + 1}\}\) and \(\vv{Q_b} = \vv{\ell_0 - 2} = \vv{\ell}\), we have \(\scomm{\hv{P_b}}{\hv{R_0}} = \scomm{\hv{Q_b}}{\hv{R_0}} = -1\), and
    \begin{align}
        \scomm{\hv{P_0}}{\hv{Q_0}} = \scomm{\hv{P_b}\hv{R_0}}{\hv{Q_b}\hv{R_0}} = \scomm{\hv{P_b}}{\hv{Q_b}}.
    \end{align}
    Finally, if \(m = N > 0\), and \(\vv{\ell_N + 3} \notin L\), then by an identical argument, we set \(R_m = \{\vv{\ell_N}, \vv{\ell_N + 1}\}\), and 
    \begin{align}
        \scomm{\hv{P_m}}{\hv{Q_m}} = \scomm{\hv{P_{m - 1}} \hv{R_m}}{\hv{Q_{m - 1}} \hv{R_m}} = \scomm{\hv{P_{m - 1}}}{\hv{Q_{m - 1}}}.
    \end{align}
    If \(\vv{\ell_N + 3} \in L\), then again by an identical argument, let \(R_m = (P_{m - 1} \cap \{\vv{j}\}_{j \geq \ell_m + 1}) \cup \vv{\ell_m}\).
    Since \(N_{P_{m - 1}}(\vv{\ell_m}) = \{\vv{\ell_m - 1}, \vv{\ell_m + 1}\}\) and \(\vv{Q_{m - 1}} = \vv{\ell_m - 2}\), we have \(\scomm{\hv{P_{m - 1}}}{\hv{R_{m}}} = \scomm{\hv{Q_{m - 1}}}{\hv{R_m}} = -1\), and
    \begin{align}
        \scomm{\hv{P_m}}{\hv{Q_m}} = \scomm{\hv{P_{m - 1}}\hv{R_m}}{\hv{Q_{m - 1}}\hv{R_m}} = \scomm{\hv{P_{m - 1}}}{\hv{Q_{m - 1}}}.
    \end{align}
    This therefore gives that every element of \(\mc{C}[P_{\leq \ell}, Q_{\leq \ell}, L]\) satisfies the contribution conditions if and only if \((P_a, Q_a)\) does.
\end{proof}

We next prove the following claim
\begin{claim}
    \label{cl:tail_parity}
    With the path pair definitions given in \cref{eq:pbdef,,eq:qbdef,,eq:pmqmdef}, let \(Q_{> \ell} \coloneqq Q_a \cap \{\vv{j}\}_{j > \ell}\), then if these path pairs, which all belong to \(\mc{C}[P_{\leq \ell}, Q_{\leq \ell}, L]\), satisfy the contribution conditions, we have \(|Q_{> \ell}|\) even.
\end{claim}
\begin{proof}
    Since \(P_b \in \mc{P}^*\), we have that \(\hv{Q_{>\ell}}\) anticommutes with \(\hv{P_a}\) by \cref{eq:pbdef}, and thus \(\scomm{\hv{P_a}}{\hv{Q_b}} = 1\) by the fact that \((P_a, Q_a)\) satisfies the contribution conditions.
    Consider the anchored induced paths given by \(\widetilde{P} = P_a \cup \{\vv{\ell + 1}\}\) and \(\widetilde{Q} = Q_b \cup \{\vv{\ell + 1}\}\) with shared endpoint \(\vv{\ell + 1}\), and let \(\widetilde{L} = \widetilde{P} \oplus \widetilde{Q}\).
    We see that \(\widetilde{P}\), \(\widetilde{Q}\), and \(\widetilde{L}\) satisfy the conditions of \cref{prop:symdiff_cond1}, and since \(\scomm{\hv{\widetilde{P}}}{\hv{\widetilde{Q}}} = 1\) with \(\widetilde{L} = |P_a| + |Q_b|\) (mod 2), we have \(|P_a| + |Q_b|\) odd by \cref{cor:PQ_anticommuting}.
    Thus, \(|Q_{> \ell}|\) is even by the contribution conditions for \((P_a, Q_a)\).
\end{proof}
\Cref{cl:tail_parity} additionally implies the following corollary, which can be viewed as a refinement of \cref{cl:tail_max_index} to incorporate the contribution conditions.
\begin{corollary}
    \label{cor:max_index}
    For the aforementioned collection of \(P\), \(Q\), \(L\), \(\kv{\ell}\), \(A\), \(S\), and \(N\), if \(P\) and \(Q\) satisfy the contribution conditions and \(\vv{\ell_N + 3} \notin L\), then the maximal index of any vertex in \(L\) is \(\ell_N + 1\).
\end{corollary}
\begin{proof}
    If \(\vv{\ell_N + 3} \notin L\), we have that either \(\ell_N + 1\) or \(\ell_N + 2\) is the maximal index in \(L\) by \cref{cl:tail_max_index}.
    If \(\vv{\ell_N + 2} \in L\), then \(Q_{> \ell} = \{\vv{\ell_j - 1}, \vv{\ell_j + 1}\}_{j = 0}^N \cup \{\vv{\ell_N + 2}\}\), so \(|Q_{> \ell}|\) is odd, in contradiction with \cref{cl:tail_parity}.
    Thus, \(\vv{\ell_N + 2} \notin L\), and this proves the corollary.
\end{proof}

Thus, we assume that \((P_a, Q_a)\) satisfies the contribution conditions, or the terms in \(\mc{C}[P_{\leq \ell}, Q_{\leq \ell}, L]\) are not included in the sum by \cref{prop:allornothing}, and the sum is trivially zero.

Next, we show the following claim, recalling that \((P_{-1}, Q_{-1}) \coloneqq (P_b, Q_b)\).
\begin{claim}
    \label{cl:partial_sum}
    For all \(m\) with \(-1 \leq m < N\), we have
    \begin{align}
        &\Omega_{(P_a, Q_a)} + \sum_{i = -1}^m \Omega_{(P_i, Q_i)} \nonumber \\
        &\mh{5mm} = -i^{|P_m| - |Q_m|} u^{|P_m| + |Q_m|} \bvs{\ell_{m + 1}} P_{G{\setminus} N[P_m]} P_{(G {\setminus} N[Q_m]){\setminus} N[\vv{\ell_{m + 1}}]} [\hv{P_m}, \hv{Q_m}]
    \end{align}
\end{claim}
\begin{proof}
    For the base case, let \(m = -1\).
    We have
    \begin{align}
        \label{eq:base_case_relative}
        |P_b| - |Q_b| &= (|P_a| + |Q_{>\ell}|) - (|Q_a| - |Q_{>\ell}|) = |P_a| - |Q_a| + 2|Q_{>\ell}| 
    \end{align}
    and
    \begin{align}
        \hv{P_a} \hv{Q_a} = \hv{P_a} \hv{Q_b} \hv{Q_{>\ell}} = -\hv{P_a} \hv{Q_{>\ell}} \hv{Q_b} = -\hv{P_b} \hv{Q_b}.
    \end{align}
    This gives
    \begin{align}
        \Omega_{(P_a, Q_a)} + \Omega_{(P_b, Q_b)} &= i^{|P_a|-|Q_a|}u^{|P_a|+|Q_a|-2} P_{G{\setminus} N[P_a]} P_{G{\setminus} N[Q_a]} [\hv{P_a},\hv{Q_a}] \nonumber \\
        &\mh{5mm} + i^{|P_b|-|Q_b|}u^{|P_b|+|Q_b|-2} P_{G{\setminus} N[P_b]} P_{G{\setminus} N[Q_b]} [\hv{P_b},\hv{Q_b}] \\
        &= -i^{|P_b|-|Q_b|}u^{|P_b|+|Q_b|-2}\left(P_{G{\setminus} N[P_a]} P_{G{\setminus} N[Q_a]}- P_{G{\setminus} N[P_b]}P_{G{\setminus} N[Q_b]}\right)[\hv{P_b},\hv{Q_b}]  
    \end{align}
    where we used \cref{cl:tail_parity} and \cref{eq:base_case_relative} to apply the identity \(i^{|P_a| - |Q_a|} = i^{|P_b| - |Q_b|}\).
    We have \(N[Q_a] = N[P_b]\).
    Since \(\vv{P_a} = \vv{\ell - 1}\) and \(\vv{Q_b} = \vv{\ell}\), we have \(N[P_a] = N[Q_b] \setminus \{\vv{\ell + 2}\}\).
    This gives
    \begin{align}
        \Omega_{(P_a, Q_a)} + \Omega_{(P_b, Q_b)} &= -i^{|P_b|-|Q_b|}u^{|P_b|+|Q_b|-2}P_{G{\setminus} N[P_b]}\left( P_{(G{\setminus} N[Q_b]) \cup \{\vv{\ell+2}\}}-P_{G{\setminus} N[Q_b]}\right) [\hv{P_b},\hv{Q_b}] \\
        &= -i^{|P_b|-|Q_b|}u^{|P_b|+|Q_b|} \bvs{\ell_0} P_{G{\setminus} N[P_b]}P_{(G{\setminus} N[Q_b])\setminus N[\vv{\ell_0}]}[\hv{P_b},\hv{Q_b}]\;,
    \end{align}
    where we used \cref{eq:independence_polynomial_recursion} with the clique \(K = \{\vv{\ell + 2}\}\), applied to the graph \((G {\setminus} N[Q_b]) \cup \{\vv{\ell + 2}\}\).
    This proves the base case.
    
    Next, assume the statement is true for \(m - 1\), and we show that it holds for \(m\).
    We have
    \begin{align}
        |P_m| - |Q_m| &= |P_{m - 1}| - |Q_{m - 1}| - 2 \\
        |P_m| + |Q_m| &= |P_{m - 1}| + |Q_{m - 1}| + 2
    \end{align}
    and
    \begin{align}
        \bvs{\ell_m} \hv{P_{m - 1}}\hv{Q_{m - 1}} &= \bvs{\ell_m} \hv{P_{m - 1} \cap \{\vv{j}\}_{j < \ell_m}} \hv{\ell_m + 1} \hv{P_{m - 1} \cap \{\vv{j}\}_{j > \ell_m + 1}}\hv{Q_{m - 1}} \\ 
        &= -\hv{P_{m - 1} \cap \{\vv{j}\}_{j < \ell_m}} \hv{\ell_m} \hv{P_{m - 1} \cap \{\vv{j}\}_{j > \ell_m + 1}}\hv{Q_{m - 1}} \hv{\ell_m} \hv{\ell_m + 1} \\
        &= -\hv{P_m}\hv{Q_m}
    \end{align}
    since \(\bvs{\ell_m} = \hvs{\ell_m}\); since both \(\hv{\ell_m}\) and \(\hv{\ell_m + 1}\) anticommute with exactly one factor of \(\hv{P_{m - 1} \cap \{\vv{j}\}_{j > \ell_m + 1}}\), namely \(\hv{\ell_m + 2}\) (recall that \(\{\vv{j}\}_{j > \ell_{m - 1} + 1} \cap L \subseteq P_{m - 1}\) and \(\vv{\ell_m + 3} = \vv{\ell_{m+1}} \notin P_{m - 1}\)); and since \(\hv{\ell_m}\) anticommutes with exactly one factor of \(\hv{Q_{m - 1}}\), namely the factor of its last vertex \(\vv{Q_{m - 1}} = \vv{\ell_m - 2}\), while \(\hv{\ell_m + 1}\) commutes with all of \(\hv{Q_{m - 1}}\).
    This gives
    \begin{align}
        &\Omega_{(P_a, Q_a)} + \sum_{i = -1}^m \Omega_{(P_i, Q_i)} = \Omega_{(P_a, Q_a)} + \sum_{i = -1}^{m - 1} \Omega_{(P_i, Q_i)} + \Omega_{(P_m, Q_m)} \\
        &\mh{5mm}=-i^{|P_{m-1}| - |Q_{m - 1}|} u^{|P_{m -1}| + |Q_{m - 1}|} \bvs{\ell_{m}} P_{G{\setminus} N[P_{m - 1}]} P_{(G {\setminus} N[Q_{m -1}])\setminus N[\vv{\ell_{m}}]} [\hv{P_{m - 1}}, \hv{Q_{m - 1}}] \\
        &\mh{10mm} + i^{|P_m| - |Q_m|} u^{|P_m| + |Q_m| - 2} P_{G{\setminus} N[P_m]}P_{G {\setminus} N[Q_m]} [\hv{P_m}, \hv{Q_m}] \\
        &\mh{5mm}=-i^{|P_{m}| - |Q_{m}|}u^{|P_{m}| + |Q_{m}| - 2}( P_{G{\setminus} N[P_{m - 1}]} P_{(G {\setminus} N[Q_{m -1}])\setminus N[\vv{\ell_{m}}]} - P_{G{\setminus} N[P_m]}P_{G {\setminus} N[Q_m]})[\hv{P_{m}}, \hv{Q_{m}}].
    \end{align}
    Using \(N[P_{m}] = N[P_{m - 1}]\) and \(N[Q_{m - 1}] \cup N[\vv{\ell_m}] = N[Q_{m}] \setminus \{\vv{\ell_{m + 1}}\}\) gives
    \begin{align}
        &\Omega_{(P_a, Q_a)} + \sum_{i = -1}^m \Omega_{(P_i, Q_i)} \nonumber \\
        &\mh{5mm}= -i^{|P_{m}| - |Q_{m}|}u^{|P_{m}| + |Q_{m}| - 2}P_{G {\setminus} N[P_m]}(P_{(G {\setminus} N[Q_{m}])\cup \{\vv{\ell_{m + 1}}\}} - P_{G {\setminus} N[Q_m]})[\hv{P_{m}}, \hv{Q_{m}}] \\
        &\mh{5mm}= -i^{|P_{m}| - |Q_{m}|}u^{|P_{m}| + |Q_{m}|} \bvs{\ell_{m + 1}} P_{G {\setminus} N[P_m]} P_{(G {\setminus} N[Q_m]) \setminus N[\vv{\ell_{m + 1}}]}[\hv{P_{m}}, \hv{Q_{m}}],
    \end{align}
    where the last step is again \cref{eq:independence_polynomial_recursion} with \(K = \{\vv{\ell_{m+1}}\}\), applied to the graph \((G {\setminus} N[Q_m]) \cup \{\vv{\ell_{m+1}}\}\), which produces the factor \(u^2 \bvs{\ell_{m+1}}\).
    This proves the claim.
\end{proof}

Finally, if \(\vv{\ell_N + 3} \notin L\), then the maximal index of any vertex in \(L\) is \(\ell_N + 1\) by \cref{cor:max_index}, and we have
\begin{align}
    \bvs{\ell_N}\hv{P_{N - 1}}\hv{Q_{N - 1}} &= \bvs{\ell_{N}} \hv{P_{N - 1} \cap \{\vv{j}\}_{j < \ell_N}} \hv{\ell_N + 1} \hv{Q_{N - 1}} \\
    &= -\hv{P_{N - 1} \cap \{\vv{j}\}_{j < \ell_N}} \hv{\ell_N} \hv{Q_{N - 1}} \hv{\ell_N} \hv{\ell_N + 1} \\
    &= -\hv{P_N} \hv{Q_N}
\end{align} 
for which 
\begin{align}
    |P_N| - |Q_N| &= |P_{N - 1}| - |Q_{N - 1}| - 2 \\
    |P_N| + |Q_N| &= |P_{N - 1}| + |Q_{N - 1}| + 2
\end{align}
with \(N[P_{N}] = N[Q_{N - 1}] \cup N[\vv{\ell_N}]\), and \(N[Q_N] = N[P_{N - 1}]\).
Therefore,
\begin{align}
    &\Omega_{(P_a, Q_a)} + \sum_{i = -1}^N \Omega_{(P_i, Q_i)} \nonumber \\
    &\mh{5mm}= -i^{|P_{N - 1}| - |Q_{N - 1}|} u^{|P_{N - 1}| + |Q_{N - 1}|} \bvs{\ell_{N}} P_{G{\setminus} N[P_{N - 1}]} P_{(G {\setminus} N[Q_{N - 1}])\setminus N[\vv{\ell_{N}}]} [\hv{P_{N - 1}}, \hv{Q_{N - 1}}] \nonumber \\
    &\mh{10mm} + i^{|P_N| - |Q_N|} u^{|P_N| + |Q_N| - 2} P_{G {\setminus} N[P_N]} P_{G {\setminus} N[Q_N]}[\hv{P_N}, \hv{Q_N}] \\
    &\mh{5mm}= -i^{|P_{N}| - |Q_{N}|} u^{|P_{N}| + |Q_{N}| - 2}P_{G{\setminus} N[Q_N]} P_{G {\setminus} N[P_N]} [\hv{P_{N}}, \hv{Q_{N}}] \nonumber \\
    &\mh{10mm} + i^{|P_N| - |Q_N|} u^{|P_N| + |Q_N| - 2} P_{G {\setminus} N[P_N]} P_{G {\setminus} N[Q_N]}[\hv{P_N}, \hv{Q_N}] \\
    &\mh{5mm}= 0
\end{align}
proving the recursive cancellation in this case.

If \(\vv{\ell_N + 3} \in L\), we have
\begin{align}
    \bvs{\ell_N}\hv{P_{N - 1}}\hv{Q_{N - 1}} &= \bvs{\ell_{N}} \hv{P_{N - 1} \cap \{\vv{j}\}_{j < \ell_N}} \hv{P_{N - 1} \cap \{\vv{j}\}_{j \geq \ell_N + 1}}\hv{Q_{N - 1}} \\
    &= -\hv{P_{N - 1} \cap \{\vv{j}\}_{j < \ell_N}} \hv{\ell_{N}} \hv{Q_{N - 1}} \hv{\ell_N} \hv{P_{N - 1} \cap \{\vv{j}\}_{j \geq \ell_N + 1}} \\
    &= -\hv{P_N} \hv{Q_N}
\end{align}
where the last step uses \(\bvs{\ell_N} = \hvs{\ell_N}\); that \(\hv{\ell_N}\) anticommutes with exactly one factor of \(\hv{Q_{N - 1}}\), namely the factor of its last vertex \(\vv{Q_{N - 1}} = \vv{\ell_N - 2}\); and that \(\hv{Q_{N - 1}}\) commutes with \(\hv{P_{N - 1} \cap \{\vv{j}\}_{j \geq \ell_N + 1}}\), since no vertex of the latter is adjacent to a vertex of \(Q_{N-1}\).
This gives
\begin{align}
    |P_N| - |Q_N| &= (|P_{N - 1}| - |P_{N - 1} \cap \{\vv{j}\}_{j \geq \ell_{N} + 1}| + 1) - (|Q_{N - 1}| + |P_{N - 1} \cap \{\vv{j}\}_{j \geq \ell_N + 1}| + 1)\\
    &= (|P_{N - 1}| - |Q_{N - 1}|) - 2|P_{N - 1} \cap \{\vv{j}\}_{j \geq \ell_{N} + 1}| \\
    &= (|P_{N - 1}| - |Q_{N - 1}|) - 2(|Q_{> \ell}| - 2(N + 1) + 1),
\end{align}
and
\begin{align}
    |P_N| + |Q_N| &= (|P_{N - 1}| - |P_{N - 1} \cap \{\vv{j}\}_{j \geq \ell_{N} + 1}| + 1) + (|Q_{N - 1}| + |P_{N - 1} \cap \{\vv{j}\}_{j \geq \ell_N + 1}| + 1)\\
    |P_N| + |Q_N| &= |P_{N - 1}| + |Q_{N - 1}| + 2,
\end{align}
as well as \(N[P_N] = N[Q_{N - 1}] \cup N[\vv{\ell_N}]\), \(N[Q_N] = N[P_{N - 1}]\).
This gives
\begin{align}
    &\Omega_{(P_a, Q_a)} + \sum_{i = -1}^N \Omega_{(P_i, Q_i)} \nonumber \\
    &\mh{5mm}= -i^{|P_{N - 1}|-|Q_{N - 1}|}u^{|P_{N - 1}|+|Q_{N - 1}|} \bvs{\ell_N} P_{G{\setminus} N[P_{N - 1}]}P_{(G{\setminus} N[Q_{N - 1}])\setminus N[\vv{\ell_N}]} [\hv{P_{N - 1}},\hv{Q_{N - 1}}] \nonumber \\
    &\mh{10mm} + i^{|P_N| - |Q_N|} u^{|P_N| + |Q_N| - 2} P_{G {\setminus} N[P_N]} P_{G {\setminus} N[Q_N]}[\hv{P_N}, \hv{Q_N}] \\
    &\mh{5mm}= i^{|P_N|-|Q_N| + 2(|Q_{> \ell}| - 2(N + 1) + 1)}u^{|P_N|+|Q_N| - 2} P_{G{\setminus} N[Q_N]}P_{G{\setminus} N[P_N]} [\hv{P_N},\hv{Q_N}] \nonumber \\
    &\mh{10mm} + i^{|P_N| - |Q_N|} u^{|P_N| + |Q_N| - 2} P_{G {\setminus} N[P_N]} P_{G {\setminus} N[Q_N]}[\hv{P_N}, \hv{Q_N}] \\
    &\mh{5mm}= 0
\end{align}
where we used the fact that \(|Q_{> \ell}|\) is even (\cref{cl:tail_parity}) in the last line.
This proves the recursive cancellation in this case, and therefore proves recursive cancellation in general.

\section{Currents Lie Algebra}
\label[appendix]{app:current_algebra}

In this section we provide a proof for a more general version of~\cref{thm:ffd_current_algebra}, which describes the algebra generated by the current operators. 
We write  \(d_K\) for the dimension of the Krylov space, as we use  \(d\) as an index.
\begin{theorem}\label{.current_lie_algebra}
    Let  \(s = 2j+1\) for some  \(j \in \mathbb{N}\), set  \(r \coloneqq \min\{s,d_K-1\}+1\) and let  \(\liek\) be the real Lie algebra generated by the odd currents  \(\cur{1}, \cur{3}, \ldots, \cur{s-2}, \cur{s},\) with Lie bracket
    \begin{equation}
        [x, y]_L \coloneqq i[x, y], \qquad x, y \in \liek .
    \end{equation}
    If either  \(r \leq 3\) or
    \begin{equation}
        C_{11}^2\left(C_{00}C_{22}-C_{11}^2\right) - C_{00}^2\left(C_{11}C_{33}-C_{22}^2\right) \neq 0,
    \end{equation}
    then
    \begin{align}
        \liek &= \set{ \frac{i}{2}\sum_{a,b=0}^{s} c_{ab} \e[\kryl{a}, \kryl{b}] }{ c_{ab}=-c_{ba}\in\mathbb{R} \text{ for } a,b \in \{0,\ldots,s\}} \label{eq:sumkappa_m}\\
        &= \set{ \frac{i}{2}\sum_{a,b=0}^{r-1} c_{ab} \e[\kryl{a}, \kryl{b}] }{ c_{ab}=-c_{ba}\in\mathbb{R}\text{ for } a,b \in \{0,\ldots,r-1\}} \label{eq:sumkappa_r}\\
        &\cong \lieso{r}. \label{eq:so(r)}
    \end{align}
    Here  \(\lieso{r}\) denotes the Lie algebra induced by the special orthogonal group in  \(r\) dimensions, equivalently the space of real skew-symmetric  \(r \times r\) matrices with the commutator as Lie bracket.
\end{theorem}

Note that in this section, the results are not restricted to the FFD model (except where expressly noted), but hold more generally if one has Hermitian operators whose anticommutators are scalars defined by a matrix with properties as in \cref{.anti_properties}. 
For example, this section holds in the case where the frustration graph is simplicial, claw-free, and even-hole free, while the latter condition can be dropped by introducing appropriate projectors as in~\cite{chapman2023unified}.

Throughout this section, let the Krylov dimension be non-zero, i.e.,  \(d_K > 0\) (if  \(d_K = 0\) the Hamiltonian must be identically zero) and set  \(C = 2\bs{M}^G\). 
From Ref.~\cite{chapman2023unified}, we have the following proposition:
\begin{proposition}[Symmetry of  \(C\)~\cite{chapman2023unified}]\label{.anti_properties}
    The matrix  \(C\) is symmetric with the following \emph{alternating Hankel} property: for all  \(a, b \in \mathbb{N}_0\) with  \(a > 0\), we have
    \begin{equation}
        C_{a,b} = - C_{a-1, b+1}\;.
    \end{equation}
    Furthermore,  \(C'\coloneqq \e*(C_{jk})_{0 \leq j,k \leq d_K-1}\) is (strictly) positive definite.
\end{proposition}
\begin{proof}
    This is all in Ref.~\cite[Section 6, Krylov Subspaces]{chapman2023unified} and in the proofs therein.
\end{proof}

We want to show that the Lie algebra generated by the current operators~(\cref{eq:current_operators}) is isomorphic to the special orthogonal Lie algebra.
We do this by introducing hermitian Krylov bilinears that are equivalent to Majorana bilinears and can be used to express the current operators, and then show that the Lie bracket closure of the currents generates all these Hermitian Krylov bilinears.
\begin{definition}[Hermitian Krylov bilinears]
    For all  \(a, b \in \mathbb{N}_0\), we define the \emph{Hermitian} Krylov bilinears as
    \begin{equation}
        \bih{a,b} = i\e(\bil{a}{b} - C_{ab}/2) = \frac{i}{2}\e(\bil{a}{b} - \bil{b}{a}) = \frac{i}{2} \e[\kryl{a}, \kryl{b}]\;.
    \end{equation}
\end{definition}
\begin{proposition}\label{.bilinear_span}
    Let  \(a, b \in \mathbb{N}_0\). 
    We have
    \begin{equation}
        \bih{a,b} = \sum_{c,d = 0}^{d_K-1} \alpha_{abcd} \bih{c,d}\;,
    \end{equation}
    for appropriate coefficients  \(\alpha_{abcd} \in \mathbb{R}\) with  \(\alpha_{abcd} \propto \delta_{ac}\) if  \(a < d_K\) and  \(\alpha_{abcd} \propto \delta_{bd}\) if  \(b < d_K\).
\end{proposition}
\begin{proof}
    By definition of the Krylov dimension,  \(\kryl{d_K}\) is a real linear combination of  \(\kryl{0}, \ldots, \kryl{d_K-1}\), that is,  \(\mc{L}\) satisfies a polynomial relation of degree  \(d_K\) on the Krylov space (its minimal polynomial).
    Applying  \(\mc{L}\) repeatedly to this relation shows that every  \(\kryl{a}\) is a real linear combination  \(\kryl{a} = \sum_{c=0}^{d_K-1} \lambda_{ac}\kryl{c}\), with  \(\lambda_{ac} = \delta_{ac}\) for  \(a < d_K\).
    Since  \(\bih{a,b} = \frac{i}{2}\e[\kryl{a}, \kryl{b}]\) is bilinear, we obtain  \(\bih{a,b} = \sum_{c,d=0}^{d_K-1} \lambda_{ac}\lambda_{bd}\,\bih{c,d}\), so  \(\alpha_{abcd} = \lambda_{ac}\lambda_{bd}\) has the stated properties.
\end{proof}
\begin{lemma}
    \label{.bilinear_commutator}
    Let  \(a, b, c, d \in \mathbb{N}_0\). 
    We have
    \begin{align}
        i\bk{\bih{a,b}, \bih{c,d}} &= -i\bk{\bil{a}{b}, \bil{c}{d}}\\
        &= i\e(C_{ac}\bil{b}{d} + C_{bd}\bil{a}{c} - C_{ad}\bil{b}{c} - C_{bc}\bil{a}{d} -C_{ac}C_{bd} + C_{ad}C_{bc})\\
        &= C_{ac}\bih{b, d} + C_{bd}\bih{a, c} - C_{ad}\bih{b, c} - C_{bc}\bih{a, d}.
    \end{align}
\end{lemma}
\begin{proof}
    This follows by elementary calculations.
\end{proof}

\begin{proposition}\label{.characterise_anti_comm}
    Let  \(a, b \in \mathbb{N}_0\). 
    If  \(a + b\) is odd, then  \(C_{ab} = 0\) and if  \(a + b\) is even, and  \(a + b \leq 2(d_K-1)\), then  \(C_{ab} \neq 0\).
\end{proposition}
\begin{proof}
    First, consider the case where  \(a+b\) is odd. 
    Without loss of generality let  \(a > b\). 
    Using the alternating Hankel property in \cref{.anti_properties}  \(a-b\) times, we get
    \begin{equation}
        C_{a,b} = (-1)^{a-b}C_{b,a} = -C_{b,a} = -C_{a,b}\;,
    \end{equation}
    where we used the symmetry of  \(C\) in the last equality; but then it must be that  \(C_{a,b} = 0\).
    
    Now let  \(a+b\) be even with  \(a + b \leq 2(d_K-1)\) and assume that  \(C_{a,b} = 0\). 
    If  \(a = b \leq d_K-1\), then it follows that  \(e_a^T C' e_a = 0\), which contradicts the (strictly positive) definiteness of  \(C'\). 
    So without loss of generality, let  \(a > b\) and set  \(k = (a-b)/2\). 
    Then, using the alternating Hankel relation again, we get  \(0 = C_{a,b} = (-1)^k C_{a-k,b+k}\), but  \(a-k = b+k\) and so we have a  \(c = a-k \leq d_K-1\) such that  \(C'_{c,c} = 0\), which again contradicts the definiteness of  \(C'\). 
    Hence  \(C_{a,b} \neq 0\).
\end{proof}
\begin{definition}[Orthonormalised Krylov Majoranas{~\cite[Section 6]{chapman2023unified}}]
    Define the orthogonal and diagonal matrices  \(U, D \in \mathbb{R}^{d_K \times d_K}\), respectively, such that  \(U^T D U = C'/2\) (which is possible as  \(C'\) is symmetric and (strictly) positive definite). 
    For  \(a \in \e{0, \ldots, d_K-1}\), define the \emph{Majorana} operators  \(\mm_a\) as
    \begin{equation}
        \mm_a = \frac{1}{\sqrt{D_{aa}}}\sum_{b = 0}^{d_K-1} U_{ab}\kryl{b}\;,
    \end{equation}
    which have canonical anticommutation relations.
\end{definition}
\begin{proposition}\label{.bilinear_majorana_relation}
    Let  \(a, b \in \e{0, \ldots, d_K-1}\). 
    Then
    \begin{equation}
        i \mm_a \mm_b = \frac{1}{\sqrt{D_{aa}D_{bb}}}\sum_{c,d = 0}^{d_K-1} U_{ac}U_{bd}\bih{c,d} + i\delta_{a,b}
    \end{equation}
    and
    \begin{equation}
        \bih{a,b} = \sum_{c,d = 0}^{d_K-1} \sqrt{D_{cc}D_{dd}} U^T_{ac}U^T_{bd}i\mm_c \mm_d - i\frac{C_{ab}}{2} = \sum_{\substack{c,d = 0\\c\neq d}}^{d_K-1} \sqrt{D_{cc}D_{dd}} U^T_{ac}U^T_{bd}i\mm_c \mm_d
    \end{equation}
\end{proposition}
\begin{proof}
    This follows by elementary calculations.
\end{proof}

We are now able to prove \Cref{.current_lie_algebra}.
\begin{proof}[Proof of \cref{.current_lie_algebra}]
    Firstly, note that the anti-symmetry of the coefficients  \(c_{ab}\) in the sets follows from the anti-symmetry of  \(\bih{a,b}\) and  \(i\mm_a\mm_b\) (similarly,  \(c_{aa} = 0\) is enforced). 
    Furthermore, the first two sets are obviously closed under the Lie bracket due to~\cref{.bilinear_commutator}, which can be specialised (by setting  \(C = 2\id\)) to Majoranas.
    
    We prove the equations in order. 
    Equality in~\cref{eq:sumkappa_m} follows inductively with respect to  \(s\): The case  \(s = 1\) is clear. 
    Assume for now that the case  \(s = 3\) also holds (will be proven separately below).
    Let the hypothesis hold for  \(s-2\), with  \(s \geq 5\). 
    We must show that by including  \(\cur{s}\) into the generators we get the operators  \(\e{\bih{0,s-1}, \ldots, \bih{s-2,s-1}, \bih{0,s}, \ldots, \bih{s-1,s}}\) into the algebra.
    
    Firstly assume that  \(s \leq d_K-1\). 
    By the induction hypothesis we have  \(\cur{s} = 4\e(\bih{0,s} - \bih{1,s-1}) + y\) for some  \(y \in \liek\). 
    Therefore, the element  \(u = \bih{0,s} - \bih{1,s-1}\) is in  \(\liek\). 
    Now we apply the Lie bracket with  \(\bih{0,2}\) to  \(u\), using \cref{.bilinear_commutator,.characterise_anti_comm}:
    \begin{align}
        i\e[\bih{0,2}, u] &= i\e[\bih{0,2}, \bih{0,s}] - i\e[\bih{0,2}, \bih{1,s-1}]\\
        &= C_{00}\bih{2,s} - C_{20}\bih{0,s} + y \quad \text{with  \(y \in \liek\)}\;.
    \end{align}
    Therefore,  \(v = C_{00}\bih{2,s} - C_{20}\bih{0,s}\) is in  \(\liek\). 
    Applying again the Lie bracket with  \(\bih{0,2}\) gives:
    \begin{align}
        i\e[\bih{0,2}, v] &= C_{00}i\e[\bih{0,2}, \bih{2,s}] - C_{20}i\e[\bih{0,2}, \bih{0,s}]\\
        &= -\e(C_{00}C_{22} - C_{20}^2)\bih{0,s}\;.
    \end{align}
    No further correction term is needed here:  \(\e(C_{00}C_{22} - C_{20}^2)\bih{0,s} \in \liek\) exactly, since the terms proportional to  \(\bih{0,2}\) vanish because  \(C_{0s} = C_{2s} = 0\) for odd  \(s\) and the  \(\bih{2,s}\) terms cancel each other.
    The value  \(C_{00}C_{22} - C_{20}^2\) is the determinant of the  \(2\times 2\) submatrix of  \(C\) with rows and columns  \(0\) and  \(2\), i.e., it is a principal minor of  \(C'\), and since  \(C'\) is (strictly) positive definite, this value is strictly positive; in particular, it is non-zero.
    Therefore  \(w = \bih{0,s}\) is in  \(\liek\). 
    Furthermore, this gives us  \(\bih{1,s-1} = w - u \in \liek\). 
    That is, we have shown that  \(\bih{0,s}, \bih{1,s-1} \in \liek\). 
    These two elements let us generate the remaining terms that we need: Let  \(a, b, c \in \mathbb{N}_0\), then we have in general
    \begin{equation}
        i\e[\bih{a,b}, \bih{b,c}] = C_{ab}\bih{b, c} + C_{bc}\bih{a, b} - C_{bb}\bih{a, c}
    \end{equation}
    Setting  \(b = 1\),  \(c = s-1\) and  \(a \in \e{0, \ldots, s-2}\) gives
    \begin{equation}
        i\e[\bih{a,1}, \bih{1,s-1}] = -C_{11}\bih{a,s-1} + y \quad \text{with  \(y \in \liek\)}\;.
    \end{equation}
    Therefore, since  \(C_{11} \neq 0\), we have  \(\bih{a, s-1} \in \liek\) for all  \(a \in \e{0, \ldots, s-2}\). 
    Analogously, setting  \(b = 0\),  \(c = s\) and  \(a \in \e{0, \ldots, s-1}\) gives us  \(\bih{a, s} \in \liek\) for all  \(a \in \e{0, \ldots, s-1}\). 
    This concludes the induction step for  \(s \leq d_K-1\).
    
    Now let  \(s > d_K-1\). 
    If  \(d_K-1\) is odd, then we are done, because the induction hypothesis, and the previous case  \(s \leq d_K-1\), already tell us that the full algebra is generated by  \(\cur{1}, \ldots, \cur{d_K-1}\). 
    Similarly, if  \(d_K-1\) is even, we see that we are done if we can show the statement for  \(s = d_K\). 
    Therefore, let  \(s = d_K\) and we only have to show that by including  \(\cur{s}\) into the generators we get the operators  \(\e{\bih{0,s-1}, \ldots, \bih{s-2,s-1}}\) (as these are the only independent bilinears that are missing and the dependent ones are then automatically in  \(\liek\) according to \cref{.bilinear_span}). 
    As  \(\bih{0,s} = \sum_{d = 0}^{s-1} \alpha_{d} \bih{0,d}\) for appropriate  \(\alpha_{d} \in \mathbb{R}\) (\cref{.bilinear_span}), we have  \(\cur{s}/4 = b\bih{0,s-1} - \bih{1,s-1} + y\) with  \(y \in \liek\) and  \(b = \alpha_{s-1}\), i.e.,  \(u = b\bih{0,s-1} - \bih{1,s-1} \in \liek\). 
    Calculating the commutator between  \(u\) and  \(\bih{0,1}\) gives
    \begin{align}
        i\e[\bih{0,1}, u] &= b\,i\e[\bih{0,1}, \bih{0,s-1}] - i\e[\bih{0,1}, \bih{1,s-1}]\\
        &=bC_{00}\bih{1,s-1} + C_{11}\bih{0,s-1} + bC_{0,s-1}\bih{0,1}\;,
    \end{align}
    where we used that  \(C_{01} = C_{1,s-1} = 0\) (alternating Hankel structure with  \(s\) odd).
    Since  \(s \geq 3\), we have  \(\bih{0,1} \in \liek\) (it is generated by  \(\cur{1}\)), so the last term lies in  \(\liek\).
    Therefore,  \(v = bC_{00}\bih{1,s-1} + C_{11}\bih{0,s-1}\) is in  \(\liek\). 
    If the two vectors  \(\e(b, -1)\) and  \(\e(C_{11}, bC_{00})\) are linearly independent we can clearly generate  \(\bih{0,s-1}\) and  \(\bih{1,s-1}\) with  \(u\) and  \(v\). 
    This can be checked by calculating the determinant of the matrix formed by these two vectors, but this determinant is given by  \(b^2C_{00} + C_{11} > 0\) according to \cref{.anti_properties}. 
    Thus we have  \(\bih{0,s-1}, \bih{1,s-1} \in \liek\) and the remaining elements  \(\bih{a,s-1}\), with  \(a \in \e{2, \ldots, s-2}\) can be generated as in the case  \(s \leq d_K-1\).
    
    All that remains is for us to consider the case  \(s = 3\). 
    If  \(r-1=1\) we have nothing to show and if  \(r-1 = 2\) we are in the induction case above with  \(s = d_K > d_K - 1\) which works for this special case. 
    Therefore, let  \(r-1 = 3\). 
    Unlike the induction step above with  \(s \leq d_K-1\), we cannot use  \(\bih{0,2}\) to generate the  \(\bih{0,s}\), since  \(\bih{0,2}\) is not in the algebra of the case  \(s-2 = 1\). 
    Instead, we have to consider higher order commutators with  \(\cur{3}\). 
    Let us denote  \(u = \cur{3}/4 \in \liek\) and calculate the commutator with  \(\bih{0,1}\):
    \begin{align}
        i\e[\bih{0,1}, u] &= i\e[\bih{0,1}, \bih{0,3}] - i\e[\bih{0,1}, \bih{1,2}]\\
        &= C_{00}\bih{1,3} + C_{11}\bih{0,2}\;.
    \end{align}
    Therefore, we have  \(v = C_{00}\bih{1,3} + C_{11}\bih{0,2} \in \liek\). 
    Now let us take the commutator with  \(u\) again:
    \begin{align}
        i\e[u, v] &= C_{00}\e(i\e[\bih{0,3}, \bih{1,3}] - i\e[\bih{1,2}, \bih{1,3}]) + C_{11}\e(i\e[\bih{0,3}, \bih{0,2}] - i\e[\bih{1,2}, \bih{0,2}])\\
        &= \e(C_{11}C_{02} - C_{00}C_{13})\e(\bih{0,3} + \bih{1,2}) - 2 C_{00}C_{11}\bih{2,3} + y \quad \text{with  \(y \in \liek\)}\;.
    \end{align}
    Therefore, with  \(a = C_{00}C_{22} - C_{11}^2 > 0\) (principal minor of  \(C'\)) and  \(b = C_{00}C_{11} > 0\), we have  \(w = a\e(\bih{0,3} + \bih{1,2}) - 2b\bih{2,3} \in \liek\). 
    Again, we take the commutator with  \(u\):
    \begin{align}
        i\e[u, w]/2 &= ai\e[\bih{0,3}, \bih{1,2}] - b\e(i\e[\bih{0,3}, \bih{2,3}] - i\e[\bih{1,2}, \bih{2,3}])\\
        &= \e(aC_{02} - bC_{22})\bih{1,3} - \e(bC_{33} + aC_{13})\bih{0,2}\;.
    \end{align}
    Therefore, we have  \(x = A \bih{1,3} + B \bih{0,2} \in \liek\) with
    \begin{equation}
        A = aC_{11} + bC_{22} = C_{11}\e(2C_{00}C_{22} - C_{11}^2)
    \end{equation}
    and
    \begin{equation}
        B = bC_{33} - aC_{22} = C_{00}C_{11}C_{33} - C_{00}C_{22}^2 + C_{11}^2C_{22}\;.
    \end{equation}
    If the two vectors  \((C_{00}, C_{11})\) and  \((A, B)\) are linearly independent, then it is again clearly possible to find a linear combination of  \(v, x\) that generates either  \(\bih{1,3}\) or  \(\bih{0,2}\). 
    This can be checked by calculating the determinant of the matrix formed by these two vectors:
    \begin{equation}
        \Delta = AC_{11} - BC_{00} = C_{11}^2 \e(C_{00} C_{22} - C_{11}^2) - C_{00}^2 \e(C_{11} C_{33} - C_{22}^2)\;.
    \end{equation}
    According to the conditions in the statement of the theorem, this value is indeed non-zero. 
    Hence, we can generate  \(\bih{1,3}\) or  \(\bih{0,2}\) isolated, and then conclude the case analogously to how it is done in the case  \(s \geq 5\). 
    This concludes the proof of \cref{eq:sumkappa_m}.
    
    \Cref{eq:sumkappa_r} is clear if  \(s \leq d_K-1\) and otherwise follows simply from \cref{.bilinear_span}. 
    Using \cref{.bilinear_majorana_relation} we can perform a basis change from the Krylov bilinears onto Majorana bilinears, and \cref{eq:so(r)} is then simply the well-known Lie isomorphism between the Majorana Lie algebra and the special orthogonal Lie algebra.
\end{proof}
\begin{remark}\label{rem:lie_condition}
    The condition  \(\Delta = C_{11}^2 \e(C_{00} C_{22} - C_{11}^2) - C_{00}^2 \e(C_{11} C_{33} - C_{22}^2) \neq 0\) is a difference of two positive terms as the terms in the parentheses are principal minors of  \(C'\). 
    For the FFD model, the explicit numerical evaluations show that  \(\Delta\) is indeed non-zero (at least in the checked regime it is always negative), but it would be interesting to see if one could show that in general analytically (or maybe find a counterexample, which is unlikely to be the FFD model; literature on Hankel matrices might be helpful, however, we believe that the explicit form of  \(C\) in terms of  \(\mbf{A}\) may be required). 
    Potentially it is possible to find other such conditions that allow the proof to hold for the  \(s=3\) case by calculating higher order commutators based on  \(i\e[v, w]\) (it holds  \(i\e[\bih{0,1}, v] \propto u\) and  \(i\e[\bih{0,1}, w]\) leads to the same  \(\Delta\) condition).
\end{remark}

\begin{corollary}[FFD specific]
    For the FFD model there exists a set of operator weights and a simplicial mode, e.g., fix  \(b_1 = \ldots = b_6 = 1\) and set  \(\chi = X_1\), such that the Lie algebra generated by the odd currents is isomorphic to the special orthogonal Lie algebra, i.e., the condition on  \(C\) in \cref{.current_lie_algebra} is satisfied and its statement holds.
\end{corollary}
\begin{proof}
    Set  \(b_1 = \ldots = b_6 = 1\), and  \(\chi = X_1\), then one can calculate explicitly that
    \begin{equation}
        (C_{ij})_{0\leq i,j \leq3} =\begin{pmatrix}
            2 & 0 & -2 & 0\\
            0 & 2 & 0 & -6\\
            -2 & 0 & 6 & 0\\
            0 & -6 & 0 & 24
        \end{pmatrix}
    \end{equation}
    which then gives  \(C_{11}^2 \e(C_{00} C_{22} - C_{11}^2) - C_{00}^2 \e(C_{11} C_{33} - C_{22}^2) = -16\).
\end{proof}

\section{More Details on the Phase Diagrams}\label[appendix]{sec:phase_diagram_details}

\noindent In this section we provide more detailed analysis of the numerical data presented in the main text.

The Hamiltonian is free fermion, therefore it is possible to express in terms of quadratic Majorana operators, by diagonalizing the model in this reference frame, we obtain the gap from the smallest single particle energy.

Intuitively, the first approach to do this is to transform the  \((\eta^{(j)})_j\) operators into  \((\mm_j)_j\) Majorana operators as described in Ref. \cite[Section 6 Krylov Subspaces, Equation (129)]{chapman2023unified} by diagonalizing the matrix  \(\bs{M}\) defined by the anticommutation relations of the  \((\eta^{(j)})_j\) operators in \cref{thm:krylov-restated}. 
Then we can express the currents in those Majorana operators (by a similar transformation) and find the orthogonal decomposition of  \(H_{\mrm{FFD}}\) in the space of the  \(\mm\) bilinears via projection. 
However, this effectively fails numerically, as the diagonalization  \(\bs{M}\) is extremely unstable.

Our second, more successful approach is to decompose  \(H\) into Majorana operators defined by the canonical fermionic modes in \cref{eq:diaghamiltonian}. 
This can be done based on the weighted independence polynomial  \(P_N\) of the FFD model.
\begin{lemma}[\cite{fendley2019free}]
    The weighted independence polynomial  \(P_N = P_N[\alpha, \beta, \gamma]\) of  \(H_{\mrm{FFD}}[\alpha, \beta, \gamma]\) can be calculated recursively via
    \begin{equation}
        P_N(x) = (1+x(\alpha + \beta + \gamma)) P_{N-1}(x) - x^2 (\alpha \beta + \alpha \gamma + \beta \gamma) P_{N-2} + x^3 \alpha \beta \gamma P_{N-3}\;,
    \end{equation}
    with  \(P_0 = 1\) and  \(P_{-1} = P_{-2} = 0\). 
    The roots of  \(P_N\),  \((r_j)_{1\leq j \leq \alpha(G)}\), are all real and negative and it holds
    \begin{equation}
        P_N(x) = \prod_{j=1}^{\alpha(G)} \left(1 - \frac{x}{r_j}\right)\;.
    \end{equation}
\end{lemma}
Note that we choose the simplicial mode  \(\chi\) to be connected to the first three vertices, i.e., if  \(K\) is the according simplicial clique, then the independence polynomial of the model without the clique is  \(P_{N-1}\). 
Furthermore, the zero mode in \cref{prop:zero_mode} vanishes as every maximum independent set of  \(G\) contains an element of the simplicial clique (\cref{cor:zero_mode_vanishes}).

If we set  \((\varepsilon_j = 1 / \sqrt{-r_j})_j\), then it holds in terms of the canonical incognito modes in (cf. \cref{eq:diaghamiltonian} and \cite{chapman2023unified,elman2021free})
\begin{equation}
    H_{\mrm{FFD}} = \sum_{j=1}^{\alpha(G)} \varepsilon_j [\psi_j,\psi_j^\dagger]\;.
\end{equation}
Defining the canonical Majorana operators  \((\emm_i)_{1 \leq i \leq 2 \alpha(G)}\) via
\begin{align}
    \begin{cases}
        \emm{2j - 1} = \psi_j + \psi_j^{\dagger} & \\
        \emm{2j} = -i \left(\psi_j -\psi_j^{\dagger}\right) & 
    \end{cases}\;,
\end{align}
for  \(j \in \{1, \ldots, \alpha(G)\}\), then gives
\begin{equation}
    H_{\mrm{FFD}} = \frac{i}{2} \emm^{T} h_{\mrm{FFD}} \emm
\end{equation}
with
\begin{equation}
    h_{\mrm{FFD}} = \bigoplus_{j=1}^{\alpha(G)}\begin{pmatrix}
        0 & -\varepsilon_j\\
        \varepsilon_j & 0
    \end{pmatrix}\;.
\end{equation}
We need to find a similar expression of the current operators  \(J^{(\ell)}\). 
From \cref{prop:l1span} we have
\begin{equation}
    \eta^{(j)} = \frac{i^j}{4} \sum_{m=1}^{\alpha(G)} N_m L_{0m} \varepsilon_m^j \left((-1)^j \psi_m^{\dagger} + \psi_m\right)\;.
\end{equation}
where \(L_{0m}\) is the constant coefficient of the \(m\)th Lagrange polynomial.
For the \emph{odd} currents, i.e.,  \(\ell\) odd, we then get with elementary calculation that
\begin{align}
    {J}^{(\ell)} &=i\sum_{j=0}^\ell (-1)^j[\kryl{j},\kryl{\ell-j}]\\
    &= \frac{i^{\ell + 1}}{16} \sum_{j=0}^\ell (-1)^j \sum_{m, m'=1}^{\alpha(G)} N_m N_{m'} L_{0m} L_{0m'} \varepsilon_m^j \varepsilon_{m'}^{\ell-j} \left[(-1)^j \psi_m^{\dagger} + \psi_m, (-1)^{\ell-j} \psi_{m'}^{\dagger} + \psi_{m'}\right]\\
    &= \frac{i^{\ell + 1}}{16} \sum_{m, m'=1}^{\alpha(G)} N_m N_{m'} L_{0m} L_{0m'} \sum_{j=0}^\ell (-1)^j \varepsilon_m^j \varepsilon_{m'}^{\ell-j} \left[i^{(j\bmod2)} \emm_{2m - (j+1\bmod2)}, i^{j+1\bmod2} \emm_{2m' - (j\bmod2)}\right] \\
    &= -\frac{i^{\ell}}{16} \sum_{m, m'=1}^{\alpha(G)} N_m N_{m'} L_{0m} L_{0m'} \left( \sum_{j = 0\, \text{even}}^{\ell} \varepsilon_m^j \varepsilon_{m'}^{\ell-j} \left[\emm_{2m - 1}, \emm_{2m'}\right] + \sum_{j = 0 \, \text{odd}}^{\ell} \varepsilon_{m'}^j \varepsilon_{m}^{\ell-j} \left[\emm_{2m-1}, \emm_{2m'}\right] \right)\\
    &= -\frac{i}{2} \sum_{m, m'=1}^{\alpha(G)} (\emm_{2m-1} \emm_{2m'} - \emm_{2m'} \emm_{2m-1}) \mu^{(\ell)}_{m, m'}
\end{align}
with
\begin{equation}
    \mu^{(\ell)}_{m, m'} = 2i^{\ell - 1} \sqrt{P_{N-1}(r_m) L_{0m}} \sqrt{P_{N-1}(r_{m'}) L_{0m'}} \left( \sum_{j = 0\, \text{even}}^{\ell} \varepsilon_m^j \varepsilon_{m'}^{\ell-j} + \sum_{j = 0 \, \text{odd}}^{\ell} \varepsilon_{m'}^j \varepsilon_{m}^{\ell-j} \right)\;.
\end{equation}
This gives
\begin{equation}
    J^{(\ell)} = -\frac{i}{2} \emm^T h^{(\ell)} \emm\;,
\end{equation}
with
\begin{equation}
    h^{(\ell)}_{2m-1, 2m'} = - h^{(\ell)}_{2m', 2m-1} = \mu^{(\ell)}_{m, m'}
\end{equation}
for all  \(m, m' \in \e{1, \ldots, \alpha(G)}\) and the other entries being  \(0\). 
In practice, the norms of the currents can become quite large with increasing  \(\ell\), potentially magnitudes larger than the norm of the underlying FFD Hamiltonian. 
Therefore we renormalise (by effectively adjusting the  \(a_{\ell}\) parameters)
\begin{equation}
    h^{(\ell)} \mapsto \frac{\norm{h_\mrm{FFD}}_1}{\norm{h^{(\ell)}}_1} h^{(\ell)}\;,
\end{equation}
where the  \(\norm{\cdot}_1\) is the sum over the absolute values of the matrix elements. 
In total, we then have
\begin{equation}
    H[\alpha, \beta, \gamma, \bs{a}] = \frac{i}{2} \emm^T h \emm \quad \text{with} \quad h = h[\alpha, \beta, \gamma, \bs{a}] = h_{\mrm{FFD}} + \sum_{\ell = 1,\text { odd}}^{\alpha(G)} \frac{a_\ell \norm{h_\mrm{FFD}}_1}{\norm{h^{(\ell)}}_1} h^{(\ell)}\;. \label{e"renormalised_currents_model}
\end{equation}
Here the sum over odd  \(\ell\) is truncated at  \(\ell \leq \alpha(G)\), for numerical reasons (we checked for low resolutions that running up to the Krylov dimension  \(d\) as in \cref{eq:hjdef} gives qualitatively the same results).
The numerical difficulty for this calculation lies in finding the roots of  \(P_N\). 
Specifically, the algorithm must be precise enough so that the roots are all real, negative and distinct.
Furthermore, when the roots are sorted, the evaluation  \(P_{N-1}(r_m)\) must have alternating signs with  \(m\) (so that  \(N^2_m\) is positive as  \(L_{0m}\) has alternating signs). 
To fulfill this for larger  \(\alpha(G)\) we restricted the parameters  \(\alpha\),  \(\beta\) and  \(\gamma\) to integer data types and used Mathematica to calculate the  \(\varepsilon_m\) and  \(\left(\sqrt{P_{N-1}(r_m) L_{0m}}\right)_m\) values.

Given the real and antisymmetric matrix  \(h\), one can then diagonalise  \(H\), for example as described in \cite[Section 2.2]{Chapman2020}, by skew-diagonalising  \(h\) as
\begin{equation}
    h = O^T \left(\bigoplus_{j=1}^{\alpha(G)}\begin{pmatrix}
        0 & -\lambda_j\\
        \lambda_j & 0
    \end{pmatrix}\right) O\;,
\end{equation}
where  \(O\) is an orthogonal matrix and  \(\lambda_j \in \mathbb{R}\). 
The gap is then given by  \(\Delta = 2 \min_{j} \abs{\lambda_j}\). 
By normalising the matrix  \(h\) for this calculation, the numerics are fairly stable and well behaved in practice.

\begin{figure*}[t]
    \centering
    \includegraphics[width=\linewidth]{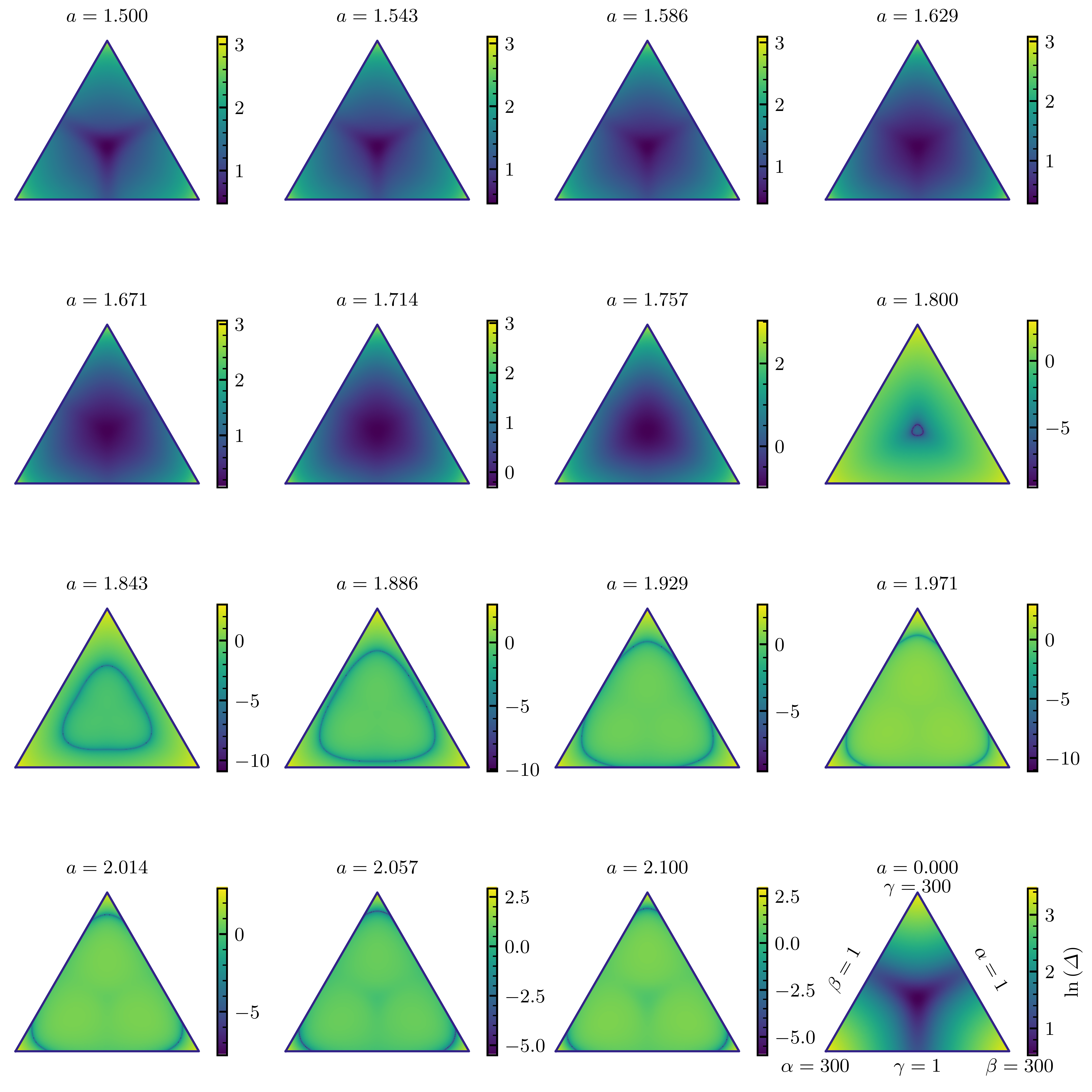}
    \caption{\Cref{fig:triangle_phase_main} but with more samples.}
    \label{fig:triangle_phase_dip1}
\end{figure*}
\begin{figure*}[t]
    \centering
    \includegraphics[width=\linewidth]{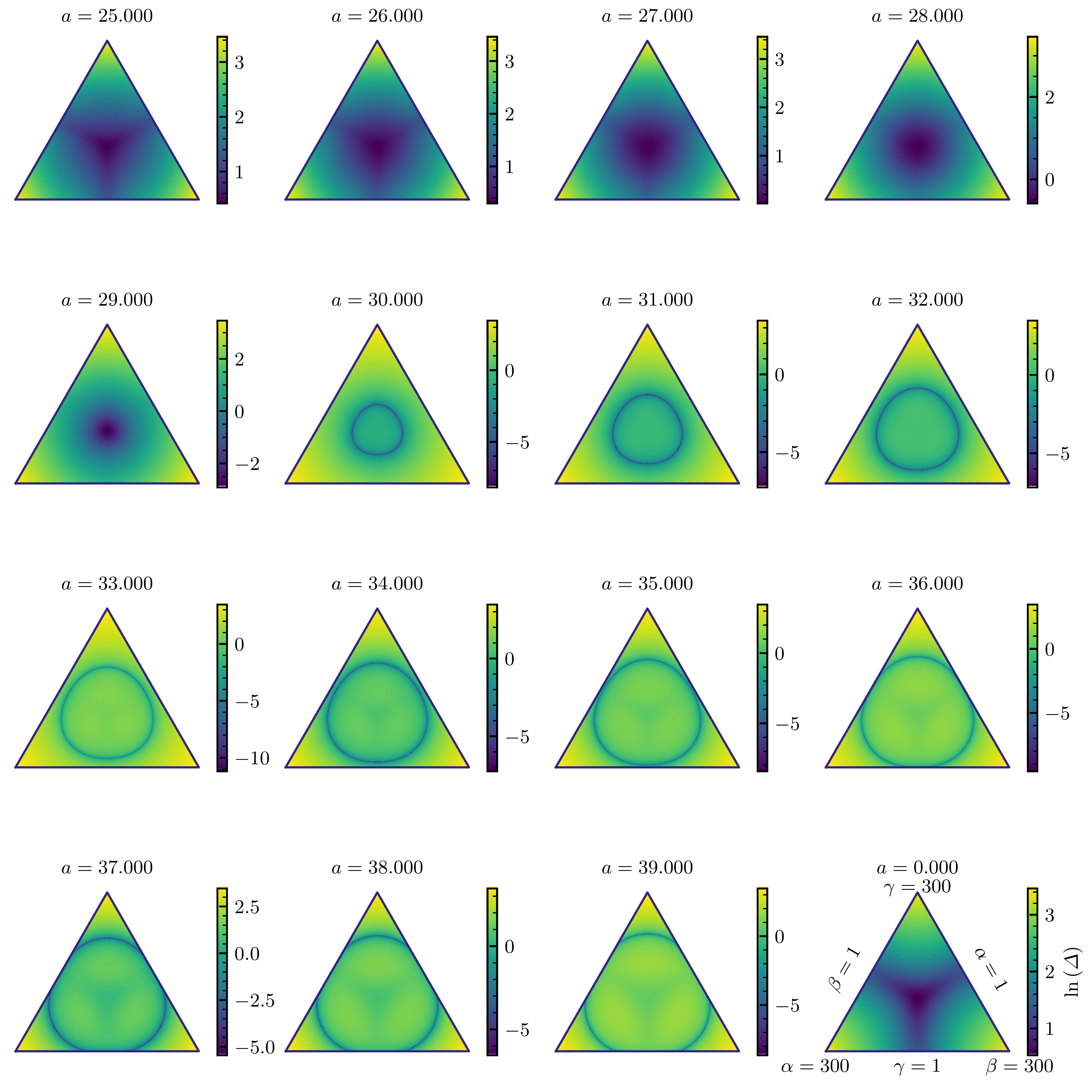}
    \caption{Like \cref{fig:triangle_phase_main,fig:triangle_phase_dip1} but around the second minimum (\(a \approx 30\)) for  \(\alpha(G)=20\).}
    \label{fig:triangle_phase_dip2}
\end{figure*}

\Cref{fig:triangle_phase_dip1,fig:triangle_phase_dip2} show plots complementing the phase diagram shown in \cref{fig:triangle_phase_main}.

\section{Further Discussion Regarding \cref{conj:krylov_path_support_free_fermions}}
\label[appendix]{app:kpsconjecture}

In this section, we provide further motivation and evidence for \cref{conj:krylov_path_support_free_fermions}, which we restate for convenience.

\Krylovconjecture*

The conjecture proposes that the representation of the fermionic modes as anchored induced paths in the graph-theoretic constructions is fundamentally more than just a convenient representation.
The qualification ``on the Krylov space generated by  \(\chi\)'' implies an operator subspace
\begin{equation}
    \mc{K}_{\chi}\coloneqq\operatorname{span}_{\mbb{R}}\e*{\kryl{k}\mid k\ge0}
\end{equation}
that is invariant under the adjoint Hamiltonian action. 
This space carries a linear description of the adjoint action of  \(H\).
The fiducial operator  \(\chi\) is required to be an element of the extended graph Clifford algebra of  \(G\).

The sufficiency direction of \cref{conj:krylov_path_support_free_fermions} is essentially contained in Ref.~\cite{chapman2023unified}, though not isolated there as a stand-alone result.
In general, we have
\begin{align}
    \{\eta^{(j)}, \eta^{(k)}\} = [\eta^{(j)}, \eta^{(k)}] + 2 \eta^{(k)} \eta^{(j)}
\end{align}
where
\begin{align}
    [\eta^{(j)}, \eta^{(k)}] &= \frac{i}{2} [[H, \eta^{(j - 1)}], \eta^{(k)}] \\ 
    &= -\frac{i}{2} ([[\eta^{(j - 1)}, \eta^{(k)}], H] - [[H, \eta^{(k)}], \eta^{(j - 1)}]) \\
    &= -\frac{i}{2} [[\eta^{(j - 1)}, \eta^{(k)}], H] + [\eta^{(k + 1)}, \eta^{(j - 1)}] \\
    &= -\frac{i}{2} [[\eta^{(j - 1)}, \eta^{(k)}], H] + \{\eta^{(k + 1)}, \eta^{(j - 1)}\} - 2 \eta^{(j - 1)} \eta^{(k + 1)}
\end{align}
If we assume by way of induction that \(\{\eta^{(j - 1)},\eta^{(k)}\} = 2M_{j-1,k}^{G} I\) for all \(k \geq 0\) (so that \([\eta^{(j - 1)},\eta^{(k)}] = 2\eta^{(j-1)}\eta^{(k)} - 2M_{j-1,k}^{G} I\)), then
\begin{align}
    -\frac{i}{2} [[\eta^{(j - 1)}, \eta^{(k)}], H] = \frac{i}{2} [H, 2\eta^{(j - 1)}\eta^{(k)}] = 2\eta^{(j)}\eta^{(k)} + 2 \eta^{(j - 1)} \eta^{(k + 1)}
\end{align}
and 
\begin{align}
    [\eta^{(j)}, \eta^{(k)}] &= 2\eta^{(j)}\eta^{(k)} + 2 \eta^{(j - 1)} \eta^{(k + 1)} + 2 M_{j - 1, k + 1}^{G} I - 2 \eta^{(j - 1)}\eta^{(k + 1)} \\ 
    \{\eta^{(j)}, \eta^{(k)}\} &= -2M_{j - 1, k + 1}^{G} I.
\end{align}
Since  \(\bs{M}^G\) is an alternating Hankel matrix,  \(M^G_{j,k} = -M^G_{j-1,k+1}\), and this is exactly the statement  \(\{\eta^{(j)}, \eta^{(k)}\} = 2M^G_{jk} I\) for index  \(j\).
This follows only from the definition of the Krylov modes and the inductive hypothesis, so if the Krylov modes admit anchored path-product expansions, then we are able to additionally claim the base case that \(\{\chi, \eta^{(k)}\} = 2 M_{0, k} I\), since every term in \(\eta^{(k)}\) either is \(\chi\) or anticommutes with \(\chi\).
Consequently, the anticommutator form on the Krylov space takes the form
\begin{equation}
    \e*{ \eta^{(j)},\eta^{(k)} } = 2M_{jk}\id,
\end{equation}
where  \(M\) is real, symmetric, and positive definite after restriction to a linearly independent Krylov basis.
The Krylov operators may therefore be orthonormalised to obtain Majorana operators  \(\{\gamma_a\}_{a=0}^{d-1}\) satisfying the canonical anticommutation relations. 
Because  \(\mc{K}_{\chi}\) is invariant under  \(\mc{L}\), the Hamiltonian adjoint action is represented by a real antisymmetric single-particle matrix on these Majorana modes. 
Decomposing the Hamiltonian into the basis of Majorana strings, this relation then shows that the Hamiltonian must be bilinear in the  \(\e{\gamma_a}_a\) operators. 
Thus, path support implies free-fermion solvability.

The evidence for the converse begins with the generic models. 
Every generic family for which we presently have a graph-theoretic solution method, namely the Hamiltonians with simplicial, claw-free frustration graphs, has this property~\cite{chapman2023unified}.
For these models no model-specific identities are required: the iterated commutators of the simplicial mode are directly supported on anchored induced paths.
The simplicial, claw-free condition may therefore be understood as a purely graphical sufficient condition ensuring that the Krylov evolution never leaves the anchored path-product span.
In this setting, the path description and the generalized canonical anticommutation relations arise from the same combinatorial structure.

More significant evidence is provided by solvable models that lie outside this generic graph class.
We have explicitly examined the Fukai--Vona--Pozsgay model~\cite{fukai2025free} and the Fendley--Pozsgay model~\cite{Fendley2024}, as well as the model in Ref.~\cite{sinha2026hidden}.
Their frustration graphs contain obstructions that permit non-path products to appear in a naive expansion of the Krylov operators.
Nevertheless, these contributions are not independent elements of the operator algebra.
In the Fukai--Vona--Pozsgay model, identities equating a Hamiltonian term  \(h_{\bar{k}}\) with a product of a string of other Hamiltonian terms allow apparently non-path products to be replaced by path-supported products.
In the Fendley--Pozsgay model, products containing forbidden pairs can instead be exchanged for equivalent products with induced-path support.
After imposing these model-specific relations, the Krylov operators again admit anchored path-product expansions.

Crucially, the algebraic relations that restore path support are exact identities of the model's graph Clifford algebra, rather than arbitrary rewritings chosen to enforce a path representation.
Depending on the model, they may follow from relations among Pauli products, from the identification of different products with the same operator, or from a restriction to a symmetry sector in which otherwise distinct products become equivalent.
The support of  \(\eta^{(k)}\) in terms of Hamiltonian terms is consequently not invariant under the choice of presentation of the operator algebra.
The relevant object is instead its support after quotienting by the exact relations of the model.
One may make this interpretation schematic (not necessarily formal) by denoting by  \(\mathfrak{A}_{H}\) the algebra generated by the Hamiltonian terms and by  \(\mathfrak{R}_{H}\) the collection of exact product and cancellation relations.
The path-support condition is then more properly written as
\begin{equation}
    \mc{K}_{\chi} \subseteq \operatorname{span} \bigl\{ h_{P}:P\in\mc{P}^{*} \bigr\} \qquad \text{in } \mathfrak{A}_{H}/\mathfrak{R}_{H}, \label{eq:krylov-path-quotient}
\end{equation}
where  \(\mc{P}^{*}\) is the set of anchored induced paths in the extended frustration graph, and elements in  \(\mathfrak{A}_H/\mathfrak{R}_H\) are equivalent if they transform into each other under operations allowed by  \(\mathfrak{R}_H\).
Because the condition is evaluated in the quotient  \(\mathfrak{A}_{H}/\mathfrak{R}_{H}\), it is manifestly independent of the presentation of the Hamiltonian: rewriting  \(H\) so that a given cancellation is no longer explicit merely relocates the corresponding identity into  \(\mathfrak{R}_{H}\), without changing the quotient in which \cref{eq:krylov-path-quotient} is tested.
When the relevant cancellations hold only on a symmetry sector, the relations in  \(\mathfrak{R}_{H}\) are to be understood as relations of the algebra restricted to that sector---equivalently,  \(\mc{K}_{\chi}\) and the path-support condition are evaluated after projection onto the invariant subspace on which the model is defined (cf.\ the claw-pairing construction of \cref{app:structural-graph}).
For generic simplicial, claw-free models, the relevant quotient is effectively trivial for this purpose.
For the non-generic models above, the relations in  \(\mathfrak{R}_{H}\) are essential and encode information that is not present in the frustration graph alone.

Our computational investigations support this interpretation. 
We have implemented a procedure that expands successive Krylov operators in products of Hamiltonian terms and tests whether the corresponding vertex sequences form anchored induced paths.
Terms failing this test are isolated as candidate obstructions.
For the non-generic models considered above, the resulting non-path terms identify precisely where the additional operator identities must be applied.
After the appropriate substitutions are imposed, every Krylov operator tested reduces to a sum of anchored path products.
We have not found a free-fermion-solvable example for which an irreducible non-path contribution remains.

Finally, the model investigated in this paper,  \(H_{\mathrm{NG}}\), also supports the conjecture. 
While the Krylov generators that we analyse are constructed by the adjoint action of the  \(H_{\mathrm{FFD}}\) Hamiltonian, it turns out that the adjoint action of the  \(H_{\mathrm{NG}}\) Hamiltonian on the same simplicial mode  \(\chi\) constructs the same Krylov space as the currents are quadratic in the original Krylov generators; more specifically, let  \(k, \ell \in \e{1, \ldots, d_K}\) with  \(\ell\) being odd, then we have
\begin{align}
    \e[\eta^{(k)}, J^{(\ell)}] &= i \sum_{j = 0}^{\ell} (-1)^j \e[\eta^{(k)}, \e[\eta^{(j)}, \eta^{(\ell - j)}]] \\
    &= i 2\sum_{j = 0}^{\ell} (-1)^j \e[\eta^{(k)}, \kryl{j}\kryl{\ell - j} - M^G_{j,\ell-j}] \\
    &= -4i \sum_{j = 0}^{\ell} (-1)^j \e(\kryl{j}M^G_{\ell-j,k} - \kryl{\ell-j}M^G_{j,k}) \\
    &= -8i \sum_{j = 0}^{\ell} (-1)^j \kryl{j}M^G_{\ell-j,k}\;.
\end{align}
Therefore, the new generators must also be path-supported, and they are even path-supported in the FFD model.

\cref{conj:krylov_path_support_free_fermions} is a statement about a fixed triple  \((H,G,\chi)\): given a fiducial operator  \(\chi\) in the extended graph Clifford algebra of  \(G\), the space  \(\mc{K}_{\chi}\) carries a Majorana representation of the adjoint action of  \(H\) if and only if the path-support condition of \cref{eq:krylov-path-quotient} holds.
The restriction on  \(\chi\) is not incidental.
The frustration graph, the anchored paths, and the path-support condition are all defined relative to a fixed presentation of  \(H\); conjugating a free-fermion model by a unitary that scrambles its terms produces an equivalent model whose natural fiducial operator need not lie in the extended graph Clifford algebra of the original presentation, and in such a rotated frame the path-support property need not be visible.
The conjecture makes no claim about fiducial operators of this kind.
For the same reason, the conjecture does not assert that every invariant Majorana space of  \(H\) arises as  \(\mc{K}_{\chi}\) for some admissible  \(\chi\): the existence of a quadratic invariant space does not, so far as we know, imply the existence of a generating fiducial operator within the extended graph Clifford algebra.
\cref{conj:krylov_path_support_free_fermions} is therefore not an exhaustive criterion by which every free-fermion solution of every model can be found.
It is instead an exact characterisation of free-fermion solvability for the Krylov spaces accessible from the graph Clifford algebra---a class that includes every graph-theoretic solution currently known---and its resolution would constitute a substantial step even within this scope.

The necessity direction is nevertheless substantially more difficult to establish than sufficiency, even within the scope just described. 
A free-fermion solution on  \(\mc{K}_{\chi}\) may be presented in a basis bearing no manifest relation to individual Hamiltonian terms or to paths in the frustration graph.
Proving necessity would therefore require reconstructing a graphical path-product representation from the mere existence of an invariant Majorana space.
Moreover, the reconstruction need not be unique: different products of Hamiltonian terms may represent the same operator, and cancellations may occur only after several Krylov contributions are combined.
A direct proof must consequently control both the combinatorics of the frustration graph and the ideal of model-dependent relations in the corresponding operator algebra.
Within this scope, necessity amounts to the following statement: every operator Krylov space generated by an admissible fiducial operator, on which the Hamiltonian adjoint action is quadratic, must possess a basis whose elements are represented by anchored induced-path products, modulo exact graphical relations.
At present we do not know how to exclude a more general mechanism by which non-path products could satisfy the canonical anticommutation relations without admitting such a reduction.
We likewise know of no example in which this occurs. 
The non-generic models considered here are particularly informative in this respect. 
Although their naive Krylov expansions initially appear to violate the proposed path-support criterion, the same algebraic relations responsible for their exceptional solvability remove the apparent obstruction.

\cref{conj:krylov_path_support_free_fermions}~therefore suggests a practical route towards discovering and certifying further non-generic free-fermion models.
Starting from a candidate fiducial operator in the extended graph Clifford algebra, one may generate the Krylov tower, detect the first non-path contributions, and search for exact relations that return these contributions to the anchored path span.
Once the finite-dimensional Krylov space has closed, it is sufficient to verify the reduction for a basis of that space rather than for infinitely many values of  \(k\). 
The conjecture predicts that success of this procedure is equivalent to the existence of a free-fermion solution on the resulting Krylov space.
In this sense, graphical cancellation relations do not merely simplify an already known solution: they expose the hidden induced-path structure through which quadratic fermionic closure is realised.

\section{A New Proof of Induced Paths of Krylov Modes}
\label[appendix]{app:induced-paths}

In this section we provide a new proof for the first key lemma of our paper. 
This lemma was previously proven in Ref.~\cite{chapman2023unified} as well as in Ref.~\cite{fukai2026solvingb}. 
We give a separate argument here for completeness; it uses less graph theory, is more intuitive, and is specialized to the chordal, claw-free setting needed in this work.
\KrylovModes*

\begin{definition}[``Chi'' modes] \label{def:chimode}
    Define
    \begin{align}
        \chi_k = \sum_{\ell = 0}^{k} (-1)^{\ell} Q^{(\ell)} \chi Q^{(k - \ell)}
    \end{align}
    for  \(k = 0, \ldots, \alpha(G)\), such that  \(\chi_0 = \chi\), and  \(\chi_1 = -[H, \chi]\).
\end{definition}
\begin{proposition}\label{prop:chispan}
    The operators  \(\chi_k\) are linear combinations of  \(h_P\) for  \(P \in \mc{P}^*\) as given in \cref{def:indpath}.
\end{proposition}
\begin{proof}
    Expanding the definition of  \(\chi_k\) in terms of the independent set charges, we write:
    \begin{align}
        \chi_k &= \sum_{\ell = \mathrm{max}(0, k - \alpha(G))}^{\mathrm{min}(k, \alpha(G))} (-1)^{\ell} Q^{(\ell)} \chi Q^{(k - \ell)} \label{eq:chikline1}\\
        &=\sum_{\ell = \mathrm{max}(0, k - \alpha(G))}^{\mathrm{min}(k, \alpha(G))} (-1)^{\ell} \sum_{\substack{S\in \mc{S}^{(\ell)} \label{eq:chikline2}\\
        S'\in \mc{S}^{(k-\ell)}}}h_S\chi h_{S'} \\
        \chi_k &=\chi\left(\sum_{\ell = \mathrm{max}(0, k - \alpha(G))}^{\mathrm{min}(k, \alpha(G))} \sum_{\substack{S\in \mc{S}^{(\ell)}\\
        S'\in \mc{S}^{(k-\ell)}}}(-1)^{\lvert S \rvert + |K\cap S|}h_S h_{S'}\right) \label{eq:chikline3}
    \end{align}
    From the second to the third line, we brought in the factor of  \((-1)^{\ell} = (-1)^{|S|}\), and we also commute  \(\chi\) through  \(h_S\) with a factor of  \((-1)^{|K \cap S|}\).
    We next use the fact that
    \begin{align}
        h_S h_{S'} = \left(\prod_{\bs{j} \in S \cap S'} b_{\bs{j}}^2\right) h_{S \setminus S'} h_{S' \setminus S}
        \label{eq:symdiffssp}
    \end{align}
    For  \(G\) a chordal, claw-free graph,  \(S \oplus S' = (S \setminus S') \cup (S' \setminus S)\) induces a set of disjoint paths.
    Furthermore, any path component  \(P \in \mc{P}\) of  \(G[S \oplus S']\) can share at most two vertices with  \(K\) since  \(K\) is a clique, and at most one path component of  \(G[S \oplus S']\) can share any vertices with  \(K\).
    If  \(P\) shares two vertices with  \(K\), then they must be neighbouring.
    Thus, one belongs to  \(S\setminus S'\), and the other belongs to  \(S' \setminus S\).
    Therefore, if  \(\lvert K \cap P \rvert\) is even (i.e.\ either 0 or 2), then  \(\lvert K \cap S \cap P \rvert = \lvert K \cap S' \cap P \rvert \in \{0, 1\}\), and if  \(\lvert K \cap P \rvert\) is odd, then  \(\lvert K \cap P \rvert = 1\).
    Finally, if  \(|K \cap P|\) is odd, then  \(\boldsymbol{j} = K \cap P\) must be an endpoint of  \(P\), or we will have a contradiction with the assumption that  \(K\) is simplicial. 
    That is, the two neighbours to  \(\boldsymbol{j}\) in  \(P\), which are not in  \(K\) by assumption, must belong to the same clique so are neighbouring to each other, but this contradicts the assumption that  \(P\) is an induced path. 
    
    In the following, let  \(P \in \mc{P}\) be a path component of  \(G[S \oplus S']\), and assume that  \(|K \cap P|\) is even. 
    For such a path, we define a unique pair of corresponding independent sets 
    \begin{align}
        \begin{cases}
            \widetilde{S} \coloneqq (S \setminus P) \cup (S' \cap P) & \\
            \widetilde{S}' \coloneqq (S' \setminus P) \cup (S \cap P)
        \end{cases}.
    \end{align}
    By the above,  \(|K \cap S \cap P| = |K \cap S' \cap P|\). 
    Since  \(\widetilde{S}\) differs from  \(S\) only inside  \(P\), this gives  \(|K \cap \widetilde{S}| = |K \cap S|\).
    Suppose that  \(P\) has odd length.
    In this case, we have  \(\lvert\widetilde{S}\rvert = \lvert S \rvert = \ell\) and  \(|\widetilde{S}'| = |S'| = k - \ell\) because  \(\lvert P \cap S \rvert = \lvert P \cap S'\rvert\).   
    This gives
    \begin{align}
        (-1)^{\lvert S \rvert + \lvert K \cap S \rvert} h_S h_{S'} + (-1)^{\lvert \widetilde{S} \rvert + \lvert K \cap \widetilde{S} \rvert} h_{\widetilde{S}} h_{\widetilde{S}'} &= (-1)^{\ell + \lvert K \cap S \rvert} \left(\prod_{\bs{j} \in S \cap S'} b_{\bs{j}}^2 \right) h_{S \setminus P} (h_{S \cap P} h_{S' \cap P} + h_{S' \cap P} h_{S \cap P}) h_{S' \setminus P} \\
        &= 0
    \end{align}
    where we arrive at the last line because  \(h_{S \cap P} h_{S' \cap P} = -h_{S' \cap P} h_{S \cap P}\), due to  \(P\) having odd length.
    Thus, if  \(G[S \oplus S']\) has a path component  \(P\) of odd length, then  \(P\) must have an endpoint in  \(K\).
    
    If  \(P\) has even length, then suppose that both endpoints (possibly the same vertex) of  \(P\) are in  \(S\).
    In this case, we have  \(\lvert \widetilde{S} \rvert = \lvert S \rvert - 1 = \ell - 1\), and  \(\lvert \widetilde{S}' \rvert = \lvert S' \rvert + 1 = k - \ell + 1\).
    Now we have 
    \begin{align}
        (-1)^{\lvert S \rvert + \lvert K \cap S \rvert} h_S h_{S'} + (-1)^{\lvert \widetilde{S} \rvert + \lvert K \cap \widetilde{S} \rvert} h_{\widetilde{S}} h_{\widetilde{S}'} &= (-1)^{\ell + \lvert K \cap S \rvert} \left(\prod_{\bs{j} \in S \cap S'} b_{\bs{j}}^2 \right) h_{S \setminus P} (h_{S \cap P} h_{S' \cap P} - h_{S' \cap P} h_{S \cap P}) h_{S' \setminus P} \\
        &= 0
    \end{align}
    where the last line follows because  \(h_{S \cap P} h_{S' \cap P} = h_{S' \cap P} h_{S \cap P}\), due to  \(P\) having even length.
    A corresponding argument follows if  \(k - \ell \geq 1\) and both endpoints of  \(P\) belong to  \(S'\).
    
    Thus, all path components of  \(G[S \oplus S']\) must have an endpoint in  \(K\), which further requires that  \(G[S \oplus S'] = P\), a single connected path in  \(\mc{P}\).
    Since  \(\chi_k\) includes a factor of  \(\chi\), corresponding to the vertex  \(\bs{j}^*\) neighbouring the unique vertex in  \(K\), we have that  \(\chi_k\) is in-fact a sum of induced paths in  \(\mc{P}^*\), as required.
\end{proof}

\begin{definition}[Fiducial representation]
    Given a general frustration graph  \(G \in (V, E)\), choose an arbitrary orientation  \(\tau\) of  \(G\) and associated directed graph  \(\vec{G} = (V, D)\).  
    Define the fiducial representation of  \(G\) with respect to  \(\tau\) as the Hamiltonian  \(H = \sum_{\bs{j} \in V} \sigma^{\bs{j}}\) on a system of  \(|E|\) qubits such that each qubit is associated to an edge of  \(G\), with
    \begin{align}
        \sigma^{\bs{j}} = \left(\bigotimes_{e \coloneqq (\bs{j} \rightarrow \bs{k}) \in D} X_e \right) \otimes \left(\bigotimes_{e \coloneqq (\bs{j} \leftarrow \bs{k}) \in D} Z_e \right).
    \end{align}
    Namely, each term acts as  \(X\) on each outgoing edge qubit and  \(Z\) on each incoming edge qubit.
\end{definition}
\begin{proposition}\label{prop:fiducial} 
    Let  \(\sigma^S \coloneqq \prod_{\bs{j} \in S} \sigma^{\bs{j}}\) for  \(S \subseteq V\) in the fiducial representation. 
    If  \(\sigma^S \propto I\), then  \(S = \varnothing\). 
\end{proposition}
\begin{proof}
    Let  \(S \subseteq V\) such that  \(\sigma^S \propto I\). 
    Now assume that  \(S \neq \varnothing\). 
    Then there exists a  \(\bs{v} \in S\) and since  \(G\) is connected  \(N(\bs{v})\) is not empty. 
    Without loss of generality let  \(\bs{v}\) have an outgoing edge, labeled  \(e \in D\). 
    Then it holds  \(\sigma^{\bs{v}}\) acts as \(X\) on \(e\). 
    The only other possible factor acting on \(e\) is \(\sigma^{\bs{w}}\), where \(\bs{w}\) is the other endpoint of \(e\) (i.e. \(e = (\bs{v} \to \bs{w}) \)), where it acts as \(Z\). 
    Thus, regardless of whether \(\bs{w} \in S\), \(\sigma^S\) does not act as the identity on \(e\), a contradiction.
\end{proof}

\begin{proof}[Proof of Lemma 1]
    The statement is clearly true for  \(k \in \{0, 1\}\). 
    We first show that  \(\kryl{k}\) is a sum over induced-path operators. 
    By \cref{prop:zero_mode}, suppose
    \begin{align}
        \kryl{1} \coloneqq \mc{L}[\chi] = \sum_{j = 1}^{\alpha(G)} \left(v_j \psi_j^{\dagger} + v^*_j \psi_j \right)
    \end{align}
    for a  \(\mbf{v} \coloneqq \{v_j \}_{j = 1}^{\alpha(G)} \in \mathds{C}^{\times \alpha(G)}\), and thus
    \begin{align}
        \kryl{k} \coloneqq \mc{L}^{k - 1}[\kryl{1}] = i^{k - 1} \sum_{j = 1}^{\alpha(G)} \left((-\varepsilon_j)^{k - 1} v_j \psi_j^{\dagger} + \varepsilon_j^{k - 1} v^*_j \psi_j \right)
    \end{align}
    by \cref{eq:eigenmodes}. 
    By \cref{def:incognito,def:chimode,prop:chispan}, we thus have that  \(\kryl{k}\) is a sum of  \(h_P\) with  \(P \in \mc{P}^*\) for all  \(k \geq 0\).
    
    We show the remaining parts of the lemma by induction. 
    Namely, assume \cref{eq:pathsupport} is true for all  \(1 \leq j < k\) and \cref{eq:deformeq} is true for all  \(1 \leq j < k - 1\). 
    We have
    \begin{align}
        \kryl{k} &= \mc{L}[\kryl{k-1}] \\
        &= \frac{i}{2}(-i)^{k-1} \sum_{P \in \mc{P}^*} \left(\mbf{A}^{k - 1}\right)_{P, \{\bs{j}^*\}} [H, h_P] \\
        \kryl{k} &= -\frac{(-i)^k}{2} \sum_{P \in \mc{P}^*} \left(\mbf{A}^{k - 1}\right)_{P, \{\bs{j}^*\}} \sum_{\bs{j} \in V} [h_{\bs{j}}, h_P] 
    \end{align}
    Here  \(h_P = \chi\, h_{\Pv{1}}\cdots h_{\Pv{|P|-1}}\) is the ordered product of \cref{eq:hpdef}, with the factors in path order starting from  \(\chi\).
    For  \(G\) a chordal, claw-free graph, there are only three ways that a term  \(h_{\bs{j}}\) in the  \(H\) can anticommute with  \(h_P\): 
    \begin{enumerate}[label=(\roman*)]
        \item  \(\bs{j} \in P\) and  \(P \setminus \{\bs{j}\} \in \mc{P}^*\)
        \item  \(\bs{j} \notin P\) and  \(P \cup \{\bs{j}\} \in \mc{P}^*\), and
        \item  \(\bs{j} \notin P\) and there exists a clone  \(\bs{k} \in P\) to  \(\bs{j}\) (denoted  \(\bs{j} \prec_P \bs{k}\)) such that  \(P \setminus \{\bs{k}\} \cup \{\bs{j}\} \in \mc{P}^*\).
    \end{enumerate}
    In each of these cases  \(h_{\bs{j}}\) anticommutes with  \(h_P\), so  \([h_{\bs{j}}, h_P] = 2h_{\bs{j}}h_P\).
    In case (i),  \(\bs{j} = \vv{P}\) is the free endpoint,  \(h_P = h_{P \setminus \{\bs{j}\}} h_{\bs{j}}\), and  \(h_{\bs{j}}\) anticommutes with exactly one factor of  \(h_{P \setminus \{\bs{j}\}}\) (its predecessor on the path, possibly  \(\chi\)). 
    In case (ii),  \(h_{\bs{j}}\) anticommutes with exactly one factor of  \(h_P\), namely  \(h_{\vv{P}}\), and  \(h_{P \cup \{\bs{j}\}} = h_P h_{\bs{j}}\). 
    Hence
    \begin{align}
        h_{\bs{j}} h_P = -b_{\bs{j}}^2\, h_{P \setminus \{\bs{j}\}} \quad \text{in case (i)}, \qquad h_{\bs{j}} h_P = -h_{P \cup \{\bs{j}\}} \quad \text{in case (ii)}.
    \end{align}
    This gives
    \begin{align}
        \kryl{k} &= -(-i)^k \sum_{P \in \mc{P}^*} \left(\mbf{A}^{k - 1}\right)_{P, \{\bs{j}^*\}} \left(\sum_{\bs{j} \in P, P \setminus \{\bs{j}\} \in \mc{P}^*} + \sum_{\bs{j} \notin P, P \cup \{\bs{j}\} \in \mc{P}^*} + \sum_{\bs{j} \notin P, \bs{j} \prec_P \bs{k}} \right) h_{\bs{j}} h_P \\
        \kryl{k} &= (-i)^k \sum_{P \in \mc{P}^*} \left(\mbf{A}^{k - 1}\right)_{P, \{\bs{j}^*\}} \left(\sum_{\bs{j} \in P, P \setminus \{\bs{j}\} \in \mc{P}^*} b_{\bs{j}}^2 h_{P \setminus \{\bs{j}\}} + \sum_{\bs{j} \notin P, P \cup \{\bs{j}\} \in \mc{P}^*} h_{P \cup \{\bs{j}\}} - \sum_{\bs{j} \notin P, \bs{j} \prec_P \bs{k}} h_{\bs{j}} h_P \right).
    \end{align}
    Regarding the first and second sums inside the parentheses (cases i and ii), we have  \(\mbf{A}_{P \setminus \{\bs{j}\}, P}~=~b_{\bs{j}}^2\) and  \(\mbf{A}_{P \cup \{\bs{j}\}, P}~=~1\), with all other amplitudes in  \(\mbf{A}\) zero. 
    This gives
    \begin{align}
        \kryl{k} &= (-i)^k \sum_{P \in \mc{P}^*} \left[\left(\mbf{A}^{k}\right)_{P, \{\bs{j}^*\}} h_P - \sum_{\bs{j} \notin P, \bs{j} \prec_P \bs{k}} \left(\mbf{A}^{ k - 1}\right)_{P, \{\bs{j}^*\}} h_{\bs{j}} h_P \right].
    \end{align}
    We know that the right-hand-side of the equation above must be a sum over induced-path operators. 
    Since  \(P \cup \{\bs{j}\} \notin \mc{P}^*\) for every term inside the sum \emph{inside} square brackets above (case iii), we have that this sum, when summed over  \(P \in \mc{P}^*\), must vanish. 
    This proves the inductive step for \cref{eq:pathsupport}, and we have
    \begin{align}
        \sum_{P \in \mc{P}^*} \sum_{\bs{j} \notin P, \bs{j} \prec_P \bs{k}} \left(\mbf{A}^{ k - 1}\right)_{P, \{\bs{j}^*\}} h_{\bs{j}} h_P = 0.
    \end{align}
    For each term indexed by  \((P, \bs{j})\) above, there is a unique term indexed by  \((\widetilde{P}, \bs{k})\) with  \(\widetilde{P}\) the deformation of  \(P\) by  \(\bs{j}\). 
    Moreover
    \begin{align}
        h_{\bs{j}} h_{P} = -h_{\bs{k}} h_{\widetilde{P}}.
    \end{align}  
    Thus, we pair uniquely corresponding terms in the sum to obtain
    \begin{align}
        0 &= \frac{1}{2}\sum_{\substack{(P, \bs{j}), P \in \mc{P}^* \\ \bs{j} \notin P, \bs{j} \prec_P \bs{k}}} \left[\left(\mbf{A}^{ k - 1}\right)_{P, \{\bs{j}^*\}} h_{\bs{j}} h_P + \left(\mbf{A}^{ k - 1}\right)_{P \setminus \{\bs{k}\} \cup \{\bs{j}\}, \{\bs{j}^*\}} h_{\bs{k}} h_{P \setminus \{\bs{k}\} \cup \{\bs{j}\}} \right] \\
        0 &= \sum_{\substack{(P, \bs{j}), P \in \mc{P}^* \\ \bs{j} \notin P, \bs{j} \prec_P \bs{k}}} \left[\left(\mbf{A}^{ k - 1}\right)_{P, \{\bs{j}^*\}} - \left(\mbf{A}^{ k - 1}\right)_{P \setminus \{\bs{k}\} \cup \{\bs{j}\}, \{\bs{j}^*\}} \right] h_{\bs{j}} h_P.
    \end{align}
    There is no other term in the sum, indexed by  \((P', \bs{j}')\), with  \(\bs{j}' \notin P'\) such that  \(\bs{j}' \prec_{P'} \bs{k}'\) for  \(\bs{k}' \in P'\) with  \(P' \cup \{\bs{j}'\} = P \cup \{\bs{j}\} = \widetilde{P} \cup \{\bs{k}\}\). 
    Choosing the fiducial representation for the Hamiltonian, we thus have that the term in square brackets must vanish for each  \((P, \bs{j})\) in the sum by \cref{prop:fiducial}. 
    This proves the inductive step for \cref{eq:deformeq}.
\end{proof}

\section{Relation Between the Cayley Transform of \texorpdfstring{ \(\mathrm{ad}_H\)}{adH} and the Transfer Operator}
\label[appendix]{app:cayley}

This appendix makes the relation between the Cayley transform and the transfer operator explicit.
It is not strictly required to prove the path support of the Krylov currents, because the relevant expansion can be obtained from Ref.~\cite{fukai2026solvingb} without it (see \cref{prop:zero_mode} and \cref{prop:TchiT}).
We include it because the resulting picture of the Krylov flow as a Cayley-parametrized unitary conjugation may be of independent interest.

From \cite[Equation (3.5)]{fendley2019free}, we have that for the FFD model, the following product decomposition of the transfer operator holds:
\begin{proposition}[Diagonal form of the transfer operator]\label{prop:diagonal_transfer}
    For the transfer operator of the FFD Hamiltonian, as given in \cref{def:transfer}, we have
    \begin{align}
        T(u) = \prod_{j = 1}^{\alpha(G)} \left(I - u \varepsilon_j [\psi_j, \psi_j^{\dagger}] \right)\;,
        \label{eq:diagonaltransfer}
    \end{align}
    for  \(u \in U_0\) (region around  \(0\)). 
\end{proposition}
Furthermore, we make use of the following well-known identity:
\begin{proposition}[Generating function for Newton's identities]
    Given the elementary symmetric polynomials  \(\{e_k(\{\bs{x}\})\}_{k = 0}^d\) and power sum symmetric polynomials  \(\{p_k(\{\bs{x}\})\}_k\) in the variables  \(\{x_j\}_{j = 1}^d\), defined by
    \begin{align}
        &e_0(\{\bs{x}\}) = 1 \qquad \text{and} \qquad e_k(\{\bs{x}\}) = \sum_{1 \leq i_1 < \dots < i_k \leq d} \left(\prod_{j = 1}^k x_{i_j}\right) \quad \text{for  \(k > 0\)}\\
        \text{and} \qquad &p_k(\{\bs{x}\}) = \sum_{j = 1}^d x_j^k
    \end{align}
    we have
    \begin{align}
        \sum_{k = 0}^d e_k(\{\bs{x}\}) (-u)^k = \mathrm{exp}\left[-\sum_{k = 1}^{\infty} \frac{p_k(\{\bs{x}\})}{k} u^k\right].
    \end{align}
    \label{prop:newtongenerating}
\end{proposition}

\CayleyTransform*
\begin{proof}
    With respect to the Hilbert-Schmidt inner product, the incognito modes are orthogonal and with that linearly independent. 
    In the  \(2\alpha(G)\)-dimensional subspace spanned by these modes, the adjoint operator is diagonal via
    \begin{equation}
        \mathrm{ad}_H \psi_j \mapsto 2\varepsilon_j \psi_j \qquad \text{and} \qquad \mathrm{ad}_H \psi_j^\dagger \mapsto -2\varepsilon_j \psi_j^\dagger\;.
    \end{equation}
    This spectral decomposition allows us to define the inverse of  \(\left(I - \frac{iu}{2} \mathrm{ad}_H \right)\), more specifically we have
    \begin{equation}
        \left(I - \frac{iu}{2} \mathrm{ad}_H \right)^{-1} \psi_j \mapsto \left(1 - iu\varepsilon_j\right)^{-1} \psi_j \qquad \text{and} \qquad \left(I - \frac{iu}{2} \mathrm{ad}_H \right)^{-1} \psi_j^\dagger \mapsto \left(1 + iu \varepsilon_j\right)^{-1} \psi_j^\dagger\;.
    \end{equation}
    Similar arguments obviously hold for  \([\psi_j, \psi_j^{\dagger}]\),  \(j = 1, \ldots, \alpha(G)\), which we implicitly use below when taking the logarithm and using Newton's identity (\cref{prop:newtongenerating}) on related terms (to effectively deal with them as scalars).
    
    Now let
    \begin{align}
        X = c_0 \phi_0 + \sum_{j = 1}^{\alpha(G)} \left(x_j \psi_j + y_j \psi_j^{\dagger}\right)
        \label{eq:incognitospan}
    \end{align}
    for appropriate complex \(c_0\),  \((x_j)_j\), and  \((y_j)_j\). 
    As the operator  \(V(u)\) is unitary, we can thus express it as the exponential of an anti-Hermitian operator as  \(V(u) \coloneqq \mathrm{exp}\left[i \mc{H}(u)\right]\), where  \(\mc{H}(u)\) is a Hermitian operator given by
    \begin{align}
        -i \mathrm{log}\left[V(u)\right] = \frac{i}{2} \mathrm{log}[P_G(u^2)]I - i \mathrm{log}[T(iu)].
    \end{align}
    Expanding the transfer operator using \cref{eq:diagonaltransfer} gives
    \begin{align}
        \mathrm{log}[T(iu)] &= \sum_{j = 1}^{\alpha(G)} \mathrm{log}\left(I - i u \varepsilon_j [\psi_j, \psi_j^{\dagger}]\right).
    \end{align}
    By \cref{prop:newtongenerating}, together with \([\psi_j, \psi_j^{\dagger}]^{2} = I\) by the canonical anticommutation relations, we have
    \begin{align}
        \mathrm{log}[T(iu)] &= -\sum_{j = 1}^{\alpha(G)} \sum_{k = 1}^{\infty} \frac{(iu \varepsilon_j)^k}{k}[\psi_j, \psi_j^{\dagger}]^{k} \\
        &= -\sum_{j = 1}^{\alpha(G)} \left\{\left(\frac{1}{2}\sum_{k = 1}^{\infty} \frac{(-u^2\varepsilon_j^2)^k}{k}\right)I + i \left(\sum_{k = 0}^{\infty} \frac{(-1)^k (u\varepsilon_j)^{2k + 1}}{2k + 1}\right)[\psi_j, \psi_j^{\dagger}]\right\} \\
        \mathrm{log}[T(iu)] &= \frac{1}{2}\mathrm{log}[P_G(u^2)]I - i \sum_{j = 1}^{\alpha(G)} \mathrm{tan}^{-1}\left(u\varepsilon_j\right)[\psi_j, \psi_j^{\dagger}]
    \end{align}
    and thus
    \begin{align}
        -i \mathrm{log}\left[V(u)\right] = - \sum_{j = 1}^{\alpha(G)} \mathrm{tan}^{-1}\left(u\varepsilon_j\right)[\psi_j, \psi_j^{\dagger}].
    \end{align}
    Next, we expand the left-hand side of \cref{eq:adjointunitary} using \cref{eq:diaghamiltonian} and the ansatz in \cref{eq:incognitospan}
    \begin{align}
        &\mathcal{C}(u)[X] = c_0 \phi_0 + \sum_{j = 1}^{\alpha(G)} \left\{x_j \left(\frac{1 + i u \varepsilon_j}{1 - i u \varepsilon_j} \right) \psi_j + y_j \left(\frac{1 - i u \varepsilon_j}{1 + i u \varepsilon_j} \right) \psi_j^{\dagger}\right\} \\
        &\mh{10mm}= c_0 \phi_0 + \sum_{j = 1}^{\alpha(G)} \left\{x_j e^{2i \mathrm{tan}^{-1}(u\varepsilon_j)} \psi_j + y_j e^{-2i \mathrm{tan}^{-1}(u\varepsilon_j)} \psi_j^{\dagger}\right\} \\
        &\mh{10mm}= \mathrm{exp}\left[i \sum_{j = 1}^{\alpha(G)} \mathrm{tan}^{-1}(u\varepsilon_j) [\psi_j, \psi_j^{\dagger}]\right]\left[c_0 \phi_0 + \sum_{j = 1}^{\alpha(G)} \left(x_j \psi_j + y_j \psi_j^{\dagger}\right)\right] \mathrm{exp}\left[-i \sum_{j = 1}^{\alpha(G)} \mathrm{tan}^{-1}(u\varepsilon_j) [\psi_j, \psi_j^{\dagger}]\right] \\
        &\left(I + \frac{iu}{2} \mathrm{ad}_H \right) \left(I - \frac{iu}{2} \mathrm{ad}_H \right)^{-1}[X] = V(u)^{\dagger} X V(u)
    \end{align}
    where  \(V(u) \coloneqq P_G^{-1/2}(u^2) T(iu)\), as desired.
\end{proof}

\section{Claw cancellation structures}
\label[appendix]{app:structural-graph}

In this appendix, we formalize a sufficient condition under which claws in the frustration graph do not obstruct the commutation of the first and third independent-set charges. 
The mechanism is the one used in the models of Refs.~\cite{fukai2025free,Fendley2024}: claws are not removed from the frustration graph, but their contributions to the charge commutator cancel pairwise after restricting to an appropriate symmetry sector.

Let
\begin{equation}
    H=\sum_{\bsj\in V}h_{\bsj}
\end{equation}
be a Pauli Hamiltonian with frustration graph  \(G=(V,E)\), and assume that all nonzero Hamiltonian coefficients are equal. 
We absorb this common coefficient into the overall energy scale, so that each  \(h_{\bsj}\) below denotes the corresponding normalized Pauli operator. 
Define the independent-set charges
\begin{equation}
    Q^{(r)} = \sum_{S\in\mc{S}^{(r)}}h_S, \qquad h_S\coloneqq\prod_{\bsj\in S}h_{\bsj},
\end{equation}
where products are taken according to a fixed ordering of  \(V\).

We write  \(\claws(G)\) for the set of induced claws of  \(G\). 
If
\[
\claw=(\bs{c};\bs{a},\bs{b},\bs{d})
\]
is a claw with centre  \(\bs{c}\) and leaves  \(\{\bs{a},\bs{b},\bs{d}\}\), we write
\begin{equation}
    V(\claw) = \{\bs{c},\bs{a},\bs{b},\bs{d}\}, \qquad h_{\claw} \coloneqq h_{\bs{c}}h_{\bs{a}}h_{\bs{b}}h_{\bs{d}}.
\end{equation}
Here the centre is placed first; this fixes the sign, since  \(h_{\bs{c}}\) anticommutes with each of the three mutually commuting leaves, so that reordering the factors can change  \(h_{\claw}\) by a sign.

We first record the elementary parity condition that produces symmetries. 
For a subset  \(U\subseteq V\), define
\begin{equation}
    h_U=\prod_{\bsj\in U}h_{\bsj}.
\end{equation}
Then, for any vertex  \(\bsk\in V\),
\begin{equation}
    h_{\bsk}h_U = (-1)^{\dgrU{U}{\bsk}}h_Uh_{\bsk}.
    \label{eq:appendix_degree_parity_commutation}
\end{equation}
Hence, if
\begin{equation}
    \dgrU{U}{\bsk}=0\mod 2 \qquad \text{for all }\bsk\in V,
    \label{eq:appendix_even_degree_into_U}
\end{equation}
then  \(h_U\) commutes with every term in  \(H\), and therefore with  \(H\) itself. 
Equivalently, if  \(\mbf{A}_G\) is the binary adjacency matrix of  \(G\), then
\begin{equation}
    \mbf{A}_G\subv{U}=0 \qquad \text{over }\mathbb{F}_2.
\end{equation}
In particular,  \(G[U]\) is Eulerian if it is connected.

\begin{definition}[Admissible claw pairing]
    \label{def:admissible_claw_pairing}
    An admissible claw pairing of  \(G\) is a fixed-point-free involution
    \begin{equation}
        \pi:\claws(G)\to\claws(G)
    \end{equation}
    such that, for every claw  \(\claw\in\claws(G)\), the symmetric difference
    \begin{equation}
        U_{\claw} \coloneqq V(\claw)\sd V(\pi(\claw))
    \end{equation}
    satisfies
    \begin{equation}
        \dgrU{U_{\claw}}{\bsj}=0\mod 2 \qquad \text{for all }\bsj\in V.
        \label{eq:admissible_claw_pairing_even_degree}
    \end{equation}
    Equivalently,  \(h_{U_{\claw}}\) is a symmetry of  \(H\).
\end{definition}

\begin{lemma}[Symmetric cancellation of claws]
    \label{lem:symmetric_claw_cancellation}
    Let  \(H\) be an equal-coefficient Pauli Hamiltonian with frustration graph  \(G\), and suppose that  \(G\) admits an admissible claw pairing  \(\pi\).
    Let  \(\mathcal{H}_{\lambda}\) be an  \(H\)-invariant symmetry sector such that, for every paired pair of claws  \(\claw\) and  \(\pi(\claw)\),
    \begin{equation}
        h_{\claw}\big|_{\mathcal{H}_{\lambda}} = -h_{\pi(\claw)}\big|_{\mathcal{H}_{\lambda}}.
        \label{eq:paired_claw_cancellation_condition}
    \end{equation}
    Then
    \begin{equation}
        [Q^{(1)},Q^{(3)}]\big|_{\mathcal{H}_{\lambda}}=0.
        \label{eq:q1q3_commute_in_sector}
    \end{equation}
\end{lemma}

\begin{proof}
    Expanding the commutator gives
    \begin{equation}
        [Q^{(1)},Q^{(3)}] = \sum_{\substack{\bsj\in V\\S\in\mc{S}^{(3)}}}
        [h_{\bsj},h_S].
    \end{equation}
    Terms for which  \(\bsj\) has an even number of neighbours in  \(S\) commute.
    When  \(\bsj\) has exactly one neighbour in  \(S\), the corresponding contributions cancel under the usual vertex-shifting involution that exchanges  \(\bsj\) with its unique neighbour in  \(S\). 
    The only unpaired contributions therefore occur when  \(\bsj\) is adjacent to all three vertices of  \(S\). 
    Since  \(S\) is independent, these four vertices induce a claw, and the corresponding term is  \([h_{\bs{j}},h_S]=2h_{\bs{j}}h_S=2h_{\claw}\) with the centre-first ordering of  \(h_{\claw}\). 
    Consequently,
    \begin{equation}
        [Q^{(1)},Q^{(3)}] = 2\sum_{\claw\in\claws(G)}h_{\claw}.
        \label{eq:q1q3_claw_obstruction}
    \end{equation}
    
    The admissible pairing partitions the claws into pairs  \(\{\claw,\pi(\claw)\}\). 
    By \cref{eq:paired_claw_cancellation_condition}, each pair satisfies
    \begin{equation}
        \left( h_{\claw} + h_{\pi(\claw)} \right) \big|_{\mathcal{H}_{\lambda}} = 0.
    \end{equation}
    The sum in \cref{eq:q1q3_claw_obstruction} therefore vanishes after restriction to  \(\mathcal{H}_{\lambda}\), which proves \cref{eq:q1q3_commute_in_sector}.
\end{proof}

The preceding lemma establishes only the cancellation required for the lowest nontrivial commutator  \([Q^{(1)},Q^{(3)}]\). 
It appears likely that the same mechanism can be extended to commutators  \([Q^{(r)},Q^{(s)}]\) for arbitrary  \(r\) and  \(s\), following the construction of \citet{fukai2025free}. 
In such an extension, the vertices outside the local claw would supply common independent-set factors, while the paired local claw contributions would cancel in the relevant symmetry sector. 
A complete proof would require showing that the claw-pairing involution is compatible with all such external independent-set extensions and with possible overlaps between the resulting configurations. 
Since this generalization is not needed for the results of the present paper, we do not pursue that additional combinatorial bookkeeping here.

\end{document}